\documentclass[11pt]{article}%
\usepackage{amsmath}
\usepackage{geometry}
\usepackage{amsfonts}
\usepackage{amssymb}
\usepackage{graphicx}
\usepackage{nopageno}%
\providecommand{\U}[1]{\protect\rule{.1in}{.1in}}
\newtheorem{theorem}{Theorem}
\pdfoutput=1

\newtheorem{corollary}{Corollary}

\newtheorem{lemma}{Lemma}

\newtheorem{proposition}{Proposition}
\newtheorem{remark}{Remark}

\newenvironment{proof}[1][Proof]{\noindent\textbf{#1.} }{\ \rule{0.5em}{0.5em}}
\begin{document}

\title{Demand and Welfare Analysis in Discrete Choice Models with Social
Interactions\thanks{We are grateful to three anonymous referees, the editor,
and Steven Durlauf, James Heckman, Xenia Matschke, GautamTripathi, and seminar
participants at the University of Chicago and the University of Luxembourg for
helpful feedback. Bhattacharya acknowledges financial support from the ERC
consolidator grant EDWEL; the first outline of this project appeared as part
b.3 of that research proposal of March 2015. Part of this research was
conducted while Kanaya was visiting the Kyoto Institute of Economic Research,
Kyoto University (under the Joint Research Program of the KIER), the support
and hospitality of which are gratefully acknowledged.}}
\author{DEBOPAM BHATTACHARYA\thanks{Address for correspondence: Faculty of Economics,
University of Cambridge, CB3 9DD. email: debobhatta@gmail.com}\\\textit{University of Cambridge}
\and PASCALINE DUPAS\\\textit{Stanford University}
\and SHIN KANAYA\\\textit{University of Essex and Kyoto University}}
\date{December 18, 2022}
\maketitle

\begin{abstract}
Many real-life settings of individual choice involve social interactions,
causing targeted policies to have spillover effects. This paper develops novel
empirical tools for analyzing demand and welfare effects of policy
interventions in binary choice settings with social interactions. Examples
include subsidies for health product adoption and vouchers for attending a
high-achieving school. We show that even with fully parametric specifications
and unique equilibrium, choice data, that are sufficient for counterfactual
\emph{demand} prediction under interactions, are \emph{insufficient} for
welfare calculations. This is because distinct underlying mechanisms producing
the same interaction coefficient can imply different welfare effects and
deadweight-loss from a policy intervention. Standard index restrictions imply
distribution-free bounds on welfare. We propose ways to identify and
consistently estimate the structural parameters and welfare bounds allowing
for unobserved group effects that are potentially correlated with observables
and are possibly unbounded. We illustrate our results using experimental data
on mosquito-net adoption in rural Kenya.

\end{abstract}

\section{INTRODUCTION}

Social interaction models -- where an individual's payoff from an action
depends on aggregate choice -- feature prominently in economic and
sociological research. In this paper, we address a substantively important
issue that has received limited attention within these literatures, viz., how
to conduct welfare analysis of economic policy intervention in such settings.
Examples include subsidies for adopting a health product and merit-based
vouchers for attending a high-achieving school, where the welfare gain of
beneficiaries may be accompanied by spillover-led welfare effects on those
unable to adopt or move, respectively. Ex-ante welfare analysis of policies is
ubiquitous in economic applications, and informs the practical decision of
whether to implement the policy in question. Furthermore, common public
interventions such as taxes and subsidies are often motivated by efficiency
losses resulting from externalities. Therefore, it is important to develop
empirical methods for welfare analysis in presence of such externalities,
which cannot be done using available tools in the literature. Developing such
methods and making them practically relevant also requires one to clarify and
extend some aspects of existing empirical models of social interaction.

\noindent\textbf{Literature Review and Contributions}: Seminal contributions
to the econometrics of social interactions include Manski (1993) for
continuous outcomes, and Brock and Durlauf (2001) (henceforth, BD01) for
binary outcomes. Bisin, Moro and Topa (2011) discuss some issues related to
identification and estimation of structural parameters in choice models with
social interaction and multiple equilibria but do not cover welfare analysis.
More recently, there has been a surge of research on the related theme of
network models, cf. de Paula (2017) who provides a comprehensive review of the
relevant literature. On the other hand, the econometric analysis of welfare in
standard discrete choice settings, i.e., with heterogeneous consumers but
without social spillover, started with Domencich and McFadden (1977), with
later contributions by Daly and Zachary (1978), Small and Rosen (1981), and
Bhattacharya (2018). The present paper builds on these two separate
literatures to examine how social interactions influence welfare effects of
policy interventions and the identifiability of such welfare effects from
standard choice data. In the context of a logit binary choice model with
social interactions, BD01 Sec 3.3 equations (16) and (17) discussed how to
infer the sign of the differences between expected ex ante (indirect) utility
at each possible equilibria resulting from the policy intervention being
studied. This differs from the average of individual compensating variations
that restore realized individual utilities to their pre-intervention level
which is a money metric, unlike the BD01 measure, and hence can be directly
compared with the cost of the intervention, yielding a theoretically justified
measure of deadweight loss. Consequently, this measure has received the most
attention in the recent literature on applied welfare analysis, cf. Hausman
and Newey (2016), Bhattacharya (2015), McFadden and Train (2019). However, in
settings involving spillover, we cannot use the methods of the above papers,
as they do not allow for individual utilities to be affected by aggregate
choices -- a feature that has fundamental implications for welfare analysis.
Therefore, new methods are required for welfare calculations under spillover,
which we develop in the present paper.

Our starting point is a theoretically coherent empirical model where many
individuals with some observed and some unobserved attributes interact with
each other to produce the aggregate choice in equilibrium before and after the
policy intervention. Individual choice data can be used to estimate
identifiable parameters of this model in BD01, which can then be used to
predict counterfactual \emph{demand}, i.e. equilibrium choice probabilities
resulting from a hypothetical price intervention, e.g., a price subsidy.
However, we show that unlike counterfactual \emph{demand} estimation,
\emph{welfare} effects are generically not identified from standard choice
data under interactions, even when utilities and the distribution of
unobserved heterogeneity are parametrically specified, equilibrium is unique,
and there are no endogeneity concerns. To understand the heuristics behind
under-identification, consider the empirical example of evaluating the welfare
effect of subsidizing an anti-malarial, insecticide-treated mosquito net.
Suppose, under suitable restrictions, we can model choice behavior in this
setting via a Brock-Durlauf type social interaction model, and the data can
identify the coefficient on the social interaction term. However, this
coefficient may reflect an aggregate effect of several distinct mechanisms,
viz. (a) a social preference for conforming, (b) learning from others'
experiences, (c) a health-concern led desire to protect oneself from
mosquitoes deflected from neighbors who adopt a bednet, and (d) desire to
free-ride on other users who increase herd-immunity by protecting themselves
and/or protect neighbors via the insecticide effects. These distinct
mechanisms, with different magnitudes in general, can make the social
interaction coefficient positive, but are not separately identifiable from
choice data (only their sum is). But they have different implications for
welfare if, say, a subsidy is introduced. In particular, if spillovers are all
due to preference for social conforming or learning and there is no
(perceived) health externality, then as more neighbors buy, a household's
perceived utility from buying will increase over and above the gain due to
price reduction. At the other extreme, if spillovers are solely due to
perceived negative health externality of buyers on non-buyers, then increased
purchase by neighbors would lower the utility of a household upon \emph{not}
buying via the health-route, but not affect it upon buying since the household
is then protected anyway. These different aggregate welfare effects are both
consistent with the same positive aggregate social interaction coefficient.
This conclusion continues to hold even if eligibility for the subsidy is
universal and there are no income effects or endogeneity concerns.

This feature is present in many other choice situations that economists
routinely study. For example, merit-based school vouchers for attending a
high-achieving school can potentially have a range of possible welfare
effects. Aggregate welfare change could be negative if, for example, with
high-ability children moving with the voucher the academic quality declines in
the resource-poor schools more than the improvement in the selective school
via peer effects. In the absence of such negative externalities, aggregate
welfare could be positive due to the subsidy-led price decline for voucher
users and any positive conforming effects that raise the utility of attending
the high-achieving school when more high-ability children also do so. These
contradictory welfare implications are compatible with the same positive
coefficient on the social interaction term in an individual school choice model.

For standard discrete choice without spillover, Bhattacharya (2015) showed
that the choice probability function itself contains \emph{all} the
information required for exact welfare analysis. For the special case of
quasilinear random utility models with extreme value errors, the popular
`logsum' formula of Small and Rosen (1981) yields average welfare of policy
interventions. These results fail to hold in a setting with spillovers because
here one \emph{cannot} set\ the utility from the outside option to zero -- an
innocuous normalization in standard discrete choice models -- since this
utility changes as the equilibrium choice-rate changes with the policy
intervention. This is in contrast to binary choice \textit{without} spillover,
where utility from the outside option, i.e., non-purchase, does not change due
to a price change of the inside good.

Nonetheless, under a standard, linear-index specification of utilities, one
can calculate distribution-free bounds on average welfare, based solely on
choice probability functions. The width of the bounds increases with (i) the
extent of net social spillover, i.e. how much the (belief about) average
neighborhood choice affects individual choice probabilities, and (ii) the
difference in average peer-choice corresponding to realized equilibria before
and after the price change. The index structure, which has been universal in
the empirical literature on social interactions, leads to dimension reduction
that helps identify spillovers effects. We therefore continue to use the index
structure as it simplifies our expressions, and comes \textquotedblleft for
free\textquotedblright, because social spillovers cannot in general be
identified without such structure anyway. Under stronger and untestable
restrictions on the nature of spillover, our bounds can shrink to a singleton,
implying point-identification of welfare. Two such restrictions are (a) the
effects of an increase in average peer-choice on individual utilities from
buying and not buying are exactly equal in magnitude and opposite in sign, or
(b) the effect of aggregate choice on either the purchase utility or the
non-purchase utility is zero.

A separate identification problem arises when there are, in addition to social
interaction, unobserved group-effects that are potentially correlated with
observed individual covariates. We address this problem through a novel latent
factor structure on the relevant variables, and developing a method of
asymptotic analysis where the dimension of parameters, i.e. the group-effects
whose magnitude may be unbounded, increases as the number of groups increase.

\noindent\textbf{Empirical Illustration}: We illustrate our theoretical
results with an empirical example of a hypothetical, targeted public subsidy
scheme for anti-malarial bednets. In particular, we use micro-data from a
pricing experiment in rural Kenya (Dupas, 2014) to estimate an econometric
model of demand for bednets, where spillovers can arise via different
channels, including a preference for conformity and perceived negative
externality arising from neighbors' use of a bednet. In this setting, we
calculate predicted effects of hypothetical income-contingent subsidies on
bednet demand and welfare. We perform these calculations by first accounting
for social interactions, and then compare these results with what would be
obtained if one had ignored these interactions. We find that allowing for
(positive) interaction leads to a prediction of lower demand when means-tested
eligibility is restricted to fewer households and higher demand when the
eligibility criterion is more lenient, relative to ignoring interactions. To
illustrate, consider a policymaker debating whether to expand eligibility for
the subsidy from 40\% to 60\% of the population. In the presence of conforming
effects, the increase in eligibility will spur more non-eligible to adopt,
such that the total demand with spillovers are \textit{larger} than without
spillovers. Conversely, if eligibility was cut from 40\% to 20\%, the drop in
adoption would be magnified by conforming effects, such that the total demand
with spillovers are \textit{lower} than without spillovers.\footnote{The
intuition can be understood via a simple example. Suppose the true regression
model is $y=\beta_{0}+\beta_{1}x+u$, where $\beta_{1}>0$. Suppose, to predict
$y$ at a value $x_{0}$ of $x$, we ignore the covariate and simply use $\bar
{y}$ as the prediction. If $x_{0}<\bar{x}$, then the naive prediction $\bar
{y}=$ $\beta_{0}+\beta_{1}\bar{x}$ will be larger than the true value
$\beta_{0}+\beta_{1}x_{0}$, whereas if $x_{0}>\bar{x}$, then the naive
prediction $\bar{y}$ will be smaller than the true value $\beta_{0}+\beta
_{1}x_{0}$.} As for welfare, allowing for social interactions may lead to a
welfare \emph{loss} for ineligible households, in turn implying higher
deadweight loss from the subsidy scheme, relative to estimates obtained
ignoring social spillovers where welfare effects for ineligibles are zero by
definition. The resulting net welfare effect, aggregated over both eligibles
and ineligibles, admits a large range of possible values including both
positive and negative ones, with associated large variation in the implied
deadweight loss estimates, all of which are consistent with the \emph{same}
coefficient on the social interaction term in the choice probability function.

An implication of these results for applied work is that welfare analysis
under spillovers effects requires knowledge of the different channels of
spillovers separately, possibly via conducting a `belief elicitation' survey
where subjects are asked the reasons for their actions; knowledge of only the
choice probability functions, inclusive of a social interaction term, is
insufficient.\smallskip

\noindent\textbf{Plan of the Paper}: The rest of the paper is organized as
follows. Section \ref{sec: Set-up} describes the set-up, Section
\ref{sec: Welfare} develops the tools for empirical welfare analysis of a
price intervention in such models, and associated deadweight loss
calculations. Section \ref{sec: Stochastic Environment} specifies the
stochastic environment and derives the convergence of equilibrium beliefs
under I.I.D. unobservables. Section \ref{sec: Estimators} establishes
consistency of our estimator. Section \ref{sec: Emprical-context-data},
describes the context of our empirical application and the data; Section
\ref{sec: Empirical-specification-result} describes the empirical results;
Section \ref{sec: conclusion} summarizes and concludes. Technical derivations,
formal proofs and additional results are collected in an Appendix.

\section{SET-UP\label{sec: Set-up}}

Consider a population of villages indexed by $v\in\left\{  1,\dots,\bar
{v}\right\}  $ and resident households in village $v$ indexed by $\left(
v,h\right)  $, with $h\in\left\{  1,\dots,N_{v}\right\}  $. For the purpose of
inference discussed later, we will think of these households as a random
sample drawn from an infinite superpopulation. The total number of households
we observe is $N=\sum_{v=1}^{\bar{v}}N_{v}$. Each household faces a binary
choice between buying one unit of an indivisible good (alternative $1$) or not
buying it (alternative $0$). Its utilities from the two choices are given by
$U_{1}(Y_{vh}-P_{vh},\Pi_{vh},\boldsymbol{\eta}_{vh})$ and $U_{0}(Y_{vh}%
,\Pi_{vh},\boldsymbol{\eta}_{vh})$ where the variables $Y_{vh}$, $P_{vh}$, and
$\boldsymbol{\eta}_{vh}$ denote respectively the income, price, and
heterogeneity of household $(v,h)$, and $\Pi_{vh}$ is household $\left(
v,h\right)  $'s subjective belief of what fraction of households in her
village would choose to buy. The variable $\boldsymbol{\eta}_{vh}$ is
privately observed by household $\left(  v,h\right)  $ but is unobserved by
the econometrician and other households. The dependence of utilities on
$\Pi_{vh}$ captures social interactions. Below, we will specify how $\Pi_{vh}$
is formed. Household $(v,h)$'s choice is described by%
\begin{equation}
A_{vh}=1\left\{  U_{1}\left(  Y_{vh}-P_{vh},\Pi_{vh},\boldsymbol{\eta}%
_{vh}\right)  \geq U_{0}\left(  Y_{vh},\Pi_{vh},\boldsymbol{\eta}_{vh}\right)
\right\}  , \label{eq: outcome-variable}%
\end{equation}
where $1\left\{  \cdot\right\}  $ denotes the indicator function. In the
mosquito-net example of our application, one can interpret $U_{1}$ and $U_{0}$
as expected utilities resulting from differential probabilities of contracting
malaria from using and not using the net, respectively.

The utilities, $U_{1}$ and $U_{0}$, may also depend on other covariates of
$(v,h)$. For notational simplicity, we will occasionally write $W_{vh}%
=(Y_{vh},P_{vh})$\footnote{All vectors are defined as row vectors.}, and
suppress other covariates for now; additional covariates are used in our
empirical implementation in Section \ref{sec: Empirical-specification-result}.

\section{WELFARE\ ANALYSIS\label{sec: Welfare}}

We now lay out the empirical framework for welfare analysis of policy
interventions under spillovers. We will assume spillovers are restricted to
the village where households reside, hence welfare effects of a policy
intervention can be analyzed village by village; so for economy of notation,
we drop the $\left(  v,h\right)  $ subscripts except when we account
explicitly for village-effects during estimation. Also, we use the same
notation $\pi$ to denote both individual beliefs $\Pi_{vh}$ entering
individual utilities, and the unique equilibrium belief about village take-up
rate entering the average demand function. The assumption of a constant
(within village) $\pi$ is justified via Proposition 1 and Proposition 2 in
Section \ref{sec: eqm-beliefs}.

In the welfare results derived below, all probabilities and expectations --
e.g., mean welfare loss -- are calculated with respect to the \emph{marginal}
distribution of aggregate unobservables, denoted by $\boldsymbol{\eta}$. In
this sense, they are analogous to `average structural functions' (ASF),
introduced by Blundell and Powell (2004). Later, when discussing estimation of
the ASF, together with the implied pre- and post-intervention aggregate choice
probabilities and average welfare in Section \ref{sec: Group-specific}, we
will allude to village-effects explicitly, and show how they are estimated and
incorporated in demand and welfare predictions\textbf{.}

Define $q_{1}\left(  p,y,\pi\right)  $ to be the \emph{structural} probability
(i.e. average structural function; ASF) of a household choosing option $1$
(e.g., buying mosquito-net) when it faces a price of $p$, has income $y$ and
belief $\pi$:%
\begin{equation}
q_{1}\left(  p,y,\pi\right)  =\int1\left\{  U_{1}\left(  y-p,\pi
,\boldsymbol{s}\right)  >U_{0}\left(  y,\pi,\boldsymbol{s}\right)  \right\}
dF_{\boldsymbol{\eta}}\left(  \boldsymbol{s}\right)  \text{,} \label{7}%
\end{equation}
where $F_{\boldsymbol{\eta}}$ is the cumulative distribution function (CDF) of
$\boldsymbol{\eta}$. This probability can be estimated via the conditional
probability of purchase given covariates when household level unobservables
are uncorrelated with the covariates, as will be assumed in our application.
The reason for focusing on the ASF, rather than the purchase probability
conditional on covariates is that ultimately, we will be interested in the
marginal distribution of welfare in a village resulting from a
\textit{potential} price-intervention, e.g. a means-tested subsidy, which is
counterfactual.\footnote{Expressing our results in terms of ASFs help clarify
that in general, the object of interest is one, whose identification and
consistent estimation may require non-experimental methods if the data at hand
are observational. Also, later in the paper, we will allude to village level
unobservables i.e. $\eta_{vh}=\xi_{v}+u_{vh}$, where $\xi_{v}$ is a village
specific unobservable variable (introduced in Section
\ref{sec: Stochastic Environment}). In that context, the object of interest
will be the marginal distribution of welfare in each village; thus the
relevant distribution to used to compute the ASF will be that of $u_{vh}$
given village specific variables.}

\noindent\textbf{Linear Index Structure}: We now specify the forms of the
utility functions. Given a moderate/small number of large peer groups (e.g.,
there are eleven large villages in our application dataset), it is not easy to
consistently estimate the impact of the belief $\Pi_{vh}$ on the choice
probability function \emph{nonparametrically} holding other regressors
constant.\footnote{This is because $\Pi_{vh}$ is a constant within a village
(as discussed in Section \ref{sec: eqm-beliefs}). In particular, the fixed
point constraint, which is a notable feature of the social interaction model,
does not help because of dimensionality problems. Indeed, in the fixed point
condition: $\pi=\int q_{1}\left(  p,y,\pi\right)  dF_{P,Y}\left(  p,y\right)
$, where the joint CDF\ $F_{P,Y}\left(  p,y\right)  $ of $\left(  P,Y\right)
$ is identified, the unknown function $q_{1}\left(  p,y,\pi\right)  $ has more
arguments than the identified $F_{P,Y}\left(  p,y\right)  $.} Accordingly,
following Manski (1993), and Brock and Durlauf (2001, 2007), we assume a
linear index structure with $\boldsymbol{\eta}=(\eta^{0},\eta^{1})$ viz. the
utilities are given by%
\begin{equation}%
\begin{array}
[t]{c}%
U_{0}\left(  y,\pi,\boldsymbol{\eta}\right)  =\delta_{0}+\beta_{0}y+\alpha
_{0}\pi+\eta^{0}\text{,}\\
U_{1}\left(  y-p,\pi,\boldsymbol{\eta}\right)  =\delta_{1}+\beta_{1}\left(
y-p\right)  +\alpha_{1}\pi+\eta^{1}\text{,}%
\end{array}
\label{A}%
\end{equation}
where we assume that $\beta_{0}>0$, $\beta_{1}>0$, i.e., non-satiation in
numeraire, and $\beta_{1}$ need not equal $\beta_{0}$, i.e. income effects can
be present.\footnote{We can also allow for concave income effects by
specifying, say,%
\begin{align*}
U_{0}\left(  y,\pi,\eta\right)   &  =\delta_{0}+\beta_{0}\ln y+\alpha_{0}%
\pi+\eta^{0}\text{,}\\
U_{1}\left(  y-p,\pi,\eta\right)   &  =\delta_{1}+\beta_{1}\ln\left(
y-p\right)  +\alpha_{1}\pi+\eta^{1}\text{,}%
\end{align*}
but we wish to keep the utility formulation as simple as possible to highlight
the complications in welfare calculations even in the simplest linear utility
specification.}

In our empirical setting of anti-malarial bednet (Insecticide-Treated Net;
ITN, henceforth) adoption, there are multiple potential sources of
interactions (i.e. $\alpha_{1},\alpha_{0}\neq0$). The first is a pure
preference for conforming; the second is increased awareness of the benefits
of a bednet when more villagers use it; the third is the perceived health
externality. The medical literature suggests that the \emph{technological}
health externality is positive, i.e. as more people are protected, the lower
is the malaria burden, but the \emph{perceived} health externality can be
negative if households believe that other households' bednet use deflects
mosquitoes to unprotected households, but ignore the fact that those deflected
mosquitoes are less likely to carry the parasite. Indeed, the implications for
adoption are different: under the positive health externality, one would
expect free-riding, hence a negative effect of others' adoption on own
adoption; under the negative health externality, the correlation would be positive.

In particular, let $\gamma_{p}$ denote the conforming plus learning effect,
and $\gamma_{H}$ denote the health externality. Then it is reasonable to
assume that $\alpha_{1}\equiv\gamma_{p}\geq0$, while $\alpha_{0}\equiv
\gamma_{H}-\gamma_{p}$ could be either negative or positive. It is natural
that the conforming/learning/peer effect $\gamma_{p}$ affects utilities from
buying and non-buying symmetrically, i.e. if $\Pi$ changes from $0$ to $\pi$
the resulting change in the utility (relative to when $\Pi$ was $0$) from
buying and the one from not buying are symmetric and of opposite sign, as is
also assumed in BD01, BD07. Further, if a household uses an ITN, then there is
no health externality from the neighborhood adoption rate since the household
is protected anyway,\footnote{We can allow for a smaller health externality,
say $\gamma_{h}<\gamma_{H}$ when one adopts the bednet. But this does not
change the fundamental point about the asymmetric effect of $\pi$ on the
utility from buying and from not buying. So we avoided adding this to save on
notation.} but if it does not adopt, then there is a net health externality
effect $\gamma_{H}$ from neighborhood use, which makes the overall effect
$\alpha_{0}=\gamma_{H}-\gamma_{p}$ and in general, there is no exact
relationship between $\alpha_{0}=$ $\gamma_{H}-\gamma_{p}$ \ and $\alpha
_{1}=\gamma_{p}$.\footnote{An analogous asymmetry is also likely in the school
voucher example mentioned in the introduction if the voucher-led `brain-drain'
leads to utility gains and losses of different amounts, e.g., if better
teaching resources in the high-achieving school substitute for -- or
complement -- peer effects in a way that is not possible in the resource-poor
local school.} Accordingly, we first assume that the perceived \textit{net}
\textit{health} externality is non-positive, and thus $\alpha_{1}\geq
0\geq\alpha_{0}$, and derive welfare results. In the next subsection, we
present the results under the case $\alpha_{1}\geq\alpha_{0}\geq0$. Note that
the sign of $\alpha=\alpha_{1}-\alpha_{0}$ is identified, and is positive in
our data, which rules out $\alpha_{0}>\alpha_{1}\geq0$. In the application, we
present the bounds separately for $\alpha_{1}\geq0\geq\alpha_{0}$ and
$\alpha_{1}\geq\alpha_{0}\geq0$, and then the union of these.

Given the linear index specification, the structural choice probability of
buying at $\left(  p,y,\pi\right)  $ is given by%
\begin{equation}
q_{1}\left(  p,y,\pi\right)  =F(\underset{\delta_{1}-\delta_{0}%
}{\underbrace{c_{0}}}+\underset{-\beta_{1}}{\underbrace{c_{1}}}%
p+\underset{\beta_{1}-\beta_{0}}{\underbrace{c_{2}}}y+\underset{\alpha
_{1}-\alpha_{0}}{\underbrace{\alpha}}\pi)\text{,} \label{14}%
\end{equation}
where $F\left(  \cdot\right)  $ denotes the marginal distribution function of
$\eta^{0}-\eta^{1}$. It is known from Brock and Durlauf (2007) that the
structural choice probabilities $F\left(  c_{0}+c_{1}p+c_{2}y+\alpha
\pi\right)  $ identify $c_{0},c_{1},c_{2}$ and $\alpha$, i.e. $\left(
\delta_{1}-\delta_{0}\right)  $, $\beta_{0}$, $\beta_{1}$ and $\left(
\alpha_{1}-\alpha_{0}\right)  =2\gamma_{p}-\gamma_{H}$, up to scale even
without knowledge of the probability distribution of $\eta^{0}-\eta^{1}$. In
the application, we will consider two different estimates of the choice
probabilities, first ignoring village-specific unobservables and using a
standard probit, and then allowing for village-fixed unobservables using a
variant of correlated random effects.\smallskip

The distinct presence of $\alpha_{1},\alpha_{0}$ makes the model different
from standard demand models for binary choice. In the standard case, for the
so-called \textquotedblleft outside option\textquotedblright, i.e. not buying,
the utility is normalized to zero. In a social spillovers setting, this cannot
be done because that utility depends on the aggregate purchase rate $\pi$. As
we will see below, in welfare evaluations of a subsidy, $\alpha_{1}$ and
$\alpha_{0}$ appear separately in the expressions for welfare-distributions,
but cannot be separately identified from demand data, which can only identify
$\alpha\equiv\alpha_{1}-\alpha_{0}$. As a result, point-identification of
welfare will in general not be possible. Below, we will consider some
\emph{untestable} special cases, under which one obtains point-identification,
e.g., with $\alpha_{1}\geq0\geq\alpha_{0}$, the interesting special cases are
(i) $\alpha_{1}=\alpha/2=-\alpha_{0}$ (i.e. $\gamma_{H}=0$: no health
externality and symmetric spillover), which is considered in BD01 for social
welfare analysis, (ii) $\alpha_{1}=\alpha,$ $\alpha_{0}=0$ (i.e. $\gamma
_{H}=\gamma_{p}$: technological health externality dominates deflection
channel and net health externality exactly offsets conforming effect), and
(iii) $\alpha_{1}=0$, $\alpha_{0}=-\alpha$ ($\gamma_{p}=0$ and $\gamma
_{H}=-\alpha$: no conforming effect and deflection channel dominates). Cases
(ii) and (iii) will yield respectively the upper and lower bounds on welfare
gain for the case $\alpha_{1}\geq0\geq\alpha_{0}$. Analogously for the case
$\alpha_{1}\geq\alpha_{0}\geq0$.\smallskip

\noindent\textbf{Policy Intervention and Welfare Expressions}: We start with a
situation where the price of the product is $p_{0}$ and the value of $\pi$ is
$\pi_{0}$. Now suppose a price subsidy is introduced such that individuals
with income less than a threshold $\tau$ become eligible to buy the product at
price $p_{1}<p_{0}$. This policy will alter the equilibrium adoption rate;
suppose the new equilibrium adoption rate changes to $\pi_{1}$, where $\pi
_{0}$ and $\pi_{1}$, solve the fixed point conditions:%
\begin{align}
\pi_{0}  &  =\int F\left(  c_{0}+c_{1}p_{0}+c_{2}y+\alpha\pi_{0}\right)
dF_{Y}\left(  y\right)  \text{,}\label{c}\\
\pi_{1}  &  =\int\left[  1\left\{  y\leq\tau\right\}  F\left(  c_{0}%
+c_{1}p_{1}+c_{2}y+\alpha\pi_{1}\right)  +1\left\{  y>\tau\right\}  F\left(
c_{0}+c_{1}p_{0}+c_{2}y+\alpha\pi_{1}\right)  \right]  dF_{Y}\left(  y\right)
\text{,} \label{d}%
\end{align}
where $F_{Y}$ is the CDF of $Y_{vh}$, and $F$ is defined in (\ref{14}). Since
the price coefficient $c_{1}<0$ and $p_{1}>p_{0}$, therefore if $\alpha>0$, we
have that%
\begin{align*}
&  F\left(  c_{0}+c_{1}p_{0}+c_{2}y+\alpha\pi\right) \\
&  \leq\left[  1\left\{  y\leq\tau\right\}  F\left(  c_{0}+c_{1}p_{1}%
+c_{2}y+\alpha\pi\right)  +1\left\{  y>\tau\right\}  F\left(  c_{0}+c_{1}%
p_{0}+c_{2}y+\alpha\pi\right)  \right]
\end{align*}
for each $\pi$, and therefore, the integrand of (\ref{c}) is smaller for every
$\pi$ than the integrand of (\ref{d}). If the solutions to (\ref{c}) and
(\ref{d}) are unique, then the value of $\pi$ at which (\ref{c}) holds must be
smaller than the value of $\pi$ where (\ref{d}) holds. So we shall get
$\pi_{1}>\pi_{0}$. This is borne out in our application where sufficient
conditions on $\alpha$ for a contraction are satisfied. Under multiple
solutions, we can at least say that if $p_{1}<p_{0}$, the smallest solution
$\pi_{1}$\ to (\ref{d})\ is greater than the smallest solution $\pi_{0}$\ to
(\ref{c}).

For given values of $\pi_{0}$ and $\pi_{1}$, we now derive expressions for
welfare resulting from the intervention. By \textquotedblleft
welfare\textquotedblright\ we mean the compensating variation (CV), viz. what
hypothetical income compensation would restore the post-change indirect
utility for an individual to its pre-change level. For a
subsidy-\textit{eligible} individual, for any potential value of $\pi_{1}$
corresponding to the new equilibrium, the individual compensating variation is
the solution $S$ to the equation%
\begin{equation}
\max\left\{  U_{1}\left(  y+S-p_{1},\pi_{1},\boldsymbol{\eta}\right)
,U_{0}\left(  y+S,\pi_{1},\boldsymbol{\eta}\right)  \right\}  =\max\left\{
U_{1}\left(  y-p_{0},\pi_{0},\boldsymbol{\eta}\right)  ,U_{0}\left(  y,\pi
_{0},\boldsymbol{\eta}\right)  \right\}  \text{,} \label{a}%
\end{equation}
whereas for a subsidy-\textit{ineligible} individual, it is the solution $S$
to%
\begin{equation}
\max\left\{  U_{1}\left(  y+S-p_{0},\pi_{1},\boldsymbol{\eta}\right)
,U_{0}\left(  y+S,\pi_{1},\boldsymbol{\eta}\right)  \right\}  =\max\left\{
U_{1}\left(  y-p_{0},\pi_{0},\boldsymbol{\eta}\right)  ,U_{0}\left(  y,\pi
_{0},\boldsymbol{\eta}\right)  \right\}  \text{.} \label{b}%
\end{equation}
Thus we interpret the CV as measuring utility changes via the value of
hypothetical income compensation that would restore utilities to their initial
level.\footnote{Note that we do not take account of peer effects of this
hypothetical income compensation, which might be an alternative way to define
the CV.} Now, since $S$ depends on the unobservables $\boldsymbol{\eta}$, the
same price change will produce a \textit{distribution} of welfare effects
across individuals; we are interested in calculating that distribution and its
functionals such as mean welfare.\medskip

The welfare effect of the subsidy can be calculated as described below.

\subsection{Welfare Calculation under $\alpha_{1}\geq0>\alpha_{0}$}

Recall that $\alpha_{1}=\gamma_{p}\geq0$, $\alpha_{0}=\gamma_{H}-\gamma_{p}$;
thus $\alpha_{1}\geq0>\alpha_{0}$ corresponds to the case where either
$\gamma_{H}<0$, i.e., deflection effect dominates positive health effect in
perception, or is positive but smaller than conforming/learning effect.

\noindent\textbf{Welfare for Eligibles (}$\alpha_{1}\geq0>\alpha_{0}%
$\textbf{): }The CV for a subsidy-eligible household is given by the solution
$S$ to%
\begin{align}
&  \max\left\{  \delta_{1}+\beta_{1}\left(  y+S-p_{1}\right)  +\alpha_{1}%
\pi_{1}+\eta^{1},\delta_{0}+\beta_{0}\left(  y+S\right)  +\alpha_{0}\pi
_{1}+\eta^{0}\right\} \nonumber\\
&  =\max\left\{  \delta_{1}+\beta_{1}\left(  y-p_{0}\right)  +\alpha_{1}%
\pi_{0}+\eta^{1},\delta_{0}+\beta_{0}y+\alpha_{0}\pi_{0}+\eta^{0}\right\}  .
\label{g}%
\end{align}
The resulting solution $S$ depends on the unobservable heterogeneity $\eta
^{0}$ and $\eta^{1}$ and hence we are interested in deriving its distribution
and functionals thereof such as mean welfare. Calculating the welfare
distribution requires us to compute the CDF of $S$, i.e. $\Pr\left(  S\leq
a\right)  $ for various values of $a$ (for given $(p_{0},\pi_{0},p_{1},\pi
_{1})$). Let $f_{\eta^{0}-\eta^{1}}\left(  \cdot\right)  $ denote the marginal
density function of $\eta^{0}-\eta^{1}$. Then the expression for the CDF of
welfare is as follows:

\begin{theorem}
\label{th: elig}Suppose the linear index structure described above holds with
$\beta_{0}>0$, $\beta_{1}>0$ and $\alpha_{1}\geq0\geq\alpha_{0}$ with
$\alpha=\alpha_{1}-\alpha_{0}$ satisfying $\left\vert \alpha\right\vert
\sup_{e\in\mathbb{R}}f_{\eta^{0}-\eta^{1}}\left(  e\right)  <1$. Then $\pi
_{1}>\pi_{0}$, and the distribution of compensating variation for the
eligibles, $S=S^{\mathrm{Elig}}$, is given by%
\begin{align}
&  \Pr\left(  S^{\mathrm{Elig}}\leq a\right) \nonumber\\
&  =\left\{
\begin{array}
[c]{ll}%
0\text{,} & \text{if }a<p_{1}-p_{0}-\frac{\alpha_{1}}{\beta_{1}}\left(
\pi_{1}-\pi_{0}\right)  \text{,}\\
q_{1}\left(  p_{1}-a,y,\pi_{0}+\frac{\alpha_{1}}{\alpha}\left(  \pi_{1}%
-\pi_{0}\right)  \right)  \text{,} & \text{if }p_{1}-p_{0}-\frac{\alpha_{1}%
}{\beta_{1}}\left(  \pi_{1}-\pi_{0}\right)  \leq a<\frac{\alpha-\alpha_{1}%
}{\beta_{0}}\left(  \pi_{1}-\pi_{0}\right)  \text{,}\\
1\text{,} & \text{if }a\geq\frac{\alpha-\alpha_{1}}{\beta_{0}}\left(  \pi
_{1}-\pi_{0}\right)  \text{.}%
\end{array}
\right.  \label{1}%
\end{align}

\end{theorem}

The proof is provided in Appendix \ref{sec: appendix-welfare}. The condition
$\left\vert \alpha\right\vert \sup_{e\in\mathbb{R}}f_{\eta^{0}-\eta^{1}%
}\left(  e\right)  <1$ essentially says that the social interaction parameter
is not too large in magnitude, so that (\ref{c}) and (\ref{d}) have unique
solutions in $\pi_{0}$ and $\pi_{1}$ respectively, whence by the argument
following (\ref{c}) and (\ref{d}), we have that $\pi_{1}>\pi_{0}$.

Now, note that in the intermediate case in (\ref{1}), where $a\in\lbrack
p_{1}-p_{0}-\frac{\alpha_{1}}{\beta_{1}}\left(  \pi_{1}-\pi_{0}\right)  ,$
$\frac{\alpha_{0}}{\beta_{0}}\left(  \pi_{0}-\pi_{1}\right)  ]$, $\Pr\left(
S\leq a\right)  $ equals%
\begin{equation}
q_{1}(p_{1}-a,y,\pi_{0}+\frac{\alpha_{1}}{\alpha}\left(  \pi_{1}-\pi
_{0}\right)  )=F\left(  c_{0}+\alpha_{1}\left(  \pi_{1}-\pi_{0}\right)
+c_{1}\left(  p_{1}-a\right)  +c_{2}y+\alpha\pi_{0}\right)  \text{.}
\label{f3}%
\end{equation}
In (\ref{f3}), the intercept $c_{0}=\delta_{1}-\delta_{0}$, the slopes
$c_{1}=-\beta_{1},c_{2}=\beta_{1}-\beta_{0}$ and $\alpha=\alpha_{1}-\alpha
_{0}$ are all identified from conditional choice probabilities; \emph{however
}$\alpha_{1}$\emph{ is not identified} and therefore (\ref{f3}) is not
point-identified from the structural choice probabilities. However, since
$\alpha_{1}\in\left[  0,\alpha\right]  $, for each feasible value of
$\alpha_{1}\in\left[  0,\alpha\right]  $, we can compute a corresponding value
of (\ref{f3}), giving us bounds on the welfare distribution.

Note also that the thresholds of $a$ at which the CDF expression changes are
also not point-identified for the same reason. However, since $\pi_{1}-\pi
_{0}>0$ and $\beta_{0}>0$, $\beta_{1}>0$, the interval%
\[
p_{1}-p_{0}-\frac{\alpha_{1}}{\beta_{1}}\left(  \pi_{1}-\pi_{0}\right)  \leq
a<\frac{\alpha_{0}}{\beta_{0}}\left(  \pi_{0}-\pi_{1}\right)
\]
will translate to the left as $\alpha_{1}$ varies from $0$ to $\alpha
$.\smallskip

\begin{remark}
\label{Remark}Note that the above theorem continues to hold even if the
subsidy is universal; we have not used the means-tested nature of the subsidy
to derive the result.
\end{remark}

\begin{corollary}
[\textbf{Mean Welfare}]From (\ref{1}), the mean welfare for the eligible is
given by%
\begin{align}
E[S^{\mathrm{Elig}}]  &  =\underset{\text{welfare gain}}{\underbrace{-\int%
_{p_{1}-p_{0}-\frac{\alpha_{1}}{\beta_{1}}\left(  \pi_{1}-\pi_{0}\right)
}^{0}q_{1}\left(  p_{1}-a,y,\pi_{0}+\frac{\alpha_{1}}{\alpha}\left(  \pi
_{1}-\pi_{0}\right)  \right)  da}}\nonumber\\
&  +\underset{\text{welfare loss}}{\underbrace{\int_{0}^{\frac{\alpha
-\alpha_{1}}{\beta_{0}}\left(  \pi_{1}-\pi_{0}\right)  }\left[  1-q_{1}\left(
p_{1}-a,y,\pi_{0}+\frac{\alpha_{1}}{\alpha}\left(  \pi_{1}-\pi_{0}\right)
\right)  \right]  da}}, \label{9}%
\end{align}
where the following formula for a random variable $X$ that has finite mean and
the CDF, $F_{X}$, is used: $E\left[  X\right]  =\int_{0}^{\infty}%
[1-F_{X}\left(  x\right)  ]dx-\int_{-\infty}^{0}F_{X}\left(  x\right)  dx$.
\end{corollary}

\noindent The width of the bounds on (\ref{1}) and (\ref{9}), obtained by
varying $\alpha_{1}$ over $\left[  0,\alpha\right]  $, depends on the extent
to which $q_{1}\left(  \cdot,\cdot,\pi\right)  $ is affected by $\pi$, i.e.
the extent of social spillover, and also the difference in the realized values
$\pi_{1}$ and $\pi_{0}$. For our single-index model, the fixed point
restrictions imply that these counterfactual $\pi_{1}$ and $\pi_{0}$ depend on
$\alpha_{1}$ and $\alpha_{0}$ only via $\alpha=\alpha_{1}-\alpha_{0}$ (cf.
(\ref{c}) and (\ref{d}) above) which is point-identified; thus every potential
value of counterfactual demand is point-identified. But given any feasible
value of $\pi_{1}$ and $\pi_{0}$, the welfare (\ref{9}) is not
point-identified in general, since $\alpha_{1}$ is unknown.

However, given $\alpha$, the welfare gain in expression (\ref{9}) is
increasing in $\alpha_{1}$; i.e., the welfare gain is largest in absolute
value when $\alpha_{1}=\alpha$ and $\alpha_{0}=0$, and the smallest when
$\alpha_{1}=0$ and $\alpha_{0}=-\alpha$; conversely for welfare loss.
Intuitively, if there is no negative externality from increased $\pi$ on
non-purchasers, then they do not suffer any welfare loss, but purchasers have
a welfare gain from both lower price and higher $\pi$. Conversely, if all the
spillovers are negative, then purchasers still receive a welfare gain via
price reduction, but non-purchasers suffer welfare loss due to increased $\pi
$. Also, note that under quasilinear utilities (i.e., utilities with
$\beta_{0}=\beta_{1}$), where income effects are absent, the $y$ drops out of
the above expressions, but the same identification problem remains, since
$\alpha_{1}$ does not disappear. Changing variables $p=p_{1}-a$, one can
rewrite (\ref{9}) as%
\begin{align}
E[S^{\mathrm{Elig}}]  &  =\underset{\text{welfare gain}}{\underbrace{-\int%
_{p_{1}}^{p_{0}+\frac{\alpha_{1}}{\beta_{1}}\left(  \pi_{1}-\pi_{0}\right)
}q_{1}\left(  p,y,\pi_{0}+\frac{\alpha_{1}}{\alpha}\left(  \pi_{1}-\pi
_{0}\right)  \right)  dp}}\nonumber\\
&  +\underset{\text{welfare Loss}}{\underbrace{\int_{p_{1}+\frac{\alpha
_{1}-\alpha}{\beta_{0}}\left(  \pi_{1}-\pi_{0}\right)  }^{p_{1}}\left[
1-q_{1}\left(  p,y,\pi_{0}+\frac{\alpha_{1}}{\alpha}\left(  \pi_{1}-\pi
_{0}\right)  \right)  \right]  dp}}\text{.} \label{10}%
\end{align}
Note that if $\alpha_{1}=0$, then the first term is the usual consumer surplus
capturing the effect of price reduction on consumer welfare; for a positive
$\alpha_{1}$, the term $\frac{\alpha_{1}}{\beta_{1}}\left(  \pi_{1}-\pi
_{0}\right)  $ yields the additional effect arising via the conforming
channel. Also, if $\alpha_{1}=0$, then the second term, i.e. the welfare loss
from not buying, is the largest (given $\alpha$): this corresponds to the case
where all of $\alpha$ is due to the negative externality.

The second term in (\ref{10}), which represents welfare change caused solely
via spillovers and no price change, is still expressed as an integral with
respect to price. This is a consequence of the index structure which enables
us to express this welfare loss in terms of foregone utility from an
equivalent price change.\smallskip

\noindent\textbf{Special Cases}: For the special case of symmetric
interactions considered in BD01, where their social welfare is calculated with
$\alpha_{1}=-\alpha_{0}$ in (\ref{A}) (e.g., if $\gamma_{H}=0$, i.e. there is
no health externality in the health-good example), we have $\dfrac{\alpha_{1}%
}{\alpha}=\dfrac{-\alpha_{0}}{-2\alpha_{0}}=\dfrac{1}{2}$, and from (\ref{10})
mean welfare equals:
\begin{equation}
\underset{\text{welfare gain}}{\underbrace{-\int_{p_{1}}^{p_{0}+\frac{\alpha
}{2\beta_{1}}\left(  \pi_{1}-\pi_{0}\right)  }q_{1}\left(  p,y,\frac{1}%
{2}\left(  \pi_{1}+\pi_{0}\right)  \right)  dp}}+\underset{\text{welfare
loss}}{\underbrace{\int_{p_{1}-\frac{\alpha}{2\beta_{0}}\left(  \pi_{1}%
-\pi_{0}\right)  }^{p_{1}}\left[  1-q_{1}\left(  p,y,\frac{1}{2}\left(
\pi_{1}+\pi_{0}\right)  \right)  \right]  dp\text{.}}} \label{11}%
\end{equation}
If $\alpha_{0}=0$, and $\alpha=\alpha_{1}$, i.e. all spillovers are via
conforming, mean welfare is given by%
\begin{equation}
\underset{\text{welfare gain}}{\underbrace{-\int_{p_{1}}^{p_{0}+\frac{\alpha
}{\beta_{1}}\left(  \pi_{1}-\pi_{0}\right)  }q_{1}\left(  p,y,\pi_{1}\right)
dp}}\text{;} \label{12}%
\end{equation}
and if, on the other hand, any spillovers are due to perceived health risk,
i.e. $\alpha=-\alpha_{0}$ and $\alpha_{1}=0$, then mean welfare is given by
\begin{equation}
\underset{\text{welfare gain}}{\underbrace{-\int_{p_{1}}^{p_{0}}q_{1}\left(
p,y,\pi_{0}\right)  dp}}+\underset{\text{welfare loss}}{\underbrace{\int%
_{p_{1}-\frac{\alpha}{\beta_{0}}\left(  \pi_{1}-\pi_{0}\right)  }^{p_{1}%
}\left[  1-q_{1}\left(  p,y,\pi_{0}\right)  \right]  dp}}\text{.} \label{13}%
\end{equation}

Expressions (\ref{12}) and (\ref{13}) correspond to the upper and lower
bounds, respectively, of the overall welfare gain for
eligibles.\footnote{Gautam (2018) obtained point-identified estimates of
welfare in parametric discrete choice models with social interactions,
purportedly using Dagsvik and Karlstrom's (2005) results for the setting
without spillover. Gautam's paper contains no explicit expression for average
welfare, but we conjecture that her derivation had implicitly assumed one of
the normalizations (\ref{11}), (\ref{12}) or (\ref{13}) under which average
welfare is point-identified.}\smallskip

\noindent\textbf{Welfare for Ineligibles (}$\alpha_{1}\geq0>\alpha_{0}%
$\textbf{): }Welfare change for \textit{ineligibles} is measured by the CV
defined as the solution $S$ to the equation:%
\begin{align}
&  \max\left\{  \delta_{1}+\beta_{1}\left(  y+S-p_{0}\right)  +\alpha_{1}%
\pi_{1}+\eta^{1},\delta_{0}+\beta_{0}\left(  y+S\right)  +\alpha_{0}\pi
_{1}+\eta^{0}\right\} \nonumber\\
&  =\max\left\{  \delta_{1}+\beta_{1}\left(  y-p_{0}\right)  +\alpha_{1}%
\pi_{0}+\eta^{1},\delta_{0}+\beta_{0}y+\alpha_{0}\pi_{0}+\eta^{0}\right\}
\text{,} \label{S_inelig}%
\end{align}
which is simply (\ref{g}) with $p_{1}$ replaced by $p_{0}$. Therefore, the
mean CV is simply (\ref{g}) with $p_{1}$ replaced by $p_{0}$.

\begin{corollary}
\label{corollary: inelig}Suppose the linear index structure described above
holds with $\beta_{0}>0$, $\beta_{1}>0$, and $\alpha_{1}\geq0\geq\alpha_{0}$.
Then for each $\alpha_{1}\in\left[  0,\alpha\right]  $, the mean welfare for
the ineligible, $S=S^{\mathrm{Inelig}}$, is given by
\begin{align}
E[S^{\mathrm{Inelig}}]  &  =-\int_{p_{0}}^{p_{0}+\frac{\alpha_{1}}{\beta_{1}%
}\left(  \pi_{1}-\pi_{0}\right)  }q_{1}\left(  p,y,\pi_{0}+\frac{\alpha_{1}%
}{\alpha}\left(  \pi_{1}-\pi_{0}\right)  \right)  dp\nonumber\\
&  +\int_{p_{0}+\frac{\alpha_{1}-\alpha}{\beta_{0}}\left(  \pi_{1}-\pi
_{0}\right)  }^{p_{0}}\left[  1-q_{1}\left(  p,y,\pi_{0}+\frac{\alpha_{1}%
}{\alpha}\left(  \pi_{1}-\pi_{0}\right)  \right)  \right]  dp\text{.}
\label{mean_noneligibles}%
\end{align}

\end{corollary}

\noindent For \textit{ineligibles}, all of the welfare effects come from
spillovers, since they experience no price change. In particular, for
ineligibles who buy, there is a welfare gain from positive spillovers due to a
higher $\pi$. For ineligibles who do not buy, there is, however, a potential
welfare loss due to increased $\pi$. This is why the CV distribution has the
support that includes both positive and negative values. The first term in
(\ref{mean_noneligibles}) captures the welfare gain resulting from a positive
$\alpha_{1}$ and higher $\pi$; this term would be zero if $\alpha_{1}=0$. The
second term in (\ref{mean_noneligibles}) captures the welfare loss also
resulting from higher $\pi$; this loss would be zero if there are no negative
impacts, i.e. $\alpha_{0}=0$. Of course, both would be zero if $\alpha
=0=\alpha_{1}=\alpha_{0}$, reflecting the fact that welfare effect on
ineligibles would be zero if there is no spillover.

In the three special cases where we have point-identification, viz. (i)
$\alpha_{1}=-\alpha_{0}=\frac{\alpha}{2}$; (ii) $\alpha=\alpha_{1}$,
$\alpha_{0}=0$; and (iii) $\alpha=-\alpha_{0}$, $\alpha_{1}=0$, mean CV
(\ref{mean_noneligibles}) reduces respectively to:
\begin{align}
&  \mathrm{(i)}\text{ \ }\underset{\text{welfare gain}}{\underbrace{-\int%
_{p_{0}}^{p_{0}+\frac{\alpha}{2\beta_{1}}\left(  \pi_{1}-\pi_{0}\right)
}q_{1}\left(  p,y,\frac{\pi_{0}+\pi_{1}}{2}\right)  dp}}%
+\underset{\text{welfare loss}}{\underbrace{\int_{p_{0}-\frac{\alpha}%
{2\beta_{0}}\left(  \pi_{1}-\pi_{0}\right)  }^{p_{0}}\left[  1-q_{1}\left(
p,y,\frac{\pi_{0}+\pi_{1}}{2}\right)  \right]  dp}}\text{;}\nonumber\\
&  \mathrm{(ii)}\text{ \ }\underset{\text{welfare gain}}{\underbrace{-\int%
_{p_{0}}^{p_{0}+\frac{\alpha}{\beta_{1}}\left(  \pi_{1}-\pi_{0}\right)  }%
q_{1}\left(  p,y,\pi_{1}\right)  dp}}\text{;}\label{15}\\
&  \mathrm{(iii)}\text{ \ }\underset{\text{welfare loss}}{\underbrace{\int%
_{p_{0}-\frac{\alpha}{\beta_{0}}\left(  \pi_{1}-\pi_{0}\right)  }^{p_{0}%
}\left[  1-q_{1}\left(  p,y,\pi_{0}\right)  \right]  dp}}\text{.} \label{16}%
\end{align}

Expressions (\ref{15}) and (\ref{16}) correspond to the upper and lower
bounds, respectively, of the overall welfare gain for ineligibles, and
therefore, the overall bounds generically contain both positive and negative
values, since $\alpha\neq0$.\smallskip

\noindent\textbf{Deadweight Loss (}$\alpha_{1}\geq0>\alpha_{0}$\textbf{): }The
mean deadweight loss (DWL) can be calculated as the expected subsidy spending
less the net welfare gain:
\begin{align*}
&  DWL(y)=\underset{\text{subsidy spending}}{\underbrace{1\left\{  y\leq
\tau\right\}  \times\left(  p_{0}-p_{1}\right)  q_{1}\left(  p_{1},y,\pi
_{1}\right)  }}\\
&  \underset{\text{welfare gain of eligibles}}{\underbrace{-1\left\{
y\leq\tau\right\}  \times\left(
\begin{array}
[c]{c}%
%TCIMACRO{\dint _{p_{1}-p_{0}-\frac{\alpha_{1}}{\beta_{1}}\left(  \pi_{1}%
%-\pi_{0}\right)  }^{0}}%
%BeginExpansion
{\displaystyle\int_{p_{1}-p_{0}-\frac{\alpha_{1}}{\beta_{1}}\left(  \pi
_{1}-\pi_{0}\right)  }^{0}}
%EndExpansion
q_{1}\left(  p_{1}-a,y,\pi_{0}+\dfrac{\alpha_{1}}{\alpha}\left(  \pi_{1}%
-\pi_{0}\right)  \right)  da\\
-%
%TCIMACRO{\dint _{0}^{\frac{\alpha-\alpha_{1}}{\beta_{0}}\left(  \pi_{1}%
%-\pi_{0}\right)  }}%
%BeginExpansion
{\displaystyle\int_{0}^{\frac{\alpha-\alpha_{1}}{\beta_{0}}\left(  \pi_{1}%
-\pi_{0}\right)  }}
%EndExpansion
\left[  1-q_{1}\left(  p_{1}-a,y,\pi_{0}+\dfrac{\alpha_{1}}{\alpha}\left(
\pi_{1}-\pi_{0}\right)  \right)  \right]  da
\end{array}
\right)  }}\\
&  \underset{\text{welfare gain of ineligibles}}{\underbrace{-1\left\{
y>\tau\right\}  \times\left[
\begin{array}
[c]{c}%
%TCIMACRO{\dint _{p_{0}}^{p_{0}+\frac{\alpha_{1}}{\beta_{1}}\left(  \pi_{1}%
%-\pi_{0}\right)  }}%
%BeginExpansion
{\displaystyle\int_{p_{0}}^{p_{0}+\frac{\alpha_{1}}{\beta_{1}}\left(  \pi
_{1}-\pi_{0}\right)  }}
%EndExpansion
q_{1}\left(  p,y,\pi_{0}+\dfrac{\alpha_{1}}{\alpha}\left(  \pi_{1}-\pi
_{0}\right)  \right)  dp\\
-%
%TCIMACRO{\dint _{p_{0}+\frac{\alpha_{0}}{\beta_{0}}\left(  \pi_{1}-\pi
%_{0}\right)  }^{p_{0}}}%
%BeginExpansion
{\displaystyle\int_{p_{0}+\frac{\alpha_{0}}{\beta_{0}}\left(  \pi_{1}-\pi
_{0}\right)  }^{p_{0}}}
%EndExpansion
\left[  1-q_{1}\left(  p,y,\pi_{0}+\dfrac{\alpha_{1}}{\alpha}\left(  \pi
_{1}-\pi_{0}\right)  \right)  \right]  dp
\end{array}
\right]  }}\text{..}%
\end{align*}
The bounds on $\alpha_{1}$ then translate into bounds for mean DWL. In
particular, if $\alpha_{0}=0$ (so that $\alpha=\alpha_{1}$), then%
\begin{align*}
&  DWL(y)=1\left\{  y\leq\tau\right\}  \times\left(  p_{0}-p_{1}\right)
q_{1}\left(  p_{1},y,\pi_{1}\right) \\
&  -1\left\{  y\leq\tau\right\}  \times\int_{p_{1}}^{p_{0}+\frac{\alpha}%
{\beta_{1}}\left(  \pi_{1}-\pi_{0}\right)  }q_{1}\left(  p,y,\pi_{1}\right)
dp-1\left\{  y>\tau\right\}  \times\int_{p_{0}}^{p_{0}+\frac{\alpha}{\beta
_{1}}\left(  \pi_{1}-\pi_{0}\right)  }q_{1}\left(  p,y,\pi_{1}\right)
dp\text{.}%
\end{align*}
Therefore, if $\frac{\alpha}{\beta_{1}}\left(  \pi_{1}-\pi_{0}\right)  $ is
sufficiently large, then the mean DWL will be \emph{negative}, i.e. the
subsidy will \emph{increase} economic efficiency under positive spillover, as
in the standard textbook case. This happens because there is no subsidy
expenditure on ineligibles, and yet those ineligibles who buy enjoy a
subsidy-induced welfare gain due to positive spillover. Subsidy-eligibles
receive an additional welfare gain via positive spillover, over and above the
welfare-gain due to reduced price, and it is only the latter that is financed
by the subsidy expenditure. In general, the deadweight loss will be lower
(more negative) when (i) the positive spillovers ($\alpha_{1}$) is larger,
(ii) the change in equilibrium adoption ($\pi_{1}-\pi_{0}$) due to the subsidy
is greater, and (iii) the price elasticity of demand ($-\beta_{1}$) is lower
-- the last effect lowers deadweight loss simply by reducing the substitution
effect, even in absence of spillover.

\subsection{Mean Welfare under $\alpha_{1}\geq\alpha_{0}\geq0$}

Recall that $\alpha_{1}=\gamma_{p}\geq0$, $\alpha_{0}=\gamma_{H}-\gamma_{p}$;
in our application, it holds that $\alpha=\alpha_{1}-\alpha_{0}>0$; thus
$\alpha_{1}>\alpha_{0}\geq0$ corresponds to the case where $\gamma_{H}>0$ i.e.
insecticide effect dominates deflection effect, is also larger than
conforming/learning but less than twice the conforming effect. Note that under
this assumption, we must also have $\alpha\leq\alpha_{1}$.\smallskip

\noindent\textbf{Welfare for Eligibles (}$\alpha_{1}\geq\alpha_{0}\geq
0$\textbf{): }For subsidy-eligibles, the mean welfare (for given $(p_{0}%
,\pi_{0},p_{1},\pi_{1})$) is presented in the following theorem:

\begin{theorem}
Suppose the linear index structure described above holds with $\beta_{1}%
\geq\beta_{0}>0$, and $\alpha_{1}\geq\alpha_{0}\geq0$. Let $\beta=\beta
_{1}-\beta_{0}$ and $\alpha=\alpha_{1}-\alpha_{0}$, which are estimable from
the choice probability function, and define
\begin{align*}
C_{1}\left(  \alpha_{1}\right)   &  :=-%
%TCIMACRO{\dint _{p_{1}-\frac{\alpha-\alpha_{1}}{\beta_{0}}\left(  \pi_{1}%
%-\pi_{0}\right)  }^{p_{0}+\frac{\alpha_{1}}{\beta_{1}}\left(  \pi_{1}-\pi
%_{0}\right)  }}%
%BeginExpansion
{\displaystyle\int_{p_{1}-\frac{\alpha-\alpha_{1}}{\beta_{0}}\left(  \pi
_{1}-\pi_{0}\right)  }^{p_{0}+\frac{\alpha_{1}}{\beta_{1}}\left(  \pi_{1}%
-\pi_{0}\right)  }}
%EndExpansion
q_{1}\left(  p,y,\pi_{0}+\frac{\alpha_{1}}{\alpha}\left(  \pi_{1}-\pi
_{0}\right)  \right)  dp\text{,}\\
C_{2}\left(  \alpha_{1}\right)   &  :=-%
%TCIMACRO{\dint _{p_{0}+\frac{\alpha-\alpha_{1}}{\beta_{0}}\left(  \pi_{1}%
%-\pi_{0}\right)  }^{p_{1}-\frac{\alpha_{1}}{\beta_{1}}\left(  \pi_{1}-\pi
%_{0}\right)  }}%
%BeginExpansion
{\displaystyle\int_{p_{0}+\frac{\alpha-\alpha_{1}}{\beta_{0}}\left(  \pi
_{1}-\pi_{0}\right)  }^{p_{1}-\frac{\alpha_{1}}{\beta_{1}}\left(  \pi_{1}%
-\pi_{0}\right)  }}
%EndExpansion
\left[  1-q_{1}\left(  p,y+p-p_{0},\pi_{1}-\frac{\alpha_{1}}{\alpha}\left(
\pi_{1}-\pi_{0}\right)  \right)  \right]  dp\text{.}%
\end{align*}
Then mean welfare for the eligible is given by%
\[
E\left[  S^{\mathrm{Elig}}\right]  =\left\{
\begin{array}
[c]{ll}%
C_{1}\left(  \alpha_{1}\right)  \text{,} & \text{if }\alpha\leq\alpha_{1}%
\leq\frac{\beta_{1}}{\beta}\left(  \beta_{1}-\beta\right)  \frac{p_{0}-p_{1}%
}{\pi_{1}-\pi_{0}}+\frac{\beta_{1}}{\beta}\alpha\text{,}\\
C_{2}\left(  \alpha_{1}\right)  \text{,} & \text{if }\alpha_{1}>\frac
{\beta_{1}}{\beta}\left(  \beta_{1}-\beta\right)  \frac{p_{0}-p_{1}}{\pi
_{1}-\pi_{0}}+\frac{\beta_{1}}{\beta}\alpha\text{.}%
\end{array}
\right.
\]

\end{theorem}

The proof is provided in Appendix \ref{sec: appendix-welfare}. Given that
$\alpha_{1}$ is unknown, this result implies the lower and upper bounds of the
mean welfare:%
\begin{align}
LB_{\alpha_{1}\geq\alpha_{0}\geq0}^{\mathrm{Elig}}  &  =\min\left\{
\inf_{\alpha_{1}\in\left[  \alpha,\frac{\beta_{1}}{\beta}\left(  \beta
_{1}-\beta\right)  \frac{p_{0}-p_{1}}{\pi_{1}-\pi_{0}}+\frac{\beta_{1}}{\beta
}\alpha\right]  }C_{1}\left(  \alpha_{1}\right)  ,\text{ }\inf_{\alpha_{1}%
\in\left[  \frac{\beta_{1}}{\beta}\left(  \beta_{1}-\beta\right)  \frac
{p_{0}-p_{1}}{\pi_{1}-\pi_{0}}+\frac{\beta_{1}}{\beta}\alpha,\infty\right)
}C_{2}\left(  \alpha_{1}\right)  \right\}  ,\label{5}\\
UB_{\alpha_{1}\geq\alpha_{0}\geq0}^{\mathrm{Elig}}  &  =\max\left\{
\sup_{\alpha_{1}\in\left[  \alpha,\frac{\beta_{1}}{\beta}\left(  \beta
_{1}-\beta\right)  \frac{p_{0}-p_{1}}{\pi_{1}-\pi_{0}}+\frac{\beta_{1}}{\beta
}\alpha\right]  }C_{1}\left(  \alpha_{1}\right)  ,\text{ }\sup_{\alpha_{1}%
\in\left[  \frac{\beta_{1}}{\beta}\left(  \beta_{1}-\beta\right)  \frac
{p_{0}-p_{1}}{\pi_{1}-\pi_{0}}+\frac{\beta_{1}}{\beta}\alpha,\infty\right)
}C_{2}\left(  \alpha_{1}\right)  \right\}  . \label{8}%
\end{align}
Thus, allowing for both $\alpha_{1}\geq\alpha_{0}\geq0$ and $\alpha_{1}%
\geq0\geq\alpha_{0}$ yields the wider bounds on mean welfare for the eligible
to:%
\begin{align}
LB^{\mathrm{Elig}}  &  =\min\left\{  LB_{\alpha_{1}\geq\alpha_{0}\geq
0}^{\mathrm{Elig}},\text{ }LB_{\alpha_{1}\geq0\geq\alpha_{0}}^{\mathrm{Elig}%
}\right\}  \text{,}\label{21}\\
UB^{\mathrm{Elig}}  &  =\max\left\{  UB_{\alpha_{1}\geq\alpha_{0}\geq
0}^{\mathrm{Elig}},\text{ }UB_{\alpha_{1}\geq0\geq\alpha_{0}}^{\mathrm{Elig}%
}\right\}  \text{.} \label{22}%
\end{align}
where $LB_{\alpha_{1}\geq0\geq\alpha_{0}}^{\mathrm{Elig}}$ and $UB_{\alpha
_{1}\geq0\geq\alpha_{0}}^{\mathrm{Elig}}$ are defined as expressions
(\ref{16}) and (\ref{15}), respectively. Since we expect $\beta_{1}>\beta_{0}$
(also borne out by the empirical results), $C_{2}\left(  \alpha_{1}\right)  $
will tend to $-\infty$ as $\alpha_{1},\alpha_{0}\rightarrow\infty$. Therefore,
as $\alpha_{1},\alpha_{0}\rightarrow\infty$, the integrand in $C_{2}\left(
\alpha_{1}\right)  $ will tend to 1 and $LB_{\alpha_{1}\geq\alpha_{0}\geq
0}^{\mathrm{Elig}}$ in (\ref{5}) will tend to $-\infty$ whereas the
$UB_{\alpha_{1}\geq\alpha_{0}\geq0}^{\mathrm{Elig}}$ in (\ref{8}) will remain
bounded. Therefore, the lower bound on welfare gain will be finite and the
upper bound infinite under $\alpha_{1}\geq\alpha_{0}\geq0$.\smallskip

\noindent\textbf{Welfare for Ineligibles (}$\alpha_{1}\geq\alpha_{0}\geq
0$\textbf{): }For subsidy-ineligibles, the mean welfare is obtained simply by
replacing $p_{1}$ by $p_{0}$ in the expressions for eligibles:

\begin{corollary}
Suppose the linear index structure described above holds with $\beta_{1}%
\geq\beta_{0}>0$, and $\alpha_{1}\geq\alpha_{0}\geq0$. Define%
\begin{align*}
D_{1}\left(  a_{1}\right)   &  :=-\int_{p_{0}+\frac{\alpha_{1}-\alpha}%
{\beta_{0}}\left(  \pi_{1}-\pi_{0}\right)  }^{p_{0}+\frac{\alpha_{1}}%
{\beta_{1}}\left(  \pi_{1}-\pi_{0}\right)  }q_{1}\left(  p,y,\pi_{0}%
+\frac{\alpha_{1}}{\alpha}\left(  \pi_{1}-\pi_{0}\right)  \right)
dp\text{,}\\
D_{2}\left(  \alpha_{1}\right)   &  :=-\int_{p_{0}+\frac{\alpha-\alpha_{1}%
}{\beta_{0}}\left(  \pi_{1}-\pi_{0}\right)  }^{p_{0}-\frac{\alpha_{1}}%
{\beta_{1}}\left(  \pi_{1}-\pi_{0}\right)  }\left[  1-q_{1}\left(
p,y+p-p_{0},\pi_{1}-\frac{\alpha_{1}}{\alpha}\left(  \pi_{1}-\pi_{0}\right)
\right)  \right]  dp\text{.}%
\end{align*}
Then, the mean welfare for the ineligible is given by%
\[
E\left[  S^{\mathrm{Inelig}}\right]  =\left\{
\begin{array}
[c]{cc}%
D_{1}\left(  \alpha_{1}\right)  \text{, } & \text{if }\alpha\leq\alpha_{1}%
\leq\frac{\beta_{1}}{\beta_{1}-\beta_{0}}\alpha\text{,}\\
D_{2}\left(  \alpha_{1}\right)  \text{,} & \text{if }\frac{\beta_{1}}%
{\beta_{1}-\beta_{0}}\alpha<\alpha_{1}<\infty\text{.}%
\end{array}
\right.
\]

\end{corollary}

From the above results, it follows that allowing for $\alpha_{1}>\alpha
_{0}\geq0$ in addition to the possibility $\alpha_{1}\geq\alpha_{0}\geq0$
widens the overall bounds for mean welfare of the ineligible from (\ref{15})
and (\ref{16}) to%
\begin{align}
LB^{\mathrm{Inelig}}  &  =\min\left\{  \int_{p_{0}-\frac{\alpha}{\beta_{0}%
}\left(  \pi_{1}-\pi_{0}\right)  }^{p_{0}}\left\{  1-q_{1}\left(  p,y,\pi
_{0}\right)  \right\}  dp,\text{ }\inf_{\alpha_{1}>\frac{\beta_{1}}{\beta
_{1}-\beta_{0}}\alpha}D_{2}\left(  \alpha_{1}\right)  ,\text{ }\inf
_{\alpha\leq\alpha_{1}\leq\frac{\beta_{1}}{\beta_{1}-\beta_{0}}\alpha}%
D_{1}\left(  \alpha_{1}\right)  \right\}  \text{,}\label{2}\\
UB^{\mathrm{Inelig}}  &  =\max\left\{  -\int_{p_{0}}^{p_{0}+\frac{\alpha
}{\beta_{1}}\left(  \pi_{1}-\pi_{0}\right)  }q_{1}\left(  p,y,\pi_{1}\right)
dp\text{, }\sup_{\alpha_{1}>\frac{\beta_{1}}{\beta_{1}-\beta_{0}}\alpha}%
D_{2}\left(  \alpha_{1}\right)  ,\text{ }\sup_{\alpha\leq\alpha_{1}\leq
\frac{\beta_{1}}{\beta_{1}-\beta_{0}}\alpha}D_{1}\left(  \alpha_{1}\right)
\right\}  \text{.} \label{19}%
\end{align}

The deadweight loss expressions are analogous to those for the case with
$\alpha_{1}\geq0\geq\alpha_{0}$ and not repeated here.

\section{STOCHASTIC ENVIRONMENT AND\ EQUILIBRIUM
BELIEFS\label{sec: Stochastic Environment}}

\noindent\textbf{Incomplete-Information Setting: }In this section, we
formulate interactions of households as an incomplete-information Bayesian
game, whose stochastic structure will be laid out below. In each village $v$,
each of the $N_{v}$ households is provided the opportunity to buy the product
at a researcher-specified price $P_{vh}$ randomly varied across households.
They have incomplete information in that each household $(v,h)$ knows her own
variables $(A_{vh},W_{vh},\boldsymbol{\eta}_{vh})$ but does not know the
values of all the variables $W_{\tilde{v}k},\boldsymbol{\eta}_{\tilde{v}%
k},A_{\tilde{v}k}$ for every household $k\neq h$ selected in the experiment.

We assume households have `consistent beliefs' in accordance with the standard
Bayes-Nash setting, i.e., each $(v,h)$'s belief is formed as%
\begin{equation}
\Pi_{vh}=\dfrac{1}{N_{v}-1}%
%TCIMACRO{\dsum \nolimits_{1\leq k\leq N_{v}\text{; }k\neq h}}%
%BeginExpansion
{\displaystyle\sum\nolimits_{1\leq k\leq N_{v}\text{; }k\neq h}}
%EndExpansion
E[A_{vk}|\mathcal{I}_{vh}], \label{eq: Belief-general}%
\end{equation}
where $A_{vk}$ is given in (\ref{eq: outcome-variable}) and $E\left[
\cdot\text{ }|\mathcal{I}_{vh}\right]  $ is the conditional expectation
computed through the probability law that governs all the relevant variables
given $(v,h)$'s information set $\mathcal{I}_{vh}$ that includes
$(W_{vh},\boldsymbol{\eta}_{vh})$. The explicit form of
(\ref{eq: Belief-general}) \emph{in equilibrium} is investigated in the next subsection.

Each household $(v,h)$ is solely concerned with behavior of other households
in the same village $v$. Thus the econometrician observes $\bar{v}$ games
($\bar{v}=11$ in our application), each with `many' households. To formalize
our model as a Bayesian game, given the form of (\ref{eq: Belief-general}),
$U_{1}$ and $U_{0}$ are to be interpreted as expected utilities. This is
possible when the underlying von Neumann-Morgenstern utility indices $u_{1}$
and $u_{0}$ satisfy%
\[
U_{1}\left(  Y_{vh}-P_{vh},\Pi_{vh},\boldsymbol{\eta}_{vh}\right)
=E[u_{1}(Y_{vh}-P_{vh},\tfrac{1}{N_{v}-1}%
%TCIMACRO{\dsum \nolimits_{1\leq k\leq N_{v}\text{; }k\neq h}}%
%BeginExpansion
{\displaystyle\sum\nolimits_{1\leq k\leq N_{v}\text{; }k\neq h}}
%EndExpansion
A_{vk},\boldsymbol{\eta}_{vh})|\mathcal{I}_{vh}]\text{,}%
\]
i.e., $u_{1}$ is linear in the second argument; $U_{0}$ and $u_{0}$ satisfy an
analogous relationship. This will hold in particular when utilities have a
linear index structure as in Manski (1993) and Brock and Durlauf (2001, 2007).
We have already presented our linear specifications of $U_{1}$ and $U_{0}$ in
(\ref{A}), but these are further elaborated below in this section and in
Section \ref{sec: Estimators}.\smallskip

\noindent\textbf{Unobserved Heterogeneity: }We assume that unobserved
heterogeneity $\{\boldsymbol{\eta}_{vh}\}_{v=1}^{N_{v}}$ ($v=1,\dots\bar{v}$)
takes the following form:
\begin{equation}
\boldsymbol{\eta}_{vh}=\boldsymbol{\xi}_{v}+\boldsymbol{u}_{vh}\text{,}
\label{eq: fixed-effect}%
\end{equation}
where $\boldsymbol{\xi}_{v}$ stands for a village-specific vector of variables
that are common to all members in the $v$th village and $\boldsymbol{u}_{vh}$
represents an individual specific variable. Let $\left(  d_{v},e_{v}\right)  $
be an underlying vector of village-specific variables such that $d_{v}$ is a
common factor affecting both the unobservable $\boldsymbol{\xi}_{v}$ and the
observable covariates $W_{vh}$, and $e_{v}$ affects only $\boldsymbol{\xi}%
_{v}$ with $\boldsymbol{\xi}_{v}$ fully determined by $\left(  d_{v}%
,e_{v}\right)  $, i.e., $\boldsymbol{\xi}_{v}=\boldsymbol{\xi}\left(
d_{v},e_{v}\right)  $.\footnote{The need to separate $d_{v}$ and $e_{v}$ will
become clear below in the context of identification of model parameters in
presence of unobserved group-effects.} Each household in village $v$ is
assumed to know $\left(  d_{v},e_{v}\right)  $, the functional form
$\boldsymbol{\xi}\left(  \cdot\right)  $, and thus $\boldsymbol{\xi}_{v}$,
while $\boldsymbol{u}_{vh}$ is a purely private variable known only to
individual $(v,h)$. None of $\left\{  \left(  d_{v},e_{v}\right)  \right\}  $,
$\left\{  \boldsymbol{\xi}_{v}\right\}  $, and $\left\{  \boldsymbol{u}%
_{vh}\right\}  $ is observable to the econometrician. Denote household
$(v,h)$'s information set by%
\begin{equation}
\mathcal{I}_{vh}=(W_{vh},\boldsymbol{u}_{vh},d_{v},e_{v})\text{.}
\label{eq: info-set}%
\end{equation}

We now impose the following conditions on the probabilistic law for these variables:

\begin{description}
\item[C1] $\{(W_{vh},\boldsymbol{u}_{vh},d_{v},e_{v})\}_{h=1}^{N_{v}}$,
$v=1,\dots,\bar{v}$, are independent across $v$.
\end{description}

Assumption \textbf{C1} says that variables in village $v$ are independent of
those in village $\tilde{v}(\neq v)$.

\begin{description}
\item[C2] (i)\ For each $v$, the sequence $\left\{  (W_{vh},\boldsymbol{u}%
_{vh})\right\}  _{h=1}^{N_{v}}$ is I.I.D. conditionally on $\left(
d_{v},e_{v}\right)  $. (ii) $\left\{  \boldsymbol{u}_{vh}\right\}
_{h=1}^{N_{v}}$ is independent of $\left\{  W_{vh}\right\}  _{h=1}^{N_{v}}$
conditionally on $\left(  d_{v},e_{v}\right)  $.
\end{description}

The conditional I.I.D.-ness imposed in \textbf{C2 (i)} leads to
equi-dependence within each village, i.e., $\mathrm{Cov}\left[
\boldsymbol{\eta}_{vh},\boldsymbol{\eta}_{vk}\right]  =\mathrm{Cov}\left[
\boldsymbol{\eta}_{v\tilde{h}},\boldsymbol{\eta}_{v\tilde{k}}\right]  (\neq0)$
for any $h\neq k$ and $\tilde{h}\neq\tilde{k}$. Further, each household
$(v,h)$'s unobservable $\boldsymbol{u}_{vh}$ is not useful for predicting
another household $(v,k)$'s variables and behavior, and therefore her belief
$\Pi_{vh}$ (\ref{eq: Belief-general}) is reduced to the average of the
\textit{unconditional }expectations (as formally shown in Proposition
\ref{Prop: belief-constancy}) below. This condition rules out spatial
correlation in unobservables which, if present, would complicate the analysis
in a non-trivial way by making a household's belief a function of its
privately known variables.

\textbf{C2 (ii)} is the exogeneity condition. This allows for identification
and consistent estimation of model parameters. In the context of the field
experiment in our empirical exercise, this exogeneity condition can be
interpreted as saying that realization of unobserved heterogeneity is
independent of how researchers have selected the sample. Note that the
exogeneity condition is \textit{conditional} on $\left(  d_{v},e_{v}\right)
$, and it does \textit{not} exclude correlation of $\boldsymbol{u}_{vh}$ and
$W_{vh}=\left(  P_{vh},Y_{vh}\right)  $ in the \textit{unconditional} sense.
In our application, prices $P_{vh}$ are randomly assigned to individuals by
researchers and thus $P_{vh}$ and $\boldsymbol{u}_{vh}$ are independent both
unconditionally and conditionally.

Note that under the (\ref{eq: fixed-effect}) introduced later, we compute the
ASF (\ref{7}) using the marginal distribution of $\boldsymbol{u}_{vh}$
conditionally on $\boldsymbol{\xi}_{v}$ in later sections (see also Footnote 3).

\subsection{Equilibrium Beliefs\label{sec: eqm-beliefs}}

We now investigate the forms of households' beliefs defined in
(\ref{eq: Belief-general}). We show that under \textbf{C2}, the high-level
assumption in BD01 that beliefs, corresponding to our $\Pi_{vh}$, are constant
and symmetric across all households in the same village can be formalized in
our incomplete-information game setting via the specification of a Bayes-Nash equilibrium.

\begin{proposition}
\label{Prop: belief-constancy}Suppose that Conditions \textbf{C1} and
\textbf{C2} are common knowledge in the Bayesian game described above. Then,
for any $k\neq h$ in village $v$ with $\left(  d_{v},e_{v}\right)  $,
\[
E[A_{vk}|\mathcal{I}_{vh}]=E[A_{vk}|d_{v},e_{v}],
\]
where the information set $\mathcal{I}_{vh}$ is defined in (\ref{eq: info-set}).
\end{proposition}

The proof of Proposition \ref{Prop: belief-constancy} is provided in Appendix
\ref{sec: Proof-IID}. Note that this proposition does not utilize any
\textit{equilibrium} condition. It simply confirms, formally, the intuitive
statement that $(v,h)$'s own variables are not useful to predict other
$(v,k)$'s behavior $A_{vk}$. Given this result, we can write the belief
$\Pi_{vh}$ (defined in (\ref{eq: Belief-general})) as%
\begin{equation}
\Pi_{vh}=\bar{\Pi}_{vh}\text{,} \label{eq: pi-bar-average}%
\end{equation}
where
\[
\bar{\Pi}_{vh}=\bar{\Pi}_{vh}(d_{v},e_{v}):=\tfrac{1}{N_{v}-1}%
%TCIMACRO{\dsum \nolimits_{1\leq k\leq N_{v}\text{; }k\neq h}}%
%BeginExpansion
{\displaystyle\sum\nolimits_{1\leq k\leq N_{v}\text{; }k\neq h}}
%EndExpansion
E[A_{vk}|d_{v},e_{v}],
\]
and $\bar{\Pi}_{vh}$ is a function of $\left(  d_{v},e_{v}\right)  $ and
independent of $(v,h)$-specific variables, $(W_{vh},\boldsymbol{u}_{vh})$; for
notational simplicity, we suppress the dependence of $\bar{\Pi}_{vh}$ on
$\left(  d_{v},e_{v}\right)  $ from now on.

Beliefs \textit{in equilibrium} solve the system of $N_{v}$ equations:%
\begin{equation}
\bar{\Pi}_{vh}=\tfrac{1}{N_{v}-1}%
%TCIMACRO{\dsum \nolimits_{1\leq k\leq N_{v}\text{; }k\neq h}}%
%BeginExpansion
{\displaystyle\sum\nolimits_{1\leq k\leq N_{v}\text{; }k\neq h}}
%EndExpansion
E_{v}\left[  1\left\{
\begin{array}
[c]{c}%
U_{1}(Y_{vk}-P_{vk},\bar{\Pi}_{vk},\boldsymbol{\eta}_{vk})\\
\geq U_{0}(y,\bar{\Pi}_{vk},\boldsymbol{\eta}_{vk})
\end{array}
\right\}  \right]  ,\text{ \ }h=1,\dots,N_{v}, \label{eq: solution-BigPi-IID}%
\end{equation}
where $E_{v}\left[  \cdot\right]  $ denotes the conditional expectation
operator given $\left(  d_{v},e_{v}\right)  $ (i.e., $E\left[  \cdot
|d_{v},e_{v}\right]  $). BD01 focus on equilibria with constant and symmetric
beliefs. Using our notation above, we say that (constant) beliefs are
symmetric when $\bar{\Pi}_{vh}=\bar{\Pi}_{vk}$ for any $h,k\in\{1,\dots
,N_{v}\}$ (for each $v$). When Brock and Durlauf's framework is interpreted as
a Bayesian game, one can justify their focus on constant and symmetric beliefs
under conditions laid out in Proposition \ref{prop: belief-symmetricity} below.

To establish this proposition, define for each $v$, given $\left(  d_{v}%
,e_{v}\right)  $, a function $m_{v}:\left[  0,1\right]  \rightarrow
\lbrack0,1]$ as
\begin{equation}
m_{v}\left(  r\right)  :=E_{v}\left[  1\left\{  U_{1}(Y_{vh}-P_{vh}%
,r,\boldsymbol{\xi}_{v}+\boldsymbol{u}_{vh})\geq U_{0}\left(  Y_{vh}%
,r,\boldsymbol{\xi}_{v}+\boldsymbol{u}_{vh}\right)  \right\}  \right]  ;
\label{eq: m-function}%
\end{equation}
note that $m_{v}\left(  r\right)  $ is independent of individual index $h$
under the conditional I.I.D. assumption given $\left(  d_{v},e_{v}\right)  $.
Then the following characterization of beliefs holds:

\begin{proposition}
\label{prop: belief-symmetricity}Suppose that the same conditions hold as in
Proposition \ref{Prop: belief-constancy} and the function $m_{v}\left(
\cdot\right)  $ defined in (\ref{eq: m-function}) is a contraction, i.e., for
some $\rho\in\left(  0,1\right)  $,
\begin{equation}
|m_{v}\left(  r\right)  -m_{v}\left(  \tilde{r}\right)  |\leq\rho|r-\tilde
{r}|\text{ \ for any }r,\tilde{r}\in\left[  0,1\right]  .
\label{eq: contraction-scalar}%
\end{equation}
Then, a solution $(\bar{\Pi}_{v1},\dots,\bar{\Pi}_{vN_{v}})$ of the system of
$N_{v}$ equations in (\ref{eq: solution-BigPi-IID}) uniquely exists and is
given by symmetric beliefs, i.e.,%
\[
\bar{\Pi}_{vh}=\bar{\Pi}_{vk}\text{ \ for any }h,k\in\{1,\dots,N_{v}\}.
\]

\end{proposition}

The proof is given in Appendix \ref{sec: Proof-IID}. Propositions
\ref{Prop: belief-constancy}-\ref{prop: belief-symmetricity} show that, given
the (conditional) I.I.D. and contraction conditions, the equilibrium is
characterized through%
\[
\Pi_{vh}=\bar{\pi}_{v}\text{ \ for any }h=1,\dots,N_{v},
\]
for some constant $\bar{\pi}_{v}:=\bar{\pi}_{v}(d_{v},e_{v})\in\left[
0,1\right]  $ within each village (given $\left(  d_{v},e_{v}\right)  $). This
implies that the \textit{beliefs can be consistently estimated by the sample
average of }$A_{vk}$\textit{ over village }$v$, which is exploited in our
empirical study.

The contraction condition (\ref{eq: contraction-scalar}) holds when the social
interactions coefficient $\alpha$ is not large (in our linear index
specification). In Section \ref{sec: unique-multi} below, we will provide
sufficient conditions for the contraction and equilibrium uniqueness, as well
as explain additional procedures that are needed for estimation and
counterfactual analysis when multiplicity of equilibria may arise.

\section{ECONOMETRIC SPECIFICATION, IDENTIFICATION AND
ESTIMATION\label{sec: Estimators}}

Taking the belief variable $\Pi_{vh}$ in the linear-index choice probability
function $q_{1}\left(  \cdot\right)  $ to be the (limit of) observed fraction
of usage in each village i.e. $\bar{\pi}_{v}=E_{v}\left[  A_{vh}\right]
(=\lim_{N_{v}\rightarrow\infty}\sum_{h=1}^{N_{v}}A_{vh}/N_{v})$ as justified
in Propositions \ref{Prop: belief-constancy}-\ref{prop: belief-symmetricity},
the index coefficients can be estimated semiparametrically using say,
Bhattacharya (2008). However, unobserved village-effects may confound the
consistency of these estimates; we overcome this by using a correlated random
effects (CRE, henceforth) probit approach to estimate $q_{1}\left(
\cdot\right)  $, which is derived from a factor structure on the covariates
and the village-effects, as follows.

\subsection{Village Effects Specification}

Our data for the application come from eleven different villages with an
average of $195$ households per village. It is plausible that utilities from
using and from not using an ITN\ are affected by village-specific unobservable
characteristics (i.e., $\xi_{v}=\xi_{v}^{1}-\xi_{v}^{0}$ introduced in
(\ref{eq: linear-epsilon})), such as the chance of contracting malaria when
not using an ITN. Recall the linear utility structure (\ref{A}) from Section
\ref{sec: Welfare}. Given this, together with the unobserved heterogeneity
specification in (\ref{eq: fixed-effect}),\text{ }$\boldsymbol{\eta}$%
$_{vh}=\boldsymbol{\xi}_{vh}+\boldsymbol{u}_{vh}$, we model%
\begin{equation}%
\begin{array}
[c]{c}%
U_{0}(Y_{vh},\Pi_{vh},\boldsymbol{\eta}_{vh})=\delta_{0}+\beta_{0}%
Y_{vh}+\alpha_{0}\Pi_{vh}+\underset{\eta_{vh}^{0}}{\underbrace{\xi_{v}%
^{0}+u_{vh}^{0}}}\text{,}\\
U_{1}(Y_{vh}-P_{vh},\Pi_{vh},\boldsymbol{\eta}_{vh}\boldsymbol{)}=\delta
_{1}+\beta_{1}\left(  Y_{vh}-P_{vh}\right)  +\alpha_{1}\Pi_{vh}+\underset{\eta
_{vh}^{1}}{\underbrace{\xi_{v}^{1}+u_{vh}^{1}}}\text{,}%
\end{array}
\label{eq: A2-linear}%
\end{equation}
where $\boldsymbol{\xi}_{v}=\left(  \xi_{v}^{0},\xi_{v}^{1}\right)  $ and
$\boldsymbol{u}_{vh}=\left(  u_{vh}^{0},u_{vh}^{1}\right)  $ denote village
and individual specific characteristics, respectively, both of which are
unobservable. Therefore,
\begin{align}
&  U_{1}\left(  Y_{vh}-P_{vh},\Pi_{vh},\boldsymbol{\eta}_{vh}\right)
-U_{0}\left(  Y_{vh},\Pi_{vh},\boldsymbol{\eta}_{vh}\right) \nonumber\\
&  =\left(  \delta_{1}-\delta_{0}\right)  -\beta_{1}P_{vh}+\left(  \beta
_{1}-\beta_{0}\right)  Y_{vh}+\left(  \alpha_{1}-\alpha_{0}\right)  \Pi
_{vh}+\underset{\xi_{v}}{\underbrace{\xi_{v}^{1}-\xi_{v}^{0}}}%
+\underset{\varepsilon_{vh}}{\underbrace{\left(  u_{vh}^{1}-u_{vh}^{0}\right)
}}\nonumber\\
&  \equiv c_{0}+c_{1}P_{vh}+c_{2}Y_{vh}+\alpha\Pi_{vh}+\xi_{v}+\varepsilon
_{vh}\text{,} \label{eq: linear-epsilon}%
\end{align}
where $\varepsilon_{vh}$ is assumed to have zero mean and unit variance for
scale and location normalization.\smallskip

\noindent\textbf{Non-identification of the village effects }$\xi_{v}%
$\textbf{:} Brock and Durlauf (2007) discussed difficulties of estimating
social interactions models in presence of group-specific unobservables and
presented a non-identification result (their Proposition 2). To see this in
our context, consider constant beliefs, $\Pi_{vh}=\bar{\pi}_{v}$ (justified in
Propositions \ref{Prop: belief-constancy}-\ref{prop: belief-symmetricity}).
Since $\xi_{v}$ is village specific and many observations per village are
available, we can estimate village specific intercepts $\gamma_{v}$ by
regression of take-up $A_{vh}$ on price and income $W_{vh}=\left(
P_{vh},Y_{vh}\right)  $ that vary across households $h$ within village $v$,
together with village dummies, i.e.,%
\begin{equation}
\Pr\left(  A_{vh}=1|W_{vh}=w;d_{v},e_{v}\right)  =F_{\varepsilon
}(w\boldsymbol{c}^{\prime}+\underset{\gamma_{v}}{\underbrace{c_{0}+\alpha
\bar{\pi}_{v}+\xi_{v}}})\text{,} \label{without}%
\end{equation}
where the left-hand side (LHS) is computed under the conditional law given
$\left(  d_{v},e_{v}\right)  $, and $F_{\varepsilon}\left(  \cdot\right)  $ is
the CDF of $-\varepsilon_{vh}$.\footnote{Recall that $\varepsilon_{vh}\left(
=u_{vh}^{1}-u_{vh}^{0}\right)  $ is assumed to be independent of $W_{vh}$
conditionally on $\left(  d_{v},e_{v}\right)  $ in \textbf{C2}; and $\xi_{v}$
is determined by $\left(  d_{v},e_{v}\right)  $. Below, it is further assumed
that $\varepsilon_{vh}$ is jointly independent of $W_{vh}$ and $\left(
d_{v},e_{v}\right)  $.}

The realized $\xi_{v}$ is a constant within each village; thus, $\xi_{1}%
,\dots,\xi_{\bar{v}}$ and the universal constant $c_{0}$ cannot be separately
identified and we reparametrize $\bar{\xi}_{v}:=c_{0}+\xi_{v}$. For each $v$
and each realized $\xi_{v}$, the LHS of (\ref{without}) is identifiable as a
function of $w$; thus, under a parametric specification of $F_{\varepsilon
}\left(  \cdot\right)  $ together with the exogeneity condition \textbf{C2}
(ii) and a rank condition for covariates (stated below), $\left(
\boldsymbol{c},\gamma_{1},\dots,r_{\bar{v}}\right)  $ is also identified. The
identified coefficients $\gamma_{1},\dots,\gamma_{\bar{v}}$ on the village
dummies therefore satisfy the equations $\gamma_{v}=\alpha\pi_{v}+\bar{\xi
}_{v}\equiv c_{0}+\xi_{v}$ ($v=1,\dots,\bar{v}$).\footnote{In the application,
we a run a probit of $A_{vh}$ on covariates $W_{vh}$ ($P_{vh}$, $Y_{vh}$, and
other variables) and a dummy for each village which corresponds to converting
these conditional moments to a set of unconditional ones.} However, even in
the reparametrized equations, there are as many $\bar{\xi}_{v}$ as there are
$\gamma_{v}$, so that we have $\bar{v}$ equations with $\bar{v}+1$ unknowns
$\bar{\xi}_{1},\dots,\bar{\xi}_{\bar{v}}$, and $\alpha$, which are needed for
policy and counterfactual analysis but cannot be separately identified.

\subsection{Factor Structure and Correlated Random Effects Modelling}

We surmount non-identification of $\xi_{v}$ by an approximate version of the
Mundlak-Chamberlain correlated random effects (CRE) structure, cf. Section
15.8.2 of Wooldridge (2010), which is routinely used as a reasonable middle
ground between fixed and random effects in the panel econometrics literature.
While the CRE device is typically intended for short panels, our setting here
may be seen like a \textquotedblleft long panel\textquotedblright\ in that
each village is supposed to have its own effect that is shared by a large
number of households (note that our dataset does not have a panel structure
but consists of several cross sectional datasets). To have our specification
consistent with the long-panel-like\ setting and Section
\ref{sec: Stochastic Environment}\ (in particular, \textbf{C2}), we consider
the following factor structure for the observable covariate $W_{vh}$ and the
village specific variable $\xi_{v}$,%
\begin{equation}
W_{vh}=d_{v}+\tau_{vh}\text{ \ and \ }\xi_{v}=d_{v}\boldsymbol{\delta}%
^{\prime}+e_{v}\text{,} \label{eq: factor}%
\end{equation}
where $d_{v}$ is a vector of \textquotedblleft factor\textquotedblright%
\ variables (with the same dimension as $W_{vh}$) that are common in $W_{vh}$
and $\xi_{v}$, $\tau_{vh}$ is the covariate specific, idiosyncratic component
that is assumed to have zero mean (for location normalization) and is defined
through $\tau_{vh}:=W_{vh}-d_{v}$\textbf{, }$\boldsymbol{\delta}$ is a (row)
vector of constant coefficients on the factor, and $e_{v}$ is a village
specific variable that affects only $\xi_{v}$.\footnote{Note that one
component $W_{vh}$ is the price $P_{vh}$ faced by the household, which is
randomized across households. The corresponding component of $d_{v}$ in
(\ref{eq: factor}) is the average price within the village and its coefficient
in $\boldsymbol{\delta}$\ is set as zero in our application (as the randomized
price does not capture village specific features).}

We assume each household in village $v$ knows the realization of $\left(
d_{v},e_{v}\right)  $, while all the right-hand-side components in
(\ref{eq: factor}) are unobservable to researchers. Let $\bar{W}_{v}:=\left(
1/N_{v}\right)  \sum_{v=1}^{N_{v}}W_{vh}$. Then, we can write $d_{v}=\bar
{W}_{v}-\left(  1/N_{v}\right)  \sum\nolimits_{v=1}^{N_{v}}\tau_{vh}$.
Plugging this into the second equation in (\ref{eq: factor}), we can write%
\begin{equation}
\xi_{v}=\bar{W}_{v}\boldsymbol{\delta}^{\prime}+e_{v}+o_{p}\left(  1\right)
\text{,} \label{eq: xi-CR}%
\end{equation}
for each $N_{v}$, which follows from $\left(  1/N_{v}\right)  \sum
\nolimits_{v=1}^{N_{v}}\tau_{vh}=O_{p}(1/\sqrt{N_{v}})=o_{p}(1)$ by a standard
central limit theorem. We note that (\ref{eq: xi-CR}) is a reduced-form
representation for each (sufficiently large) $N_{v}$ derived from the
structural assumption (\ref{eq: factor}). We further assume that the error
term satisfies
\begin{equation}
e_{v}\perp\left(  \{(W_{vh},\varepsilon_{vh})\}_{h=1}^{N_{v}},d_{v}\right)
\text{ \ and \ }e_{v}\sim N\left(  0,(\sigma_{e}^{\ast})^{2}\right)  \text{,}
\label{eq: xi-CR-normal}%
\end{equation}
for each $v$, where we note that $\left\{  e_{v}\right\}  _{v=1}^{\bar{v}}$ is
I.I.D. under \textbf{C1} and (\ref{eq: xi-CR-normal}) (we denote by
$\sigma_{e}^{\ast}$ the \emph{true} standard deviation parameter; and
subsequently, $\ast$ is often used to denote true parameters).\footnote{Brock
and Durlauf (2007) have also considered a (linear) restriction on the group
effects similar to (\ref{eq: xi-CR}) (see their Section 4.1.2 and Assumption
L.1) and argue that it may help partial identification.} In standard
short-panel cases, a distributional assumption is directly imposed on group
effects, say, $\xi_{v}|\left\{  W_{vh}\right\}  _{h=1}^{N_{v}}\sim N(\bar
{W}_{v}\boldsymbol{\delta}^{\prime},(\sigma_{e}^{\ast})^{2})$ (see Wooldridge,
2010, p. 615); in our setting, this conditional normality of $\xi_{v}$ holds
in an approximate sense with a small order $o_{p}\left(  1\right)  $ term in
(\ref{eq: xi-CR}).\footnote{The original CRE model, the so-called
Mundlak-Chamberlain device, is not derived from a factor structure as in
(\ref{eq: factor}); we do not know of any other paper that considers a CRE
model as a reduced form derived from some factor structure, which can be
thought of as a separate contribution of the present paper. Our derivation of
the approximate CRE model makes the households' information structure
transparent which is required for constructing an econometric framework
consistent with the game structure and \textbf{C2}. If we directly imposed
(\ref{eq: xi-CR}) as is done in standard short panel contexts, it would be
difficult to see which parts of $\xi_{v}$ should be known to households and to
interpret the conditional i.i.d.-ness in \textbf{C2 }given the village
specific variables.} We further assume that
\begin{equation}
\varepsilon_{vh}\left\vert \left(  \{W_{vh}\}_{h=1}^{N_{v}},d_{v}\right)
\right.  \sim N\left(  0,1\right)  , \label{eq: normal-epsilon}%
\end{equation}
which is analogous to specifications in Chamberlain (1980) and Wooldridge
(2010). Putting all of this together, we can write%
\[
A_{vh}=1\{W_{vh}\boldsymbol{c}^{\prime}+c_{0}+\alpha\bar{\pi}_{v}%
+d_{v}\boldsymbol{\delta}^{\prime}+e_{v}+\varepsilon_{vh}\geq0\}\text{ \ for
each }(v,h)
\]
and compute the conditional probability as%
\begin{align}
\Pr\left(  A_{vh}=1|W_{vh}=w;d_{v}\right)   &  =\int F_{\varepsilon}\left(
w\boldsymbol{c}^{\prime}+c_{0}+\alpha\bar{\pi}_{v}+d_{v}\boldsymbol{\delta
}^{\prime}+e\right)  \phi_{\sigma_{e}^{2}}\left(  e\right)  de\nonumber\\
&  =F_{\varepsilon+e}\left(  w\boldsymbol{c}^{\prime}+c_{0}+\alpha\bar{\pi
}_{v}+d_{v}\boldsymbol{\delta}^{\prime}\right) \nonumber\\
&  =F_{\varepsilon+e}\left(  w\boldsymbol{c}^{\prime}+c_{0}+\alpha\bar{\pi
}_{v}+\bar{W}_{v}\boldsymbol{\delta}^{\prime}\right)  +o_{p}\left(  1\right)
\text{,} \label{with}%
\end{align}
where the probability on the LHS is computed under the conditional law given
$d_{v}$ (i.e., it is with respect to the distribution of $-\left(
\varepsilon_{vh}+e_{v}\right)  $); $F_{\varepsilon+e}$ is the CDF of
$-(\varepsilon_{vh}+e_{v})\sim N(0,1+\sigma_{e}^{2})$, and last equality holds
since $d_{v}=\bar{W}_{v}+o_{p}\left(  1\right)  $; and (\ref{with}) can be
shown to hold uniformly over $\left(  v,h\right)  $, $w$, and $\left(
\boldsymbol{c},c_{0},\alpha,\boldsymbol{\delta}\right)  $, under compactness
of the parameter space and Condition \textbf{CR2 (ii)}, imposed in the next
subsection. Denote by $\Phi$ the CDF of $N\left(  0,1\right)  $. Then, the
leading term on the right-hand side (RHS) of (\ref{with}) can be written as%
\begin{equation}
\Phi\left(  \frac{w\boldsymbol{c}^{\prime}+c_{0}+\alpha\bar{\pi}_{v}+\bar
{W}_{v}\boldsymbol{\delta}^{\prime}}{\sqrt{1+(\sigma_{e}^{\ast})^{2}}}\right)
\equiv\Phi(w\boldsymbol{\bar{c}}^{\prime}+\bar{c}_{0}+\bar{\alpha}\bar{\pi
}_{v}+\bar{W}_{v}\boldsymbol{\bar{\delta}}^{\prime})\text{.}
\label{with-approx}%
\end{equation}

For calculating the LHS of (\ref{without}), $e_{v}$ is treated as a part of
the parameter $\gamma_{v}$; in contrast, the LHS of (\ref{with}) is calculated
with respect to the distribution of the unobservable $\varepsilon_{vh}+e_{v}$
across all households over all villages. Both the probabilities in
(\ref{without}) and (\ref{with}) concern the same outcome variable $A_{vh}$
but they differ in conditioning variables. The former probability can be
consistently estimated within each village as $N_{v}\rightarrow\infty$ for
each $v$, while consistent estimation \ of the latter requires $\bar
{v}\rightarrow\infty$ (in addition to $N_{v}\rightarrow\infty$) since village
specific effects $e_{v}$ have to be averaged out to match the probability
(\ref{with}) computed as the integral of $e_{v}$ via its approximation
(\ref{with-approx}).

Putting all this together, our estimation steps are as follows:

\begin{enumerate}
\item First run a probit of $A_{vh}$ on $W_{vh},$ $\bar{W}_{v}$ and $\hat{\pi
}_{v}\equiv\left(  1/N_{v}\right)  \sum_{v=1}^{N_{v}}A_{vh}$ corresponding to
(\ref{with-approx}) to obtain estimates of $\boldsymbol{\bar{c}},\bar{c}%
_{0},\bar{\alpha},\boldsymbol{\bar{\delta}}$ ;

\item Then run a probit of $A_{vh}$ on $W_{vh}$ and village dummies
corresponding to (\ref{without}) and obtain estimates $\boldsymbol{c}^{\prime
},\gamma_{1},\gamma_{2},...\gamma_{\bar{v}}$;

\item Estimate $\sigma_{e}^{\ast}$ by the ratio of the price coefficient in
the former to that in the latter probit;

\item Estimate $c_{0}$ via $\bar{c}_{0}\times\sqrt{1+(\sigma_{e}^{\ast})^{2}}$
and $\alpha$ via $\bar{\alpha}\times\sqrt{1+(\sigma_{e}^{\ast})^{2}}$;

\item From (\ref{without}), estimate $\xi_{v}=\gamma_{v}-$ $c_{0}-\alpha
\hat{\pi}_{v}$.
\end{enumerate}

These are all the quantities we need for empirical calculation of welfare
expressions outlined in Section 3. In the empirical application below, the
parameters $\boldsymbol{\bar{c}},\bar{c}_{0},\bar{\alpha},\boldsymbol{\bar
{\delta}}$ are estimated via pseudo-MLE by running an ordinary probit
regression of $A_{vh}$ on $W_{vh},\bar{\pi}_{v}$ and $\bar{W}_{v}$.

Thus to summarize, it follows from Brock and Durlauf's (2007) arguments,
outlined above, that identification of village specific parameters, $\xi_{v}$,
is in general impossible in the presence of social interaction effects. We
overcome this through our CRE condition (\ref{eq: xi-CR}) which imposes more
structure on $\xi_{v}$ and letting the number of groups, i.e. $\bar
{v}\rightarrow\infty$, as formally stated in the next subsection and the proof
of consistency in Appendix \ref{sec: consistency-proof}. As such, this is a
new finding for social-interactions models. Note that if $e_{v}$ is
non-stochastic (i.e., $\sigma_{e}^{\ast}=0$ and $\xi_{v}=\bar{W}%
_{v}\boldsymbol{\delta}^{\prime}+o_{p}\left(  1\right)  $, instead of
(\ref{eq: xi-CR})), the above scheme using two probit regressions leads to
identification and consistent estimation without the many-village assumption
of $\bar{v}\rightarrow\infty$.\footnote{This identification/estimation scheme
of CRE models using two probit regressions appears new, which allows us to
recover the standard deviation $\sigma_{e}^{\ast}$ (which is not typically
identified in standard short-panel cases; see e.g. p. 617 of Wooldridge, 2010)
and further all the realized values of $e_{1},\dots,e_{\bar{v}}$.}

\subsection{Estimation and Consistency\label{sec: consistency}}

Now we discuss consistency of the estimation procedure outlined in the
previous subsection. We focus on the consistency of the first probit
(\ref{without}), the setting of which is non-standard under the CRE structure
and the many-village asymptotics $\bar{v}\rightarrow\infty$; in contrast, the
setting of the second probit (\ref{with}) or (\ref{with-approx}) can be
analyzed in the same way as in Hahn and Kuersteiner (2011), and a detailed
discussion of its consistency is omitted.\footnote{Our second probit setting
is even simpler than Hahn and Kuersteiner's in that the number of parameters
do not increase $N_{v}\rightarrow\infty$ or $\bar{v}\rightarrow\infty$. A
notable difference is that the objective function $\hat{R}$ incurs some
approximation error $o_{p}\left(  1\right)  $ by using (\ref{with-approx})
instead of (\ref{with}); but given the uniformity of the $o_{p}\left(
1\right)  $ as stated, this error can be negligible for the consistency
discussion.}

For verification of consistency, we assume that the number of households in
each village can be written as%
\begin{equation}
N_{v}=r_{v}N_{0}, \label{eq: Nv-assumption}%
\end{equation}
where $r_{v}\in\left(  \underline{r},\bar{r}\right)  $ is a constant that is
independent of $N_{0}$ and $\bar{v}$ with $0<\underline{r}\leq\bar{r}<\infty$
(i.e., $r_{v}$ is uniformly bounded from below and above), and let
$N=\sum_{v=1}^{\bar{v}}N_{v}$ is the total number of households in all
villages combined. This assumption means that all $N_{1},\dots,N_{\bar{v}}$
grow at the same rate, so that none of villages is asymptotically negligible.

Comparing the two probabilities in (\ref{with-approx}), we can
identify/estimate all the parameters $\gamma_{v}^{\ast}$, $\left(
\boldsymbol{c}^{\ast},c_{0}^{\ast},\alpha^{\ast},\boldsymbol{\delta}^{\ast
}\right)  $, and $\sigma_{e}^{\ast}$, which allows us to obtain estimates of
$\xi_{1},\dots\xi_{\bar{v}}$. Consistent estimation of these parameters can be
achieved through the following two probit regressions.\footnote{Note that our
practical estimation procedure exploits the fixed point restriction by an
iteration process (discussed in Section \ref{sec: unique-multi}). It is
slightly more complicated than the procedure outlined here; but the substance
of our identification arguments does not change between the two procedures;
our exposition here is based on the simpler procedure.} First, a probit of
$A_{vh}$ \ on $W_{vh}$ and village dummies allows us to obtain estimates
\[
\left(  \boldsymbol{\hat{c}},\hat{\gamma}_{1},\dots.,\hat{\gamma}_{\bar{v}%
}\right)  =\underset{\boldsymbol{c}\in \Upsilon_{1};\text{ }\left(  \gamma
_{1},\dots,\gamma_{\bar{v}}\right)  \in \Upsilon\left(  \bar{v}\right)
\times\dots\times \Upsilon\left(  \bar{v}\right)  }{\operatorname{argmax}}%
\hat{Q}\left(  \boldsymbol{c},\gamma_{1},\dots.,\gamma_{\bar{v}}\right)  ,
\]
where the objective function $\hat{Q}$ is%
\begin{equation}
\hat{Q}\left(  \boldsymbol{c},\gamma_{1},\dots.,\gamma_{\bar{v}}\right)
=\dfrac{1}{N}%
%TCIMACRO{\dsum \nolimits_{v=1}^{\bar{v}}}%
%BeginExpansion
{\displaystyle\sum\nolimits_{v=1}^{\bar{v}}}
%EndExpansion%
%TCIMACRO{\dsum \nolimits_{h=1}^{N_{v}}}%
%BeginExpansion
{\displaystyle\sum\nolimits_{h=1}^{N_{v}}}
%EndExpansion
\mathcal{L}_{vh}\left(  \boldsymbol{c},\gamma_{v}\right)  , \label{def: Qhat}%
\end{equation}%
\begin{equation}
\mathcal{L}_{vh}\left(  \boldsymbol{c},\gamma_{v}\right)  :=A_{vh}\log
\Phi\left(  W_{vh}\boldsymbol{c}^{\prime}+\gamma_{v}\right)  +\left(
1-A_{vh}\right)  \log\left(  1-\Phi\left(  W_{vh}\boldsymbol{c}^{\prime
}+\gamma_{v}\right)  \right)  , \label{def: Uvh}%
\end{equation}
$\Upsilon_{1}$ is a compact set in $\mathbb{R}^{d_{W}}$, and $\Upsilon\left(
\bar{v}\right)  $ is a compact interval on $\mathbb{R}$ that may grow as
$\bar{v}\rightarrow\infty$ (specified in (\ref{eq: gamma-space}) in Appendix
\ref{sec: consistency-proof}). Second, via a second probit of $A_{vh}$ on
$(W_{vh},1,\hat{\pi}_{v},\bar{W}_{v})$, we can estimate the coefficients
$(\boldsymbol{\bar{c}}^{\ast},\bar{c}_{0}^{\ast},\bar{\alpha}^{\ast
},\boldsymbol{\bar{\delta}}^{\ast})$ through%
\[
(\widehat{\boldsymbol{\bar{c}}},\widehat{\bar{c}}_{0},\widehat{\bar{\alpha}%
},\widehat{\boldsymbol{\bar{\delta}}})=\underset{(\boldsymbol{\bar{c}},\bar
{c}_{0},\bar{\alpha},\boldsymbol{\bar{\delta}})\in \Upsilon_{2}%
}{\operatorname{argmax}}\hat{R}(\boldsymbol{\bar{c}},\bar{c}_{0},\bar{\alpha
},\boldsymbol{\bar{\delta}}),
\]
where the objective function $\hat{R}$ is defined as%
\begin{align*}
\hat{R}(\boldsymbol{\bar{c}},\bar{c}_{0},\bar{\alpha},\boldsymbol{\bar{\delta
}})  &  =\frac{1}{N}\sum\nolimits_{v=1}^{\bar{v}}\sum\nolimits_{h=1}^{N_{v}%
}\left[  A_{vh}\log\Phi(W_{vh}\boldsymbol{\bar{c}}^{\prime}+\bar{c}_{0}%
+\bar{\alpha}\hat{\pi}_{v}+\bar{W}_{v}\boldsymbol{\bar{\delta}}^{\prime
})\right. \\
&  \left.  +\left(  1-A_{vh}\right)  \log(1-\Phi(W_{vh}\boldsymbol{\bar{c}%
}^{\prime}+\bar{c}_{0}+\bar{\alpha}\hat{\pi}_{v}+\bar{W}_{v}\boldsymbol{\bar
{\delta}}^{\prime}))\right]  ,
\end{align*}
and $\Upsilon_{2}$ is a compact set in $\mathbb{R}^{2d_{W}+1}$.\footnote{By
the results in Section \ref{sec: eqm-beliefs}, we have $\bar{\pi}_{v}=E\left[
A_{vh}|d_{v},e_{v}\right]  $ in the equilibrium, which can be consistently
estimated by an average within each village, $\hat{\pi}_{v}=\frac{1}{N_{v}%
}\sum_{h=1}^{N_{v}}A_{vh}$.} Then, we can recover an estimate of $\hat{\sigma
}_{e}^{2}$ through a ratio of the first (or any other) components of
$\boldsymbol{\hat{c}}$ and $\widehat{\boldsymbol{\bar{c}}}$, which yields
$\sqrt{1+\hat{\sigma}_{e}^{2}}$ and further $(\hat{c}_{0},\hat{\alpha
},\boldsymbol{\hat{\delta}})=(\widehat{\bar{c}}_{0},\widehat{\bar{\alpha}%
},\widehat{\boldsymbol{\bar{\delta}}})\times\sqrt{1+\hat{\sigma}_{e}^{2}}$.
Finally, given these estimates, we can compute
\[
\hat{\xi}_{v}=\hat{\gamma}_{v}-\hat{c}_{0}-\hat{\alpha}\hat{\pi}_{v}%
\]
for each $v$, which then allows us to calculates the welfare estimates of
Section \ref{sec: Welfare}.

The proof of consistency for the first probit is involved due to the CRE
structure and the formal steps are provided in Appendix
\ref{sec: consistency-proof}. The key substantive assumption delivering
consistency is as follows:

\begin{description}
\item[CR1] (i) For each $v$, let $\lambda_{v}^{\min}$ be the minimum of the
eigenvalues of $E_{\boldsymbol{\omega}_{v}}[(W_{vh},1)^{\prime}(W_{vh},1)]$,
which is a square (real symmetric) matrix of order $d_{W}+1$, where $d_{W}$ is
the dimension of $W_{vh}$. Then, $\inf_{v\geq1}$ $\lambda_{v}^{\min}>0$. (ii)
The covariates and unobservables $(W_{vh},\xi_{v},\varepsilon_{vh})\ $satisfy
(\ref{eq: factor}), (\ref{eq: xi-CR-normal}), and (\ref{eq: normal-epsilon}).
(iii) Let $\bar{W}_{v}^{\ast}:=\operatorname*{plim}\limits_{N_{v}%
\rightarrow\infty}\frac{1}{N_{v}}\sum_{h=1}^{N_{v}}W_{vh}%
(=E_{\boldsymbol{\omega}_{v}}\left[  W_{vh}\right]  )$ for each $v$ (the
existence of $\bar{W}_{v}^{\ast}$ is supposed), and%
\[
\mathbb{\bar{W}}:=\left[
\begin{array}
[c]{ccc}%
1 & \bar{\pi}_{1} & \bar{W}_{1}^{\ast}\\
\vdots &  & \vdots\\
1 & \bar{\pi}_{\bar{v}} & \bar{W}_{\bar{v}}^{\ast}%
\end{array}
\right]  ,
\]
which is a $\bar{v}\times\left(  2+d_{W}\right)  $ matrix. Then, suppose that
$\mathbb{\bar{W}}$ is of rank $2+d_{W}$.
\end{description}

Condition (i) of \textbf{CR1} allows us to identify $(\boldsymbol{c}^{\ast
},\gamma_{v}^{\ast})$ for each $v$ as the maximizer of $Q_{v}\left(
\boldsymbol{c},\gamma_{v}\right)  =E_{\boldsymbol{\omega}_{v}}\left[
\mathcal{L}_{vh}\left(  \boldsymbol{c},\gamma_{v}\right)  \right]  $ whose
empirical analogue $\hat{Q}_{v}$ defined in (\ref{def: Qvhat-Qv}) is a
constituent of the objective $\hat{Q}$ defined in (\ref{def: Qhat}); see the
expression (\ref{eq: def-Qhat}) in Appendix \ref{sec: consistency-proof}. The
condition on the uniform lower bound of $\lambda_{v}^{\min}$ together with
(\ref{eq: Nv-assumption}) guarantees the identification/consistency of all
$\gamma_{v}^{\ast}$ when $\bar{v}\rightarrow\infty$. (iii) of \textbf{CR1 }is
used to verify identification of the parameters in the second probit,
$(\boldsymbol{\bar{c}}^{\ast},\bar{c}_{0}^{\ast},\bar{\alpha}^{\ast
},\boldsymbol{\bar{\delta}}^{\ast})$, where we note that by the definition in
(\ref{without}) and the reduced form expression of $\xi_{v}$ in
(\ref{eq: xi-CR}), we can write the true village-specific effect $\gamma
_{v}^{\ast}=c_{0}^{\ast}+\alpha^{\ast}\bar{\pi}_{v}+\bar{W}_{v}^{\ast
}\boldsymbol{\delta}^{\prime}+e_{v}$.

The formal consistency statement is expressed in the following proposition:

\begin{proposition}
\label{prop: consistency-many}Suppose that Conditions \textbf{C1},
\textbf{C2}, \textbf{CR1 (i)-(ii)}, and the technical condition \textbf{CR2}
(stated in Appendix \ref{sec: consistency-proof}), Specifications
(\ref{eq: xi-CR-normal}) and (\ref{eq: normal-epsilon}) hold and that $\bar
{v}$ and $N_{0}$ satisfy Assumption (\ref{eq: Nv-assumption}) with%
\begin{equation}
\left.  \bar{v}^{4(\sigma_{e}^{\ast})^{2}}\left(  \log\bar{v}\right)
^{3+4(\sigma_{e}^{\ast})^{2}}\left(  \log N_{0}\right)  \right/
N_{0}\rightarrow0\text{ \ (as }N_{0}\rightarrow\infty\text{).}
\label{eq: rate}%
\end{equation}
Then, as $N_{0}\rightarrow\infty$ and $\bar{v}\rightarrow\infty$,%
\[
\left\Vert \boldsymbol{\hat{c}}-\boldsymbol{c}^{\ast}\right\Vert
\overset{p}{\rightarrow}0\text{ \ and \ }\max_{v\in\left\{  1,\dots,\bar
{v}\right\}  }\left\vert \hat{\gamma}_{v}-\gamma_{v}^{\ast}\right\vert
\overset{p}{\rightarrow}0.
\]

\end{proposition}

The proof is provided in Appendix \ref{sec: consistency-proof}.

Verification of this proposition for the first probit is not trivial. This is
because (I) given the asymptotic assumption $\bar{v}\rightarrow\infty$,
required for the consistency in the second probit, the number of parameters
tends to infinity; and (II) each parameter $\gamma_{v}^{\ast}=c_{0}^{\ast
}+\alpha^{\ast}\bar{\pi}_{v}+\bar{W}_{v}^{\ast}\boldsymbol{\delta}^{\ast
\prime}+e_{v}$ includes a realization of $e_{v}\sim N\left(  0,\sigma_{e}%
^{2}\right)  $ and thus the maximum of realized $\left\vert \gamma_{1}^{\ast
}\right\vert ,\dots,\left\vert \gamma_{\bar{v}}^{\ast}\right\vert $ grows with
positive probability as $\bar{v}\rightarrow\infty$ since $e_{v}$ has unbounded
support $(-\infty,\infty)$.

Several previous papers on panel models have considered a setting like (I),
such as Hahn and Kuersteiner (2011) and Fern\'{a}ndez-Val and Weidner (2016).
However, in these papers, the \textquotedblleft growing \textit{magnitude} of
parameters\textquotedblright\ as (II) is not allowed for, i.e., typically, all
parameters are supposed to be in a fixed compact set.\footnote{This is
explicitly assumed in Hahn and Kuersteiner's Condition 4, while cases like
(II) have to be typically excluded by Fern\'{a}ndez-Val and Weidner's
Assumption 4.1 (v) (the presence of uniform bounds $b_{\min}$ and $b_{\max}$
for the derivative of their objective function).}

These problems, in particular (II), make it hard to establish the
identification of $\left(  \boldsymbol{c}^{\ast},\gamma_{v}^{\ast}\right)  $.
However, we can overcome this by showing that given the I.I.D. $\left\{
e_{v}\right\}  _{v\geq1}$, the maximum of $\left\vert e_{1}\right\vert
,\dots,\left\vert e_{\bar{v}}\right\vert $ is bounded by $\sqrt{4(\sigma
_{e}^{\ast})^{2}\log[\bar{v}\left(  \log\bar{v}\right)  ^{t}]}$ almost surely
for any $t>1/2$ (Lemma \ref{lem: realized-gamma}). This result allows us to
restrict possible support of each $\gamma_{v}^{\ast}$ as a compact set that
grows (as $\bar{v}\rightarrow\infty$); and within this support, if $||\left(
\boldsymbol{c},\gamma_{v}\right)  -\left(  \boldsymbol{c}^{\ast},\gamma
_{v}^{\ast}\right)  ||>\epsilon_{1}$, we can always find some constant
$C_{Q}>0$ such that $Q_{v}\left(  \boldsymbol{c}^{\ast},\gamma_{v}^{\ast
}\right)  -Q_{v}\left(  \boldsymbol{c},\gamma_{v}\right)  \geq C_{Q}[\bar
{v}(\log\bar{v})]^{-2(\sigma_{e}^{\ast})^{2}}$ almost surely (shown in Lemma
\ref{lem: identification-lower}), which means the identification of $\left(
\boldsymbol{c}^{\ast},\gamma_{v}^{\ast}\right)  $ as the unique maximizer of
$Q_{v}$. If the support were not restricted, the lower bound of the difference
between $Q_{v}\left(  \boldsymbol{c}^{\ast},\gamma_{v}^{\ast}\right)  $ and
$Q_{v}\left(  \boldsymbol{c},\gamma_{v}\right)  $ would be zero as $\bar
{v}\rightarrow\infty$, which would complicate identification and
consistency.\footnote{To see this, let $||\boldsymbol{c}-\boldsymbol{c}^{\ast
}||>\epsilon_{1}$ and $\gamma_{v}=\gamma_{v}^{\ast}$, for example. Then, for
any $\boldsymbol{c}\neq\boldsymbol{c}^{\ast}$, $[Q_{v}\left(  \boldsymbol{c}%
^{\ast},\gamma_{v}^{\ast}\right)  -Q_{v}\left(  \boldsymbol{c},\gamma
_{v}^{\ast}\right)  ]\rightarrow0$ as $\left\vert \gamma_{v}^{\ast}\right\vert
\rightarrow\infty$, which holds for some $v$ as $\bar{v}\rightarrow\infty$
since $\gamma_{v}^{\ast}=c_{0}^{\ast}+\alpha^{\ast}\bar{\pi}_{v}+\bar{W}%
_{v}^{\ast}\boldsymbol{\delta}^{\ast}{}^{\prime}+e_{v}$ includes the normally
distributed variable $e_{v}$. Note that for large $\left\vert \gamma_{v}%
^{\ast}\right\vert $, both $\Phi\left(  W_{vh}(\boldsymbol{c}^{\ast})^{\prime
}+\gamma_{v}^{\ast}\right)  $ and $\Phi\left(  W_{vh}\boldsymbol{c}^{\prime
}+\gamma_{v}^{\ast}\right)  $ (the normal CDF's) are very close to $1$ or $0$
regardless of $\boldsymbol{c}\neq\boldsymbol{c}^{\ast}$ (i.e., variation of
$\Phi\left(  \cdot\right)  $ is tiny in the tail region); thus, the difference
between $Q_{v}\left(  \boldsymbol{c}^{\ast},\gamma_{v}^{\ast}\right)  $ and
$Q_{v}\left(  \boldsymbol{c},\gamma_{v}^{\ast}\right)  $ is very small, which
are computed through these CDF's.}

The rate condition (\ref{eq: rate}) requires $\bar{v}$ to grow slower than
$N_{0}$ in particular when the variance of $e_{v}$ is large, which is
reasonable in the context of our empirical application, where $\bar{v}=11$ may
be regarded as small relative to $N_{v}=r_{v}N_{0}$ which is $195$ on
average.\footnote{Note that regardless of the rate condition (\ref{eq: rate}),
for the first probit, the magnitude of $\bar{v}$ does not directly affect
estimation precision of $\left(  \boldsymbol{\hat{c}},\hat{\gamma}_{v}\right)
$ (up to first order), whose convergence rate is $1/\sqrt{N_{v}}$. In
contrast, the rate condition matters for the second probit, for which the
integration with respect to $e_{v}$ has to be approximated by the sum over
$e_{1},\dots,e_{\bar{v}}$.} It guarantees that the difference between
$Q_{v}\left(  \boldsymbol{c}^{\ast},\gamma_{v}^{\ast}\right)  $ and
$Q_{v}\left(  \boldsymbol{c},\gamma_{v}\right)  $ is larger than that between
$\hat{Q}_{v}$ and $Q_{v}$, justifying the maximizer of $\hat{Q}_{v}$ as an
estimator, implying the consistency.

To see how our two-step probit performs in finite samples, we implemented a
small Monte Carlo exercise that is reported in Appendix \ref{sec: simulation}%
\ and shows reassuring results for magnitudes of sample size resembling ours.

\subsection{Calculation of Predicted Demand and
Welfare\label{sec: Group-specific}}

In order to calculate our welfare-related quantities, we need to estimate the
structural choice probabilities $q_{1}\left(  p,y,\pi\right)  $ and the
equilibrium values of the choice probabilities, $\pi_{0}$ and $\pi_{1}$, in
the pre and post intervention situations. To do this we will assume that the
unobservables $\varepsilon_{vh}\left(  =u_{vh}^{1}-u_{vh}^{0}\right)  $ are
independent of price and income, conditional on unobserved
village-effects.\footnote{$q_{1}\left(  p,y,\pi\right)  $ defined in
(\ref{14}) as the probability computed with respect to the distribution of
$\eta^{0}-\eta^{1}\left(  =\eta_{vh}^{0}-\eta_{vh}^{1}\right)  $. But given
the specification of $\boldsymbol{\eta}_{vh}=\left(  \eta_{vh}^{0},\eta
_{vh}^{1}\right)  $ in Sections \ref{sec: eqm-beliefs}-\ref{sec: Estimators},
it should be now interpreted as the one with respect to the distribution of
$\varepsilon_{vh}$ (conditionally on the village-fixed effects $\left(
d_{v},e_{v}\right)  $ or $\xi_{v}$), i.e., the probability (\ref{with}) as a
function of $w=\left(  p,y\right)  $ and $\bar{\pi}_{v}=\pi$.} Note that
prices in our data are randomly assigned, so the endogeneity concern is solely
regarding income. Under income endogeneity, Bhattacharya (2018) had discussed
interpretation of welfare distributions as conditional on income (see our
discussion at the end of Section \ref{sec: Empirical-specification-result}%
).\smallskip

\noindent\textbf{Welfare\ Calculations}: Once we have estimates of the
structural choice probabilities from the parametric model above, we can
proceed with welfare calculation in presence of social spillovers and
unobserved group-effects, as follows. Consider an initial situation where
everyone faces the unsubsidized price $p_{0}$, so that the predicted take-up
rate $\pi_{0v}$ in village $v$ solves:%
\begin{equation}
\pi_{0v}=\int\Phi\left(  c_{0}+c_{1}p_{0}+c_{2}y+\alpha\pi_{0v}+\xi
_{v}\right)  dF_{Y}^{v}\left(  y\right)  \text{,} \label{3}%
\end{equation}
where $F_{Y}^{v}\left(  y\right)  $ is the CDF of income $Y_{vh}$ in village
$v$. Section 5.2 above outlines the calculations of all parameters including
the $\xi_{v}$'s appearing in (\ref{3}).

Now consider a policy induced price regime $p_{0}$ for ineligibles (with
wealth larger than $a$) and $p_{1}$ for the eligible (with wealth less than or
equal to $a$). Then the resulting usage $\pi_{1}=\pi_{1v}$ in village $v$ is
obtained via solving the fixed point $\pi_{1v}$ in the equation:%
\begin{equation}
\pi_{1v}=\int\left[
\begin{array}
[c]{c}%
1\left\{  y\leq\tau\right\}  \Phi\left(  c_{0}+c_{1}p_{1}+c_{2}y+\alpha
\pi_{1v}+\xi_{v}\right) \\
+1\left\{  y>\tau\right\}  \Phi\left(  c_{0}+c_{1}p_{0}+c_{2}y+\alpha\pi
_{1v}+\xi_{v}\right)
\end{array}
\right]  dF_{Y}^{v}\left(  y\right)  \text{.} \label{4}%
\end{equation}
For fixed $\left(  p_{0},p_{1}\right)  $, the right-hand sides of the above
fixed point equations (\ref{3}) and (\ref{4}), viewed as functions of
$\pi_{0v}$ and $\pi_{1v}$ respectively, are a map from $\left[  0,1\right]  $
to $\left[  0,1\right]  $ ($\pi_{0v}$ and $\pi_{1v}$ are probabilities taking
values in $\left[  0,1\right]  $). By the continuity of $\Phi$ and Brouwer's
fixed point theorem, there is at least one solution in $\pi_{0v}$ and
$\pi_{1v}$, respectively, implying \textquotedblleft
coherence\textquotedblright. However, there may be multiple solutions, and
then our welfare expressions would have to be applied separately for each
feasible pair of values $\left(  \pi_{0v},\pi_{1v}\right)  $ (see our
discussion on the equilibrium multiplicity in Section \ref{sec: unique-multi}%
). Note that even if the solutions to (\ref{3}) and (\ref{4}) are unique, our
expressions in Theorem \ref{th: elig} and Corollary \ref{corollary: inelig}
imply that welfare distributions are still not point-identified.

Finally, mean welfare effect of the policy change in village $v$ can be
calculated as%
\begin{equation}
\mathcal{W}_{v}=\int\left[  1\left\{  y\leq\tau\right\}  \mathcal{W}%
_{v}^{\mathrm{Elig}}\left(  y\right)  +1\left\{  y>\tau\right\}
\mathcal{W}_{v}^{\mathrm{Inelig}}\left(  y\right)  \right]  dF_{Y}^{v}\left(
y\right)  \text{,} \label{6}%
\end{equation}
where $\mathcal{W}_{v}^{\mathrm{Elig}}\left(  y\right)  $ and $\mathcal{W}%
_{v}^{\mathrm{Inelig}}(y)$ are mean welfare at income $y$ in village $v$,
calculated from (\ref{1}) for the eligible and (\ref{2}) for the ineligible,
respectively, using $\pi_{0v}$ and $\pi_{1v}$ as the predicted take-up
probabilities in village $v$ (analogous to $\pi_{0}$ and $\pi_{1}$ in
(\ref{1}) and (\ref{2})), $\alpha_{1}\in\left[  0,\alpha\right]  $ as above).

\subsection{Equilibrium Existence and Uniqueness\label{sec: unique-multi}}

In this section, we present sufficient conditions of unique equilibrium and
then discuss multiplicity of equilibria as well as its implication for our
demand and welfare estimation. Note that given the parametric model, our
equilibrium condition (or fixed point restriction) takes the form:%
\begin{equation}
\pi_{v}=\int\Phi\left(  w\boldsymbol{c}^{\prime}+c_{0}+\alpha\pi_{v}+\xi
_{v}\right)  d\tilde{F}_{W}\left(  w\right)  \text{ \ (}v=1,\dots,\bar
{v}\text{)}, \label{eq: FP-integral}%
\end{equation}
where $\tilde{F}_{W}\left(  \cdot\right)  $ is a CDF of $W_{vh}=\left(
P_{vh},Y_{vh}\right)  $. Some different $\tilde{F}_{W}$ has to be used,
depending on the context. For example, on the RHS of (\ref{3}), $\tilde{F}%
_{W}$ corresponds to the distribution that gives a point mass for
$P_{vh}=p_{0}$ (when considering a counterfactual analysis with this $p_{0}$)
and $Y_{vh}\sim F_{Y}^{v}\left(  y\right)  $ (the marginal distribution of the
observable variable $Y_{vh}$); and on the RHS of (\ref{4}), a different
$\tilde{F}_{W}$ representing the new subsidy scheme is used.

As stated in the previous subsection, existence of a solution to
(\ref{eq: FP-integral}) follows from Brouwer's fixed point theorem. It is also
clear that if $\alpha\leq0$, the solution is unique. On the other hand, if
$\alpha>0$, a contraction condition is sufficient for uniqueness. The
contraction condition (\ref{eq: contraction-scalar}) in Proposition
\ref{prop: belief-symmetricity} can be verified on a case by case basis. In
particular, for the linear index model, it is easy to see that the condition
for contraction is%
\[
\left\vert \alpha\right\vert \sup_{e\in\mathbb{R}}f_{\varepsilon}\left(
e\right)  <1,
\]
where $\alpha$ denotes the social interaction term, and $f_{\varepsilon
}\left(  \cdot\right)  $ denotes the probability density of $-\varepsilon
_{vh}$. In a probit specification in which $\varepsilon_{vh}$ is the standard
normal variable, $\sup_{e\in\mathbb{R}}f_{\varepsilon}\left(  e\right)
=1/\sqrt{2\pi}$ and thus we require $\left\vert \alpha\right\vert <\sqrt{2\pi
}(\simeq2.506)$ and for a logit specification, $\sup_{e\in\mathbb{R}%
}f_{\varepsilon}\left(  e\right)  =1/4$, and thus $\left\vert \alpha
\right\vert <4$. We check that our probit estimate satisfies this condition in
our application.

Note that the contraction condition (\ref{eq: contraction-scalar}) is not
necessary for uniqueness. That is, if a solution $(\bar{\Pi}_{v1},\dots
,\bar{\Pi}_{vN_{v}})$ to the system of equations (\ref{eq: solution-BigPi-IID}%
) is unique and $m_{v}\left(  \cdot\right)  $ (defined in
(\ref{eq: m-function})), which also depends on the distribution of covariates,
has a unique fixed point (i.e., a solution to $r=m_{v}\left(  r\right)  $ is
unique), the uniqueness for the equilibrium solution holds. We have imposed
(\ref{eq: contraction-scalar}) as it is a convenient condition that guarantees
uniqueness equilibrium solution; it is also typically easy to verify in
applications.\medskip

\noindent\textbf{The Maximum Number of Equilibrium Solutions:} The variable
$w$ in (\ref{eq: FP-integral}) is multivariate but its RHS can be written as
$\int\Phi\left(  q+c_{0}+\alpha\pi_{v}+\xi_{v}\right)  dF_{W\boldsymbol{c}%
^{\prime}}^{v}\left(  q\right)  $ in terms of the integral with respect to the
univariate variable $W_{vh}\boldsymbol{c}^{\prime}$, using its CDF
$F_{W\boldsymbol{c}^{\prime}}^{v}$ and support $[\underline{q}_{v},\bar{q}%
_{v}]$, where existence of the finite endpoints of the support of
$W_{vh}\boldsymbol{c}^{\prime}$ is guaranteed under Condition \textbf{CR2}
(ii) (provided in Appendix A.3). Then, applying the mean value theorem for
Stieltjes integrals, we can find some $q_{v}\in\lbrack\underline{q}_{v}%
,\bar{q}_{v}]$ such that%
\begin{align*}
\int_{\underline{q}_{v}}^{\bar{q}_{v}}\Phi\left(  q+c_{0}+\alpha\pi_{v}%
+\xi_{v}\right)  dF_{W\boldsymbol{c}^{\prime}}^{v}\left(  q\right)   &
=\Phi\left(  q_{v}+c_{0}+\alpha\pi_{v}+\xi_{v}\right)  \int_{\underline{q}%
_{v}}^{\bar{q}_{v}}dF_{W\boldsymbol{c}^{\prime}}^{v}\left(  q\right) \\
&  =\Phi\left(  q_{v}+c_{0}+\alpha\pi_{v}+\xi_{v}\right)  \text{.}%
\end{align*}
Therefore, for each $v$, the fixed point restriction can be re-written as
\begin{equation}
\pi_{v}=\Phi\left(  q_{v}+c_{0}+\alpha\pi_{v}+\xi_{v}\right)  \text{ \ for
each }v\text{.} \label{eq: mean-value-FP}%
\end{equation}
Here, by shape properties of the standard normal CDF $\Phi\left(  x\right)  $
(e.g., its derivative is the normal density $\phi\left(  x\right)  $ with the
two inflection points, $-1$ and $1$), we can see that (\ref{eq: mean-value-FP}%
) has at most three solutions ($\pi_{v}=0$ or $1$ cannot be a solution since
each value of $\Phi$ is on $\left(  0,1\right)  $). In particular, a continuum
of solutions cannot exist since $\Phi\left(  x\right)  $ does not have a
linear part on any interval in the real line. This is summarized via the
following proposition whose proof is also evident from the above discussion.

\begin{proposition}
\label{prop: three-eqm}For each $v$, the maximum number of (equilibrium)
solutions to (\ref{eq: FP-integral}) is three.
\end{proposition}

This is analogous to Proposition 2 of BD01 for the logit distribution case
without covariates $W_{vh}$. The number of equilibria is determined by the
value of $\alpha$ as well as those of $q_{v}$, $c_{0}$, and (in particular)
the unobserved group effects $\xi_{v}$. We now discuss implications of
equilibrium multiplicity in estimation.\medskip

\noindent\textbf{Preference and Demand Estimation under Multiple Equilibria:
}Our estimation involves maximum likelihood in nonlinear models with strategic
interactions; thus, it is useful to recall Hahn and Moon (2010) who consider
estimation of game theoretic models possibly with multiple equilibria under a
panel setting, i.e., observations from many markets are obtained repeatedly
over several time periods. These authors interpret unobservables affecting
equilibrium selection as an unobserved fixed effect, assuming that which
equilibrium is selected in one group is fully characterized by each unobserved
fixed effect, which may be correlated with observed characteristics. Then they
show that equilibrium multiplicity is unlikely to be a problem in panel
settings when the number of equilibria is finite which, in the panel
terminology, is equivalent to the fixed effect having finitely many support
points. However, this result requires that the number of equilibria be
constant across parameters and covariates, which is a strong restriction.

In contrast, in our setting, each village/group effect can be interpreted as a
part of players' preference parameters and it does not fully determine which
equilibrium arises. That is, under the same value of group/village effect
$\xi_{v}$, we may observe different equilibria in village $v$. Given the
knowledge of all the preference parameters and the distribution of all the
covariates and error variables, one can determine the number of possible
equilibria and possible values of beliefs by investigating all solutions of
the fixed point equation (\ref{3}), but cannot in advance see which
equilibrium would be realized. We note that our model setting is not equipped
with any equilibrium selection mechanism (just like thus in Brock and
Durlauf's). In our setup, given $\bar{\pi}_{v}$ and $\xi_{v}$, we do not need
to solve the equilibrium system to predict each player's behavior, which is
determined by $A_{vh}=1\{W_{vh}\boldsymbol{c}^{\prime}+c_{0}+\alpha\bar{\pi
}_{v}+\xi_{v}+\varepsilon_{vh}\geq0\}$, and preference parameters can be
estimated without exploiting the equilibrium fixed point
condition.\footnote{This is quite different from the so-called two step
estimation approach (typically used in the empirical industrial organization
game literature) as in Hotz and Miller (1993) and Pesendorfer and
Schmidt-Dengler (2008), in which equilibrium conditions provide the basis of
identification.} This is possible since (A)\ the only objects that are
endogenously determined in equilibrium are $\bar{\pi}_{v}$ ($v=1,\dots,\bar
{v}$), which can be identifiable as $E_{v}\left[  A_{vh}\right]  $ and thus
consistently estimated in our \textquotedblleft large market\textquotedblright%
\ setting with a large number of players in each village (as discussed in
Section \ref{sec: eqm-beliefs}); and (B) the group effects parameters $\xi
_{v}$ can also be consistently estimated under the correlated random effects structure.

These features of our (and Brock and Durlauf's) modeling allow us to avoid
intrinsically difficult problems caused by the equilibrium multiplicity. In
particular, in any equilibrium realization, the same preference parameters
(that are invariant under different equilibria) can be identified and thus
consistently estimated. As a further illustration, consider a case in which
there are three equilibria, i.e., the fixed point equation (\ref{3}) has three
solutions, $\bar{\pi}_{v}^{H}$, $\bar{\pi}_{v}^{M}$, and $\bar{\pi}_{v}^{L}$,
where we let $\bar{\pi}_{v}^{H}>\bar{\pi}_{v}^{M}>\bar{\pi}_{v}^{L}$ and call
each of equilibria as $H$, $M$, or $L$. Then, depending on $t\in\{H,M,L\}$, we
have a different discrete choice model:%
\[
A_{vh}=1\{W_{vh}\boldsymbol{c}^{\prime}+c_{0}+\alpha\pi_{v}^{t}+\xi
_{v}+\varepsilon_{vh}\geq0\}.
\]
Note that the outcome variable changes depending on which equilibrium arises
(i.e., one can write $A_{vh}=A_{vh}^{t}$); and thus, for each equilibrium $t$,
$\pi_{v}^{t}$ can be consistently estimated by $\frac{1}{N_{v}}\sum
_{h=1}^{N_{v}}A_{vh}$. By plugging in the estimated version of $\pi_{v}%
=\pi_{v}^{t}$, we can construct objective functions (i.e., likelihood
functions) to be maximized, based on which we can consistently estimate the
preference parameters regardless of the realized equilibrium $t\in\{H,M,L\}$.
As for consistency, the preference parameters can be identified as the unique
maximizers of the limits of the objective functions. In particular, for our
\emph{first} probit regression (cf. Section \ref{sec: consistency}), the
choice probability is
\[
\Pr\left(  A_{vh}^{t}=1|W_{vh}=w;d_{v},e_{v}\right)  =\Phi(w(\boldsymbol{c}%
^{\ast})^{\prime}+\gamma_{v}^{\ast t}),
\]
under equilibrium $t$, where $\gamma_{v}^{\ast t}(=c_{0}^{\ast}+\alpha^{\ast
}\bar{\pi}_{v}^{t}+\xi_{v}$) depends on which equilibrium has occurred, and
the (limit) objective function (under equilibrium $t$) is
\[
Q_{v}\left(  \boldsymbol{c},\gamma_{v}\right)  =E_{v}\left[  A_{vh}^{t}%
\log\Phi\left(  W_{vh}\boldsymbol{c}^{\prime}+\gamma_{v}\right)  +\left(
1-A_{vh}^{t}\right)  \log\left(  1-\Phi\left(  W_{vh}\boldsymbol{c}^{\prime
}+\gamma_{v}\right)  \right)  \right]  ,
\]
where $E_{v}\left[  \cdot\right]  =E\left[  \cdot|d_{v},e_{v}\right]  $. Then,
given the true parameter $(\boldsymbol{c}^{\ast},\gamma_{v}^{t\ast}%
)(\neq\left(  \boldsymbol{c},\gamma_{v}\right)  )$%
\begin{align*}
&  Q_{v}\left(  \boldsymbol{c}^{\ast},\gamma_{v}^{\ast t}\right)
-Q_{v}\left(  \boldsymbol{c},\gamma_{v}\right) \\
&  =-E_{v}\left[  \Phi(W_{vh}(\boldsymbol{c}^{\ast})^{\prime}+\gamma
_{v}^{t\ast})\log\left(  \frac{\Phi\left(  W_{vh}\boldsymbol{c}^{\prime
}+\gamma_{v}\right)  }{\Phi\left(  W_{vh}\boldsymbol{c}^{\ast\prime}%
+\gamma_{v}^{t\ast}\right)  }\right)  \right. \\
&  \left.  +\left(  1-\Phi(W_{vh}(\boldsymbol{c}^{\ast})^{\prime}+\gamma
_{v}^{t\ast})\right)  \log\left(  \frac{1-\Phi\left(  W_{vh}\boldsymbol{c}%
^{\prime}+\gamma_{v}\right)  }{1-\Phi\left(  W_{vh}\boldsymbol{c}^{\ast\prime
}+\gamma_{v}^{t\ast}\right)  }\right)  \right] \\
&  >-E_{v}\log\left\{  \Phi\left(  W_{vh}\boldsymbol{c}^{\prime}+\gamma
_{v}\right)  +1-\Phi\left(  W_{vh}\boldsymbol{c}^{\prime}+\gamma_{v}\right)
\right\}  =-\log\left\{  1\right\}  =0,
\end{align*}
where the equality has used the law of iterated expectation and correct
specification assumption (i.e., $E_{v}\left[  A_{vh}^{t}|W_{vh}\right]
=\Phi(W_{vh}(\boldsymbol{c}^{\ast})^{\prime}+\gamma_{v}^{t\ast})$), and the
strict inequality follows from Jensen's inequality, the strict convexity of
$-\log\left(  \cdot\right)  $, and the rank condition on $(W_{vh},1)$
(\textbf{CR1} (ii)).\footnote{This identification argument is standard (as in
Newey and McFadden, 1994, Example 1.2 on page 2125). \newline While we believe
that this inequality for each $v$ is useful for illustrating identification
under the equilibrium multiplicity, it is not sufficient for consistency when
$\bar{v}\rightarrow\infty$ i.e., Proposition \ref{prop: consistency-many}. For
verification of the proposition, we have derived uniform lower bound of
$Q_{v}\left(  \boldsymbol{c}^{\ast},\gamma_{v}^{\ast t}\right)  -Q_{v}\left(
\boldsymbol{c},\gamma_{v}\right)  $ (Lemma \ref{lem: identification-lower} in
the Appendix).} That is, we have%
\[
Q_{v}\left(  \boldsymbol{c}^{\ast},\gamma_{v}^{t\ast}\right)  >Q_{v}\left(
\boldsymbol{c},\gamma_{v}\right)  \text{ for any }\left(  \boldsymbol{c}%
,\gamma_{v}\right)  \neq\left(  \boldsymbol{c}^{\ast},\gamma_{v}^{t\ast
}\right)  .
\]
Thus, $\left(  \boldsymbol{c}^{\ast},\gamma_{v}^{t\ast}\right)  $ is
identified as the unique maximizer of $Q_{v}\left(  \cdot,\cdot\right)  $; in
particular, the same $\boldsymbol{c}^{\ast}$ is always identified under any of
equilibrium $t$, while the identified $\gamma_{v}^{t\ast}$ depends on $t$,
which corresponds to our specification of $\gamma_{v}^{t\ast}=c_{0}^{\ast
}+\alpha^{\ast}\bar{\pi}_{v}^{t}+\xi_{v}$ (including the equilibrium object
$\bar{\pi}_{v}^{t}$). The identification argument for the \emph{second} probit
under the CRE structure is analogous (so details are omitted here): under any
equilibrium $t$, the same $(\boldsymbol{\bar{c}}^{\ast},\bar{c}^{\ast}%
,\bar{\alpha}^{\ast}\boldsymbol{\bar{\delta}}^{\ast})$ is obtained as the
unique maximizer of the limit of $\hat{R}(\boldsymbol{\bar{c}},\bar{c}%
_{0},\bar{\alpha},\boldsymbol{\bar{\delta}})$. That is, given the CRE
structure, the same group effects $\xi_{v}$ can be identified in any
equilibrium $t$ through the procedure outlined in Section
\ref{sec: consistency}. Thus our estimation procedure need not to use the
equilibrium fixed point restriction and is robust to equilibrium multiplicity.

In our empirical application, we use an iterative estimator that exploits the
equilibrium fixed point restriction as in Pastorello, Patilea, and Renault
(2003), and Dominitz and Sherman (2005). This estimator is more efficient
under correct specification than the estimator that does not use
(\ref{eq: FP-integral}). For this iterative estimation, the contraction
property of the fixed point mapping (implying unique equilibrium) is key, the
sufficient condition for which is a \textquotedblleft small $\alpha
$\textquotedblright. Through preliminary investigation i.e. checking estimates
obtained without exploiting the fixed-point condition (\ref{eq: FP-integral}),
we have confirmed that the estimate of $\alpha$ is small, so that the
contraction condition is satisfied.\medskip

\noindent\textbf{Counterfactual Welfare Estimation under Equilibrium
Multiplicity: }As discussed above, we do not need to solve the equilibrium
condition (\ref{eq: FP-integral}) for estimation of preference parameters and
the $\xi_{v}$'s (we do use the equilibrium conditions to predict the
counterfactual $\pi$ resulting from the policy experiment), and are therefore
not affected by the multiplicity. However, when predicting counterfactual
outcomes, we need to solve the equilibrium fixed point condition, and find
solutions $\pi_{v}$'s in the counterfactual scenario, e.g., a hypothetical
subsidy rule to buy an ITN. Given the already estimated structural parameter
values, the solutions $\pi_{v}$ of the fixed point equation
(\ref{eq: FP-integral}) in the counterfactual scenario can be computed for
each $v$. If equilibrium multiplicity is anticipated, e.g., when the estimated
$\alpha$ is larger than the threshold for contraction, we can compute these
multiple solutions for each $v\in\left\{  1,\dots,\bar{v}\right\}  $ through
eye-balling\ since the number of solutions is at most three. This can be done
by drawing a graph of the LHS of (\ref{eq: FP-integral}) and checking points
on the graph that intersect the 45 degree line (or the graph for a least
squares objective function as in (\ref{eq: LS-problem}), presented in Figure
2). The number of villages is eleven in our application and eye-balling is not
difficult. Then, given multiple solutions, a bound for average welfare can be
computed for each solution, and one can report multiple (at most three) bounds
of them or a single union of the multiple intervals.

\section{EMPIRICAL\ CONTEXT AND\ DATA\label{sec: Emprical-context-data}}

Our empirical application concerns the provision of anti-malarial bednets.
Malaria is a life-threatening parasitic disease transmitted from human to
human through mosquitoes. In 2019, an estimated 229 million cases of malaria
occurred worldwide, with 90\% of the cases in sub-Saharan Africa (WHO, 2017).
The main tool for malaria control in sub-Saharan Africa is the use of
insecticide treated bednets. Regular use of a bednet reduces overall child
mortality by around 18 percent and reduces morbidity for the entire population
(Lengeler, 2004). However, at \$6 or more a piece, bednets are unaffordable
for many households, and to palliate the very low coverage levels observed in
the mid-2000s, public subsidy schemes were introduced in numerous countries in
the last 15 years. Our empirical exercise is designed to evaluate such subsidy
schemes not just in terms of their effectiveness in promoting bednet adoption,
but also their impact on individual welfare and deadweight loss. Based on our
discussion in Section 4, we focus on two main sources of spillover, viz. (a) a
preference for conformity, and (b) a concern that mosquitoes will be deflected
to oneself when neighbors protect themselves. Both will generate a positive
effect of the aggregate adoption rate on one's own adoption decision, but they
have different implications for the welfare impact of a price subsidy
policy.\smallskip

\noindent\textbf{Experimental Design: }We exploit data from a 2007 randomized
bednet\ pricing experiment conducted in eleven villages of Western Kenya,
where malaria is transmitted year-round. In each village, a list of $150$ to
$200$ households was compiled from school registers, and households on the
list were randomly assigned to a price at which they could purchase a
long-lasting ITN, a new, highly effective type of antimalarial bednet. After
the random assignment had been performed in office, trained enumerators
visited each sampled household to administer a baseline survey. At the end of
the interview, the household was given a voucher for one long-lasting ITN at
the randomly assigned price level. The amount of subsidy (for those who
received any) varied from $40$\% to $100$\% of the market price in two
villages, and from $40$\% to $90$\% in the remaining $9$ villages; there were
$22$ corresponding final prices faced by households, ranging from $0$ to $300$
Ksh (US \$$5.50$), where $300$ Ksh would be the non-subsidized sale price.
Vouchers could be redeemed within three months at participating local
retailers.\smallskip

\noindent\textbf{Data: }We use data on bednet adoption as observed from coupon
redemption and verified acquisition through a follow-up survey. We also use
data on baseline household characteristics measured during the baseline
survey. The three main baseline characteristics we consider are wealth (the
combined value of all durable and animal assets owned by the household); the
number of children under ten years old; and the education level of the female
head of household.\smallskip

\noindent\textbf{Nonparticipating Households}: While all households in a given
village were potentially interacting, our sample does not cover all village
members. This can potentially cause a problem since selected households might
interact with non-selected ones. However, at the time of the experiment,
non-selected households did not have the opportunity to buy a long-lasting
ITN, so the outcome variable $A$ for such households is zero, whose
conditional expectations are zero as well. Thus, in our specification, even if
we allow for interactions among all the village members, it is easy to do the
necessary adjustments in the empirics, viz. replace $\Pi_{vh}$ in
(\ref{eq: Belief-general}) by%
\begin{equation}
\check{\Pi}_{vh}=(\tfrac{N_{v}-1}{\check{N}_{v}-1})\tfrac{1}{N_{v}-1}%
%TCIMACRO{\dsum \nolimits_{1\leq k\leq N_{v}\text{; }k\neq h}}%
%BeginExpansion
{\displaystyle\sum\nolimits_{1\leq k\leq N_{v}\text{; }k\neq h}}
%EndExpansion
E[A_{vk}|\mathcal{I}_{vh}]=(\tfrac{N_{v}-1}{\check{N}_{v}-1})\Pi_{vh},
\label{eq: non-selected}%
\end{equation}
where $\check{N}_{v}$ equals the total number of households in the village,
and $N_{v}$ those participating in the game. In our empirical setting, this
ratio is about $0.8$ for each village, and we apply this adjustment throughout
the empirical analysis.

\section{EMPIRICAL SPECIFICATION AND
RESULTS\label{sec: Empirical-specification-result}}

We work with the linear index structure (\ref{eq: A2-linear}), where $Y_{vh}$
is taken to be the household wealth, $P_{vh}$ is the experimentally set price
faced by the household, $\Pi_{vh}$ is the observed average adoption in the
village. We also use additional controls, denoted by $Z_{vh}$ below, that can
potentially affect preferences and therefore the take-up of bednet, viz.
presence of children under the age of ten, and years of education of the
oldest female member of the household. A village-specific variable that could
affect adoption is the extent of malaria exposure risk in the village. We
measure this in our data from the response to the question: ``Did anyone in
your household have malaria in the past month?''. Summary statistics are
reported in Table 1, and their village averages are shown in Table 2, for each
of the eleven villages in the data.

Our first set of results correspond to taking $F\left(  \cdot\right)  $ to be
the probit CDF of $\eta_{vh}=\eta_{vh}^{0}-\eta_{vh}^{1}$ (as in (\ref{14})),
i.e. with no village-effects), and then our main results use the correlated
random effects probit model that accounts for village-level unobservables. The
marginal effects at mean are presented in Table 3, corresponding to both the
probit model without village-effects and the CRE probit model that accounts
for village-level unobservables. The fact that the price elasticity is very
similar in the two specifications is expected, since price was exogenously
assigned in the experiment, so accounting for village-fixed unobservables has
no impact on the marginal effect.

It is evident from the table that the demand coefficient is negative and
significantly different from zero (the averaged price elasticity is $-0.12)$,
and that bednet adoption in the village has a significant positive association
with private adoption, conditional on price and other household
characteristics, i.e. $\alpha>0$ in our notation above. The social interaction
coefficient $\alpha$ is $2.2$ for the probit, which is less than the upper
bound for the fixed point map to be a contraction (see discussion in Section
\ref{sec: unique-multi}). The effect of children is negative, likely
reflecting that households with children had already invested in other
anti-malarial steps, e.g., had bought a less effective traditional bednet
prior to the experiment.

Next, we consider a hypothetical subsidy rule, where those with wealth less
than $\tau$ are eligible to get the bednet for $50$ KSh ($83$\% subsidy),
whereas those with wealth larger than $\tau$ get it for the price of $250$ KSh
($17$\% subsidy). Based on our preferred CRE probit model (which is used for
all subsequent results, unless mentioned otherwise), we plot the predicted
aggregate take-up of bednets corresponding to different income thresholds
$\tau$. In Figure 1, for each threshold $\tau$, we plot the fraction of
households eligible for a subsidy on the horizontal axis, and the predicted
fraction choosing the bednet on the vertical axis, based on coefficients
obtained by including (solid) and excluding (small dash) the spillover effect.
The 45 degree line (large dash), showing the fraction eligible for the
subsidy, is also plotted in the same figure for comparison.

It is evident from Figure 1 that ignoring spillovers leads to over-estimation
of adoption at lower thresholds and underestimation at higher thresholds of
eligibility. This happens because ignoring a covariate (here $\pi$) with
positive impact on the outcome in prediction amounts to \textquotedblleft
smoothing\textquotedblright\ over values of $\pi$.\smallskip

Having obtained these (uncompensated) effects, we now turn to calculating the
demand and the mean compensating variation for a hypothetical subsidy scheme.
We consider an initial situation where everyone faces a price of $250$ KSh for
the bednet, and a final situation where a bednet is offered for $50$ KSh to
households with wealth less than $\tau=8000$ KSh (about the $27$th percentile
of the wealth distribution), and for the price of $250$ KSh to those with
wealth above that. The demand results are reported in Table 4, and the welfare
results in Table 5. We perform these calculations village-by-village, and then
aggregate across villages. To calculate these numbers, we first predict the
bednet adoption when everyone is facing a price of $250$ KSh, and then when
eligibles face a price of $50$ KSh and the rest stay at $250$ KSh, giving us
the equilibrium values of $\pi_{0}$ and $\pi_{1}$, respectively. In all such
calculations with our data, we always detected a single solution to the fixed
point $\pi$ (i.e. a unique equilibrium) as can be seen from Figure 2, where we
plot the squared difference between the RHS and the LHS of an empirical
version of the fixed point equation (\ref{d}) (with the additional covariate
$Z_{vh}$), i.e.%
\begin{equation}
\left[  \pi_{1}-\int\left[  1\left\{  y\leq\tau\right\}  \hat{q}_{1}\left(
p_{1},y,z,\pi_{1}\right)  +1\left\{  y>\tau\right\}  \hat{q}_{1}\left(
p_{0},y,z,\pi_{1}\right)  \right]  d\hat{F}_{Y,Z}\left(  y,z\right)  \right]
^{2} \label{eq: LS-problem}%
\end{equation}
on the vertical axis, and $\pi_{1}$ on the horizontal axis, separately for
each of the eleven villages, where $\hat{q}_{1}\left(  p,y,z,\pi\right)  $ is
the predicted demand (choice probability) function at $\left(  p,y,z,\pi
\right)  $. The globally convex nature of each objective function is evident
from Figure 1. The minima are relatively close to each other around $0.15$,
except village 7 and 10, where it is larger. As for $\pi_{0}$, which minimizes
$\left[  \pi_{0}-\int\hat{q}_{1}\left(  p_{1},y,z,\pi_{0}\right)  d\hat
{F}_{Y,Z}\left(  y,z\right)  \right]  ^{2}$, the objective function is also
convex with a minimizing value close to zero in every village, reflecting that
very few households would buy at this high price. These predicted values of
$\pi_{0}$ and $\pi_{1}$ are used as inputs into the prediction of demand as
the structural choice probability (\ref{7}) and welfare as per Theorem
\ref{th: elig} and Corollary \ref{corollary: inelig}.\smallskip

The first row of Table 4 shows the pre-subsidy predicted demand by subsidy
eligibility. In the second row, we calculate the predicted effect of the
subsidy on demand, and break that up by the own price effect (Row 2) and the
spillover effect (Row 3). The own effect is obtained by changing the price in
accordance with the subsidy but keeping the village demand equal to the
pre-subsidy value; the spillover effect is the difference between the overall
effect and the own effect. It is clear that spillover effects on both
eligibles and ineligibles are large in magnitude. In particular, the
spillovers effect raises demand for ineligibles by an amount that nearly
equals its pre-subsidy level.

In Table 5, we report welfare calculations with standard errors computed via
the simple nonparametric bootstrap where households were resampled within each
village in each bootstrap replication. In the first row, we report the welfare
gain of the subsidy rule for eligibles, first assuming no spillovers and using
a probit model without village-effects. In this case, we simply use the
results of Bhattacharya (2015) to calculate the (point-identified) CV for
eligibles as the price changes from $250$ KSh to $50$ KSh. This yields the
value of welfare gain to be $52.589$ KSh. As there is no spillover, the
welfare change of ineligibles is zero by definition, and therefore the net
welfare gain is simply the fraction eligible ($0.27$) times the CV for
eligibles. This is reported in the third column of Table 5. The case with
spillovers under probit and assuming $\alpha_{1}\geq0\geq\alpha_{0}$ are
reported in the second panel of Row 1 using (the negatives of) (\ref{10}),
(\ref{12}) and (\ref{13}) for eligibles, and using (\ref{mean_noneligibles}),
(\ref{15}) and (\ref{16}) for ineligibles.

The 2nd-4th row present analogous results using CRE probit to control for
fixed effects; the 2nd row does this for $\alpha_{1}\geq0\geq\alpha_{0}$ ; the
third row for $\alpha_{1}\geq\alpha_{0}\geq0$, using a large upper limit of
$\alpha_{1}$ (and concurrently $\alpha_{0}=\alpha_{1}-\alpha>0$) to proxy
$\alpha_{1}\nearrow\infty$, (cf. (\ref{5}) and (\ref{8}) above). Finally, the
4th row presents the overall bounds by taking union of the previous two cases.

Under $\alpha_{1}\geq0\geq\alpha_{0}$, both specifications imply that
ineligibles can suffer a large welfare \textit{loss} due to the subsidy. This
is because the subsidy facilitates usage for solely the eligibles, raising the
equilibrium usage $\pi$ in the village, but the ineligibles keep facing the
high price, and thus a lower utility from not buying because $\pi$ is now
higher and $\alpha_{0}\leq0$\textbf{.} However, the few ineligibles who buy,
despite the high price, get some welfare increase from a rise in the adoption
rate, that explains the small upper bound corresponding to the case
$\alpha_{0}=0$. The overall welfare gain aggregated over eligibles and
ineligibles is reported in the column headed \textquotedblleft Net Welfare
Gain\textquotedblright.\smallskip

\noindent\textbf{Deadweight Loss}: To compute the deadweight loss, we subtract
the net welfare from the predicted subsidy expenditure. The latter equals the
amount of subsidy ($200$ KSh) times the demand at the subsidized price 50 KSh
of the eligibles. Thus the expression for DWL is given by%
\[
D=\int\left[
\begin{array}
[c]{l}%
200\times1\left\{  y\leq\tau\right\}  q_{1}\left(  50,y,z,\pi_{1}\right) \\
-1\left\{  y\leq\tau\right\}  \mu^{\mathrm{Elig}}\left(  y,z,\pi_{1},\pi
_{0}\right)  -1\left\{  y>\tau\right\}  \mu^{\mathrm{Inelig}}\left(
y,z,\pi_{1},\pi_{0}\right)
\end{array}
\right]  dF\left(  y,z\right)  \text{,}%
\]
where $y$ denotes wealth, $z$ denotes other covariates, $q_{1}\left(
50,y,z,\pi_{1}\right)  $ denotes predicted demand at price $50$ KSh including
the effect of spillover, and $\mu^{\mathrm{Elig}}$ and $\mu^{\mathrm{Inelig}}$
refer to welfare \textit{gain} for eligibles and ineligibles, respectively.
Ignoring spillovers leads to the point-identified deadweight loss%
\[
D=\int\left[  200\times1\left\{  y\leq\tau\right\}  \times q_{1}%
^{\mathrm{No}\text{\textrm{-}}\mathrm{spillover}}\left(  50,y,z\right)
-1\left\{  y\leq\tau\right\}  \mu^{\mathrm{No}\text{\textrm{-}}%
\mathrm{spillover}}\left(  y,z\right)  \right]  dF\left(  y,z\right)  \text{.}%
\]
\smallskip\smallskip

For the case $\alpha_{1}\geq\alpha_{0}\geq0$, in the last-but-one row of Table
5, there is no welfare loss for anyone, since all spillover is positive, which
explains the negative deadweight loss lower bound, i.e. an efficiency from
subsidizing a positive externality.\medskip

These welfare and DWL numbers support the overall conclusion that accounting
for spillovers can lead to much lower estimates of net welfare gain from the
subsidy program and higher deadweight loss. Some of this difference arises
from potential welfare loss suffered by ineligibles that is missed upon
assuming no spillover, and some from the impact of including spillovers terms
on the prediction of counterfactual purchase-rates (cf. Fig 1). Furthermore,
the two cases $\alpha_{1}\geq0\geq\alpha_{0}$ and $\alpha_{1}\geq\alpha
_{0}\geq0$, which are both consistent with the observed $\alpha=\alpha
_{1}-\alpha_{0}>0$, yield vastly different bounds on welfare, resulting in
wide \emph{overall} bounds on net welfare gain and deadweight loss that
include zero (cf. last row of Table 5), which is the key substantive point of
this paper.\smallskip

\noindent\textbf{Endogeneity}: Price variation is exogenous in our
application, since price was varied randomly by the experimenter. Indeed, it
is still possible that wealth $Y$ is correlated with $\boldsymbol{\eta}$, the
unobserved determinants of bednet purchase (even conditionally on the village
specific effects). However, experimental variation in price $P$ implies also
that $P$ is independent of $\boldsymbol{\eta}$, given $Y$. Consequently, one
can invoke the argument presented in Bhattacharya (2018, Section 3.1), and
interpret the estimated choice probabilities and the corresponding welfare
numbers as conditional on $y$, and then integrating with respect to the
marginal distribution of $y$. This overcomes the problem posed by potentially
endogenous income.

\section{CONCLUSION\label{sec: conclusion}}

This paper develops tools for economic demand and welfare analysis in binary
choice models with social interactions. The key finding is that under
interactions, \emph{welfare} distributions resulting from policy changes such
as a price subsidy are generically not point-identified for given values of
counterfactual aggregate demand, unlike the case \textit{without} spillovers.
This is true even when utility functions and distribution of unobserved
heterogeneity are fully parametrized and there is a unique equilibrium.
Non-identification results from the inability of standard choice data to
distinguish between \emph{different} underlying latent mechanisms, e.g.
conforming motives, consumer learning, negative externalities etc., which
produce the same aggregate social interaction coefficient, but have different
welfare implications depending on which mechanism dominates. This feature is
endemic to many practical settings that economists study, including the
health-product adoption case examined here. Another prominent example is
school-choice, where merit-based vouchers to attend a fee-paying selective
school can create negative externalities by lowering the academic quality of
the free local school via increased departure of high-achieving students. The
resulting welfare implications cannot be calculated based solely on a
Brock-Durlauf style empirical model of individual school-choice inclusive of a
social interaction term. This is in contrast to models \emph{without} social
interaction, where choice probability functions have been shown to contain all
the information required for welfare analysis. Nonetheless, we show that under
standard linear index restrictions, welfare distributions can be bounded.
Under some special and empirically untestable cases e.g. exactly symmetric
spillovers effects or absence of negative externalities, these bounds shrink
to point-identified values. Next, we develop methods of identification and
consistent estimation for the structural utility parameters, required for
prediction of counterfactual outcome and welfare bounds, when there is
unobserved group-level heterogeneity possibly correlated with observable
covariates. This is achieved via a novel latent factor modelling of unobserved
group-effects and observed covariates, and developing a method of asymptotic
analysis where the dimension of nuisance parameters, i.e. the group-effects
whose magnitude may be unbounded, increases as the number of groups increase.

We apply our methods to an empirical setting of adoption of anti-malarial
bednets, using data from a pricing experiment by Dupas (2014) in rural Kenya.
We find that accounting for spillovers provides different predictions for
demand and welfare resulting from hypothetical, means-tested subsidy rules. In
particular, with \emph{positive} interaction effects, predicted demand when
including spillovers are \emph{lower} for less generous eligibility criteria,
compared to demand predicted by ignoring spillovers. At more generous
eligibility thresholds, the conclusion reverses. As for welfare, if negative
health externalities are present, then subsidy-ineligibles can suffer welfare
loss due to increased use by subsidized buyers in the neighborhood; if solely
conforming effects are present and there is no health-related externality,
then welfare can improve.

The implication of these results for applied work is that under social
interactions, welfare analysis of potential interventions requires more
information regarding individual channels of spillovers than knowledge of
solely the choice probability functions inclusive of a social interaction
term. Belief-eliciting surveys, recording the reasons behind the subjects'
actions, can provide a potential solution.\pagebreak

%

%TCIMACRO{\TeXButton{appendix}{\appendix}}%
%BeginExpansion
\appendix
%EndExpansion

\section{Appendix}

\subsection{Proofs of Welfare Results\label{sec: appendix-welfare}}

\begin{proof}
[\textbf{Proof of Theorem 1}]The condition $\left\vert \alpha\right\vert
\sup_{e\in\mathbb{R}}f_{\eta^{0}-\eta^{1}}\left(  e\right)  <1$ guarantees
that the maps on the RHS of (\ref{c}) and (\ref{d}) are contractions by
exactly the same argument as in Subsection \ref{sec: unique-multi}, and
therefore by Bruower's fixed point theorem, the solutions to (\ref{c}) and
(\ref{d}) are unique. Hence it follows by the argument following (\ref{c}) and
(\ref{d}) that $\pi_{1}>\pi_{0}$.

Next, note that since $\beta_{0},\beta_{1}>0$, the LHS of (\ref{g}) is
strictly increasing in $S$, so the condition $S\leq a$ is equivalent to%
\begin{align}
&  \max\left\{  \delta_{1}+\beta_{1}\left(  y+a-p_{1}\right)  +\alpha_{1}%
\pi_{1}+\eta^{1},\delta_{0}+\beta_{0}\left(  y+a\right)  +\alpha_{0}\pi
_{1}+\eta^{0}\right\} \nonumber\\
&  \geq\max\left\{  \delta_{1}+\beta_{1}\left(  y-p_{0}\right)  +\alpha_{1}%
\pi_{0}+\eta^{1},\delta_{0}+\beta_{0}y+\alpha_{0}\pi_{0}+\eta^{0}\right\}
\text{.} \label{d2}%
\end{align}
If $a<p_{1}-p_{0}-\frac{\alpha_{1}}{\beta_{1}}\left(  \pi_{1}-\pi_{0}\right)
<0$, then each term on the LHS of (\ref{d2}) is smaller than the corresponding
term on the RHS. If $a\geq\frac{\alpha_{0}}{\beta_{0}}\left(  \pi_{0}-\pi
_{1}\right)  >0$, then each term on the LHS is larger than the corresponding
term on the RHS.\footnote{Note that the above reasoning also helps establish
existence of a solution to (\ref{g}). We know from above that for
$S<p_{1}-p_{0}-\frac{\alpha_{1}}{\beta_{1}}\left(  \pi_{1}-\pi_{0}\right)  $,
the LHS of (\ref{g}) is strictly smaller than the RHS, and for $S\geq
\frac{\alpha_{0}}{\beta_{0}}\left(  \pi_{0}-\pi_{1}\right)  $, the LHS of
(\ref{g}) is strictly larger than the RHS. By continuity, and the intermediate
value theorem, it follows that there must be at least one $S$ where (\ref{g})
holds with equality.} This gives us the support of $S$:%
\[
\Pr\left(  S\leq a\right)  =\left\{
\begin{array}
[c]{ll}%
0\text{,} & \text{if }a<p_{1}-p_{0}-\frac{\alpha_{1}}{\beta_{1}}\left(
\pi_{1}-\pi_{0}\right)  \text{,}\\
1\text{,} & \text{if }a\geq\frac{\alpha_{0}}{\beta_{0}}\left(  \pi_{0}-\pi
_{1}\right)  \text{.}%
\end{array}
\right.
\]

Now consider the intermediate case where%
\[
a\in\lbrack\underset{<0}{\underbrace{p_{1}-p_{0}-\frac{\alpha_{1}}{\beta_{1}%
}\left(  \pi_{1}-\pi_{0}\right)  }},\text{ }\underset{>0}{\underbrace{\frac
{\alpha_{0}}{\beta_{0}}\left(  \pi_{0}-\pi_{1}\right)  }}).
\]
In this case, the first term on LHS of (\ref{d2}) is larger than first term on
RHS for all $\eta_{1}$, and the second term on LHS of (\ref{d2}) is smaller
than the second term on the RHS for all $\eta_{0}$, and thus (\ref{d2}) is
equivalent to%
\begin{equation}%
\begin{array}
[c]{ll}
& \delta_{1}+\beta_{1}\left(  y+a-p_{1}\right)  +\alpha_{1}\pi_{1}+\eta
^{1}\geq\delta_{0}+\beta_{0}y+\alpha_{0}\pi_{0}+\eta^{0}\\
\Leftrightarrow & \delta_{1}+\beta_{1}\left(  y+a-p_{1}\right)  +\alpha_{1}%
\pi_{0}+\alpha_{1}\left(  \pi_{1}-\pi_{0}\right)  +\eta^{1}\geq\delta
_{0}+\beta_{0}y+\alpha_{0}\pi_{0}+\eta^{0}\text{.}%
\end{array}
\label{f2}%
\end{equation}
Thus, for any given $\alpha_{1}$, we have that the probability of (\ref{f2})
reduces to%
\begin{align}
&  F\left(  c_{0}+\alpha_{1}\left(  \pi_{1}-\pi_{0}\right)  +c_{1}\left(
p_{1}-a\right)  +c_{2}y+\alpha\pi_{0}\right) \nonumber\\
&  =q_{1}(p_{1}-a,y,\pi_{0}+\frac{\alpha_{1}}{\alpha}\left(  \pi_{1}-\pi
_{0}\right)  )\text{.}%
\end{align}

\end{proof}

\begin{proof}
[Proof of Theorem 2]We have from (\ref{d2}) that $\Pr\left(  S\leq a\right)  $
equals%
\begin{align*}
&  \max\left\{  \delta_{1}+\beta_{1}\left(  y+a-p_{1}\right)  +\alpha_{1}%
\pi_{1}+\eta^{1},\delta_{0}+\beta_{0}\left(  y+a\right)  +\alpha_{0}\pi
_{1}+\eta^{0}\right\} \\
&  \geq\max\left\{  \delta_{1}+\beta_{1}\left(  y-p_{0}\right)  +\alpha_{1}%
\pi_{0}+\eta^{1},\delta_{0}+\beta_{0}y+\alpha_{0}\pi_{0}+\eta^{0}\right\}
\text{.}%
\end{align*}
Now, there are 2 cases to consider. If $p_{1}-p_{0}-\frac{\alpha_{1}}%
{\beta_{1}}\left(  \pi_{1}-\pi_{0}\right)  <-\frac{\alpha_{0}}{\beta_{0}%
}\left(  \pi_{1}-\pi_{0}\right)  $, then $Pr\left(  S\leq a\right)  $ reduces
to (\ref{1}). However, because the support is entirely negative now (since
$-\frac{\alpha_{0}}{\beta_{0}}\left(  \pi_{1}-\pi_{0}\right)  \leq0$), mean
welfare is given by $E\left(  S\right)  =-\int_{p_{1}-p_{0}-\frac{\alpha_{1}%
}{\beta_{1}}\left(  \pi_{1}-\pi_{0}\right)  }^{\frac{\alpha-\alpha_{1}}%
{\beta_{0}}\left(  \pi_{1}-\pi_{0}\right)  }F_{S}\left(  a\right)  da$ which
equals%
\begin{align*}
&  -\int_{p_{1}-p_{0}-\frac{\alpha_{1}}{\beta_{1}}\left(  \pi_{1}-\pi
_{0}\right)  }^{\frac{\alpha-\alpha_{1}}{\beta_{0}}\left(  \pi_{1}-\pi
_{0}\right)  }q_{1}\left(  p_{1}-a,y,\pi_{0}+\frac{\alpha_{1}}{\alpha}\left(
\pi_{1}-\pi_{0}\right)  \right)  da\\
&  =-\int_{p_{1}-\frac{\alpha-\alpha_{1}}{\beta_{0}}\left(  \pi_{1}-\pi
_{0}\right)  }^{p_{0}+\frac{\alpha_{1}}{\beta_{1}}\left(  \pi_{1}-\pi
_{0}\right)  }q_{1}\left(  p,y,\pi_{0}+\frac{\alpha_{1}}{\alpha}\left(
\pi_{1}-\pi_{0}\right)  \right)  dp\text{.}%
\end{align*}

The 2nd case is where $p_{1}-p_{0}-\frac{\alpha_{1}}{\beta_{1}}\left(  \pi
_{1}-\pi_{0}\right)  >-\frac{\alpha_{0}}{\beta_{0}}\left(  \pi_{1}-\pi
_{0}\right)  $, then for $a\in\lbrack-\frac{\alpha_{0}}{\beta_{0}}\left(
\pi_{1}-\pi_{0}\right)  ,p_{1}-p_{0}-\frac{\alpha_{1}}{\beta_{1}}\left(
\pi_{1}-\pi_{0}\right)  ]$, we have that $S^{\mathrm{Elig}}\leq a$ is
equivalent to%
\[
\delta_{0}+\beta_{0}\left(  y+a\right)  +\alpha_{0}\pi_{1}+\eta^{0}\geq
\delta_{1}+\beta_{1}\left(  y-p_{0}\right)  +\alpha_{1}\pi_{0}+\eta
^{1}\text{,}%
\]
whose probability equals%
\begin{align*}
\eta^{0}-\eta^{1}  &  \geq\delta_{1}-\delta_{0}+\left(  \beta_{1}-\beta
_{0}\right)  y-\beta_{1}p_{0}-\beta_{0}a+\alpha_{1}\pi_{0}-\alpha_{0}\pi_{1}\\
&  =c_{0}+c_{1}\left(  p_{0}+a\right)  +c_{2}\left(  y+a\right)  +\alpha
_{1}\pi_{0}+\left(  \alpha-\alpha_{1}\right)  \pi_{1}\\
&  =1-q_{1}\left(  p_{0}+a,y+a,\pi_{1}-\frac{\alpha_{1}}{\alpha}\left(
\pi_{1}-\pi_{0}\right)  \right)  \text{.}%
\end{align*}
Thus we get%
\begin{align}
&  \Pr\left(  S^{\mathrm{Elig}}\leq a\right) \nonumber\\
&  =\left\{
\begin{array}
[c]{ll}%
0\text{,} & \text{if }a<-\frac{\alpha_{0}}{\beta_{0}}\left(  \pi_{1}-\pi
_{0}\right)  \text{,}\\
1-q_{1}\left(  p_{0}+a,y+a,\pi_{1}-\frac{\alpha_{1}}{\alpha}\left(  \pi
_{1}-\pi_{0}\right)  \right)  \text{,} & \text{if }-\frac{\alpha_{0}}%
{\beta_{0}}\left(  \pi_{1}-\pi_{0}\right)  \leq a<p_{1}-p_{0}-\frac{\alpha
_{1}}{\beta_{1}}\left(  \pi_{1}-\pi_{0}\right)  \text{,}\\
1\text{,} & \text{if }a\geq p_{1}-p_{0}-\frac{\alpha_{1}}{\beta_{1}}\left(
\pi_{1}-\pi_{0}\right)  \text{.}%
\end{array}
\right.  \label{17}%
\end{align}
Using that $E\left(  S\right)  =-\int_{-\frac{\alpha_{0}}{\beta_{0}}\left(
\pi_{1}-\pi_{0}\right)  }^{p_{1}-p_{0}-\frac{\alpha_{1}}{\beta_{1}}\left(
\pi_{1}-\pi_{0}\right)  }F_{S}\left(  a\right)  da$, we get from (\ref{17})
that%
\begin{align}
E\left(  S^{\mathrm{Elig}}\right)   &  =-\int_{-\frac{\alpha_{0}}{\beta_{0}%
}\left(  \pi_{1}-\pi_{0}\right)  }^{p_{1}-p_{0}-\frac{\alpha_{1}}{\beta_{1}%
}\left(  \pi_{1}-\pi_{0}\right)  }\left(  1-q_{1}\left(  p_{0}+a,y+a,\pi
_{1}-\frac{\alpha_{1}}{\alpha}\left(  \pi_{1}-\pi_{0}\right)  \right)
\right)  da\nonumber\\
&  =-\int_{p_{0}+\frac{\alpha-\alpha_{1}}{\beta_{0}}\left(  \pi_{1}-\pi
_{0}\right)  }^{p_{1}-\frac{\alpha_{1}}{\beta_{1}}\left(  \pi_{1}-\pi
_{0}\right)  }\left(  1-q_{1}\left(  p,y+p-p_{0},\pi_{1}-\frac{\alpha_{1}%
}{\alpha}\left(  \pi_{1}-\pi_{0}\right)  \right)  \right)  dp\text{.}
\label{18}%
\end{align}

\end{proof}

\subsection{Proofs of Equilibrium Results\label{sec: Proof-IID}}

\begin{description}
\item[C2] (i)\ For each $v$, the sequence $\left\{  (W_{vh},\boldsymbol{u}%
_{vh})\right\}  _{h=1}^{N_{v}}$ is I.I.D. conditionally on $\left(
d_{v},e_{v}\right)  $. (ii) $\left\{  \boldsymbol{u}_{vh}\right\}
_{h=1}^{N_{v}}$ is independent of $\left\{  W_{vh}\right\}  _{h=1}^{N_{v}}$
conditionally on $\left(  d_{v},e_{v}\right)  $.
\end{description}

\begin{proof}
[Proof of Proposition \ref{Prop: belief-constancy}]For notational simplicity,
we write $\boldsymbol{\omega}_{v}:=\left(  d_{v},e_{v}\right)  $ in this
proof. By the definition in (\ref{eq: Belief-general}), $\Pi_{vk}=\tfrac
{1}{N_{v}-1}%
%TCIMACRO{\tsum \nolimits_{1\leq j\leq N_{v}\text{; }j\neq k}}%
%BeginExpansion
{\textstyle\sum\nolimits_{1\leq j\leq N_{v}\text{; }j\neq k}}
%EndExpansion
E[A_{vj}|\mathcal{I}_{vk}]$. Since this is the conditional expectations given
$\mathcal{I}_{vk}=(W_{vk},\boldsymbol{u}_{vk},\boldsymbol{\omega}_{v})$, we
can write $(v,k)$'s belief as
\[
\Pi_{vk}=g_{vk}(W_{vk},\boldsymbol{u}_{vk},\boldsymbol{\omega}_{v}),
\]
using a function $g_{vk}(\cdot)$ which may depend on each index $\left(
v,k\right)  $ but is non-random. Thus, plugging this expression of $\Pi_{vk}$
into $A_{vk}=1\{U_{1}(Y_{vk}-P_{vk},\Pi_{vk},\boldsymbol{\eta}_{vk})\geq
U_{0}(Y_{vk},\Pi_{vk},\boldsymbol{\eta}_{vk})\}$, we can write%
\begin{equation}
A_{vk}=f_{vk}(W_{vk},\boldsymbol{u}_{vk},\boldsymbol{\omega}_{v}),
\label{eq: Avk}%
\end{equation}
for some non-random function $f_{vk}(\cdot)$, where $W_{vk}=(Y_{vk},P_{vk})$.

By \textbf{C2}, we have the two conditional independence restrictions:
$\left(  \boldsymbol{u}_{vh},\boldsymbol{u}_{vk}\right)  \perp W_{vh}%
|\boldsymbol{\omega}_{v}$and $\boldsymbol{u}_{vh}\perp\boldsymbol{u}%
_{vk}|\boldsymbol{\omega}_{v}$. These imply that%
\begin{equation}
\text{\textquotedblleft}\boldsymbol{u}_{vk}\perp W_{vh}|\left(  \boldsymbol{u}%
_{vh},\boldsymbol{\omega}_{v}\right)  \ \text{\&}\ \boldsymbol{u}_{vk}%
\perp\boldsymbol{u}_{vh}|\boldsymbol{\omega}_{v}\text{\textquotedblright%
\ }\Leftrightarrow\text{ }\boldsymbol{u}_{vk}\perp(W_{vh},\boldsymbol{u}%
_{vh})|\boldsymbol{\omega}_{v}, \label{eq: ind-uvk}%
\end{equation}
where we have used the following conditional independence relation: for random
objects $Q$, $R$, and $S$,
\begin{equation}
\text{\textquotedblleft}Q\perp R|\left(  S,\boldsymbol{\omega}_{v}\right)
\text{ \& }Q\perp S|\boldsymbol{\omega}_{v}\text{\textquotedblright\ is
equivalent to \textquotedblleft}Q\perp(R,S)|\boldsymbol{\omega}_{v}%
\text{\textquotedblright,} \label{eq: equivalent-ind}%
\end{equation}
which is applied with $Q=\boldsymbol{u}_{vk}$, $R=W_{vh}$, and
$S=\boldsymbol{u}_{vh}$. By the same argument, \textbf{C2 }implies that%
\[%
\begin{tabular}
[c]{ll}
& $(W_{vk},W_{vh})\perp(\boldsymbol{u}_{vk},\boldsymbol{u}_{vh}%
)|\boldsymbol{\omega}_{v}$ \&\ $W_{vk}\perp W_{vh}|\boldsymbol{\omega}_{v}$\\
$\Rightarrow$ & $W_{vk}\perp(\boldsymbol{u}_{vk},\boldsymbol{u}_{vh}%
)|(W_{vh},\boldsymbol{\omega}_{v})$ \&\ $W_{vk}\perp W_{vh}|\boldsymbol{\omega
}_{v}$,
\end{tabular}
\ \ \ \
\]
which is equivalent to
\begin{equation}
W_{vk}\perp(W_{vh},\boldsymbol{u}_{vk},\boldsymbol{u}_{vh})|\boldsymbol{\omega
}_{v}. \label{eq: ind-WL}%
\end{equation}
Below, we denote by $E_{v}\left[  \cdot\right]  $ the conditional expectation
operator given $\boldsymbol{\omega}_{v}\left(  \boldsymbol{=}\left(
d_{v},e_{v}\right)  \right)  $ and $E_{v}[\cdot|B]\equiv E[\cdot
|\boldsymbol{\omega}_{v},B]$. Given the above, we have%
\begin{align*}
E[A_{vk}|\mathcal{I}_{vh}]  &  =E_{v}[f_{vk}(W_{vk},\boldsymbol{u}%
_{vk},\boldsymbol{\omega}_{v})|W_{vh},\boldsymbol{u}_{vh}]\\
&  =\int E_{v}[f_{vk}(W_{vk},\tilde{u},\boldsymbol{\omega}_{v})|W_{vh}%
,\boldsymbol{u}_{vh},\boldsymbol{u}_{vk}=\tilde{u}]dF_{\boldsymbol{u}}%
^{v}(\tilde{u}|\boldsymbol{\omega}_{v})\\
&  =\int E_{v}[f_{vk}(W_{vk},\tilde{u},\boldsymbol{\omega}_{v}%
)]dF_{\boldsymbol{u}}^{v}(\tilde{u}|\boldsymbol{\omega}_{v})\\
&  =E_{v}[f_{vk}(W_{vk},\boldsymbol{u}_{vk},\boldsymbol{\omega}_{v}%
)]=E[A_{vk}|\boldsymbol{\omega}_{v}],
\end{align*}
where the first equality uses (\ref{eq: Avk}), the second and third equalities
follow from (\ref{eq: ind-uvk}) and (\ref{eq: ind-WL}), respectively, the
fourth equality holds since $W_{vk}\perp\boldsymbol{u}_{vk}|\boldsymbol{\omega
}_{v}$, completing the proof.
\end{proof}

\begin{proof}
[Proof of Proposition \ref{prop: belief-symmetricity}]Let
\begin{equation}
\bar{\pi}_{vk}=\bar{\pi}_{vk}(d_{v},e_{v}):=E[A_{vk}|d_{v},e_{v}]\text{ \ for
}h=1,\dots,N_{v}, \label{eq: pi-bar}%
\end{equation}
where henceforth we suppress the dependence of $\bar{\pi}_{vk}$ on $\left(
d_{v},e_{v}\right)  $ for notational simplicity. By Proposition
\ref{Prop: belief-constancy} and (\ref{eq: pi-bar-average}), we have
\begin{equation}
\Pi_{vh}=\bar{\Pi}_{vh}=\tfrac{1}{N_{v}-1}%
%TCIMACRO{\dsum \nolimits_{1\leq k\leq N_{v}\text{; }k\neq h}}%
%BeginExpansion
{\displaystyle\sum\nolimits_{1\leq k\leq N_{v}\text{; }k\neq h}}
%EndExpansion
\bar{\pi}_{vk}\text{.} \label{eq: Pi-bar-av}%
\end{equation}
Given these, we can write%
\begin{equation}
\bar{\pi}_{vh}=E\left[  \left.  1\left\{
\begin{array}
[c]{c}%
U_{1}(Y_{vh}-P_{vh},\tfrac{1}{N_{v}-1}%
%TCIMACRO{\tsum \nolimits_{1\leq k\leq N_{v}\text{; }k\neq h}}%
%BeginExpansion
{\textstyle\sum\nolimits_{1\leq k\leq N_{v}\text{; }k\neq h}}
%EndExpansion
\bar{\pi}_{vk},\boldsymbol{\eta}_{vh})\\
\geq U_{0}(y,\tfrac{1}{N_{v}-1}%
%TCIMACRO{\tsum \nolimits_{1\leq k\leq N_{v}\text{; }k\neq h}}%
%BeginExpansion
{\textstyle\sum\nolimits_{1\leq k\leq N_{v}\text{; }k\neq h}}
%EndExpansion
\bar{\pi}_{vk},\boldsymbol{\eta}_{vh})
\end{array}
\right\}  \right\vert d_{v},e_{v}\right]  ,\text{ \ }h=1,\dots,N_{v}\text{.}
\label{eq: simultaneous-eq}%
\end{equation}
We can easily see that if a symmetric solution to the system of $N_{v}%
$\ equations in (\ref{eq: simultaneous-eq}) exists uniquely, then that of
(\ref{eq: solution-BigPi-IID}) in terms of $\{\bar{\Pi}_{vh}\}_{h=1}^{N_{v}}$
also exists uniquely (vice versa; note that $\bar{\pi}_{vh}=\sum_{k=1}^{N_{v}%
}\bar{\Pi}_{vk}-\left(  N_{v}-1\right)  \bar{\Pi}_{vh}$ by
(\ref{eq: Pi-bar-av})). Therefore, we investigate (\ref{eq: simultaneous-eq}).

Corresponding to (\ref{eq: simultaneous-eq}), define an $N_{v}$-dimensional
vector-valued function of $\boldsymbol{r}=(r_{1},r_{2},\dots,r_{N_{v}}%
)\in\left[  0,1\right]  ^{N_{v}}$ as
\[
\mathcal{M}_{v}(\boldsymbol{r}):=\left(  m_{v}(\tfrac{1}{N_{v}-1}%
%TCIMACRO{\tsum \nolimits_{k\neq1}}%
%BeginExpansion
{\textstyle\sum\nolimits_{k\neq1}}
%EndExpansion
r_{k}),\dots,m_{v}(\tfrac{1}{N_{v}-1}%
%TCIMACRO{\tsum \nolimits_{k\neq N_{v}}}%
%BeginExpansion
{\textstyle\sum\nolimits_{k\neq N_{v}}}
%EndExpansion
r_{k})\right)  ,
\]
where we write $%
%TCIMACRO{\tsum \nolimits_{1\leq k\leq N_{v}\text{; }k\neq h}}%
%BeginExpansion
{\textstyle\sum\nolimits_{1\leq k\leq N_{v}\text{; }k\neq h}}
%EndExpansion
=%
%TCIMACRO{\tsum \nolimits_{k\neq h}}%
%BeginExpansion
{\textstyle\sum\nolimits_{k\neq h}}
%EndExpansion
$, and the metric in the domain and range spaces of $\mathcal{M}_{v}$ is
defined as%
\[
||\boldsymbol{s}-\boldsymbol{\tilde{s}||}_{\infty}:=\max_{1\leq h\leq N_{v}%
}\left\vert s_{h}-\tilde{s}_{h}\right\vert ,
\]
for any $\boldsymbol{s}=(s_{1},\dots,s_{N_{v}})$, $\boldsymbol{\tilde{s}%
}=(\tilde{s}_{1},\dots,\tilde{s}_{N_{v}})\in\left[  0,1\right]  ^{N_{v}}$
(note that both the spaces are taken to be $\left[  0,1\right]  ^{N_{v}}$).
Given these definitions of $\mathcal{M}_{v}(\boldsymbol{r})$ and the metric,
we can easily show that the contraction property of $m_{v}(\cdot)$ carries
over to $\mathcal{M}_{v}(\cdot)$, i.e.,
\[
\left\Vert \mathcal{M}_{v}(\boldsymbol{r})-\mathcal{M}_{v}(\boldsymbol{\tilde
{r}})\right\Vert _{\infty}\leq\rho\left\Vert \boldsymbol{r-\tilde{r}%
}\right\Vert _{\infty},
\]
which implies that there exists a unique solution\ $\boldsymbol{r}^{\ast}$ to
the ($N_{v}$-dimensional) vector-valued equation:%
\begin{equation}
\boldsymbol{r}=\mathcal{M}_{v}(\boldsymbol{r}). \label{eq: vector-valued}%
\end{equation}

Now, consider the following scalar-valued equation $r=m_{v}\left(  r\right)
$. By the contraction property (\ref{eq: contraction-scalar}), it has a unique
solution. Denote this solution by $\bar{r}^{\ast}\in\left[  0,1\right]  $. By
the definition of $\mathcal{M}_{v}(\cdot)$, the vector $\boldsymbol{\bar{r}%
}^{\ast}=(\bar{r}^{\ast},\dots,\bar{r}^{\ast})\in\left[  0,1\right]  ^{N_{v}}$
must be a solution to (\ref{eq: vector-valued}). Then, by the uniqueness of
the solution to (\ref{eq: vector-valued}), this $\boldsymbol{\bar{r}}^{\ast}$
must be a unique solution, which is a set of symmetric beliefs. The proof is
therefore complete.
\end{proof}

\subsection{Proof of Consistency\label{sec: consistency-proof}}

To verify the consistency of our first probit estimator, we introduce the
following functions for each $v\in\left\{  1,\dots,\bar{v}\right\}  $:%
\begin{equation}
\hat{Q}_{v}\left(  \boldsymbol{c},\gamma_{v}\right)  :=\frac{1}{N_{v}}%
\sum\nolimits_{v=1}^{N_{v}}\mathcal{L}_{vh}\left(  \boldsymbol{c},\gamma
_{v}\right)  \text{ \ and \ }Q_{v}\left(  \boldsymbol{c},\gamma_{v}\right)
:=E_{\boldsymbol{\omega}_{v}}\left[  \mathcal{L}_{vh}\left(  \boldsymbol{c}%
,\gamma_{v}\right)  \right]  , \label{def: Qvhat-Qv}%
\end{equation}
where
\[
\mathcal{L}_{vh}\left(  \boldsymbol{c},\gamma_{v}\right)  =A_{vh}\log
\Phi\left(  W_{vh}\boldsymbol{c}^{\prime}+\gamma_{v}\right)  +\left(
1-A_{vh}\right)  \log\left(  1-\Phi\left(  W_{vh}\boldsymbol{c}^{\prime
}+\gamma_{v}\right)  \right)  .
\]
Given this definition of $\hat{Q}_{v}$, we can write the objective function of
the first probit estimator as%
\begin{equation}
\hat{Q}\left(  \boldsymbol{c},\gamma_{1},\dots.,\gamma_{\bar{v}}\right)
=\dfrac{1}{N}%
%TCIMACRO{\dsum \nolimits_{v=1}^{\bar{v}}}%
%BeginExpansion
{\displaystyle\sum\nolimits_{v=1}^{\bar{v}}}
%EndExpansion%
%TCIMACRO{\dsum \nolimits_{h=1}^{N_{v}}}%
%BeginExpansion
{\displaystyle\sum\nolimits_{h=1}^{N_{v}}}
%EndExpansion
\mathcal{L}_{vh}\left(  \boldsymbol{c},\gamma_{v}\right)  =\frac{1}{\bar{v}}%
%TCIMACRO{\dsum \nolimits_{v=1}^{\bar{v}}}%
%BeginExpansion
{\displaystyle\sum\nolimits_{v=1}^{\bar{v}}}
%EndExpansion
\dfrac{\bar{v}N_{v}}{N}\hat{Q}_{v}\left(  \boldsymbol{c},\gamma_{v}\right)  .
\label{eq: def-Qhat}%
\end{equation}
Note that each of the weights $\bar{v}N_{v}/N$ does not degenerate. This is
because, under our assumption that $N_{v}=r_{v}N_{0}$ with $\underline{r}\leq
r_{v}\leq\bar{r}$ ($0<\underline{r}\leq\bar{r}<\infty$) for any $v$ and
$N=\sum_{v=1}^{\bar{v}}r_{v}N_{0}$, we have%
\[
0<\left.  \underline{r}\right/  \bar{r}\leq\dfrac{\bar{v}N_{v}}{N}\leq\left.
\bar{r}\right/  \underline{r}<\infty.
\]
Thus, each of $\hat{Q}_{1},\dots,\hat{Q}_{\bar{v}}$ has a non-negligible
contribution to $\hat{Q}$ relative to others, so that the consistency of each
of $\hat{\gamma}_{1},\dots,\hat{\gamma}_{\bar{v}}$ can be guaranteed.

The objective function $\hat{Q}\left(  \boldsymbol{c},\gamma_{1},\dots
.,\gamma_{\bar{v}}\right)  $ as well as its limit when $N_{0}\rightarrow
\infty$ and $\bar{v}\rightarrow\infty$ is strictly concave. Thus, in practice,
there is no need to specify parameter spaces for $\left(  \boldsymbol{c}%
,\gamma_{1},\dots.,\gamma_{\bar{v}}\right)  $. However, for a rigorous
theoretical derivation of the uniform convergence of $\hat{Q}_{v}$ to $Q_{v}$
and identification of $\left(  \boldsymbol{c}^{\ast},\gamma_{v}^{\ast}\right)
$ based on the limit $Q_{v}$, it is convenient to specify the parameter
spaces. We let $\Upsilon_{\boldsymbol{1}}$ be a fixed compact set on
$\mathbb{R}^{d_{W}}$ in which the true parameter $\boldsymbol{c}^{\ast}$ lies
and $\Upsilon\left(  \bar{v}\right)  $ be a compact interval on $\mathbb{R}$
defined as%
\begin{equation}
\Upsilon\left(  \bar{v}\right)  :=\left\{  \left\vert \gamma\right\vert \leq
s\left\vert \text{ }s=\sqrt{4(\sigma_{e}^{\ast})^{2}\log[\bar{v}\left(
\log\bar{v}\right)  ^{t}]}\right.  \right\}  , \label{eq: gamma-space}%
\end{equation}
which grows when $\bar{v}\rightarrow\infty$, where $t\in\left(  1/2,1\right)
$ is an arbitrary constant, introduced in Lemma \ref{lem: realized-gamma}.
Note that we can treat $\left\{  \gamma_{v}^{\ast}\right\}  _{v=1}^{\bar{v}}$
as if these are parameters to be estimated since the law of covariates is
given conditionally on $\boldsymbol{\omega}_{v}=(d_{v},e_{v})$ in \textbf{C2}
and the law of $\left\{  e_{v}\right\}  _{v=1}^{\bar{v}}$ is independent of
the rest of variables as supposed in \textbf{C1} and (\ref{eq: xi-CR-normal}%
).\medskip

We impose the following additional conditions:

\begin{description}
\item[CR2] (i) Let $\gamma_{v}^{\ast}=c_{0}^{\ast}+\alpha^{\ast}\bar{\pi}%
_{v}+d_{v}\boldsymbol{\delta}^{\ast\prime}+e_{v}$ (the village-specific
unobservable variable as defined in (\ref{without}) and (\ref{eq: factor})).
Suppose that the village specific factors $\left\{  d_{v}\right\}  $ satisfy
$\sup_{v\geq1}\left\Vert d_{v}\right\Vert \leq C_{d}\in\left(  0,\infty
\right)  $. (ii) The observable variables $\left\{  W_{vh}\right\}  $ satisfy
$\sup_{\bar{v}\geq1,\text{ }h\geq1}\left\Vert W_{vh}\right\Vert \leq C_{W}%
\in\left(  0,\infty\right)  $. $\left\{  e_{v}\right\}  _{v=1}^{\bar{v}}$ is
independent of $\{\left(  \{W_{vh}\}_{h=1}^{N_{v}},\varepsilon_{vh}%
,d_{v}\right)  \}_{v=1}^{\bar{v}}$. (iii) $\bar{v}$ tends to $\infty$ with
$\log\bar{v}$ being at most of polynomial order of $N_{0}$ as $N_{0}%
\rightarrow\infty$ (i.e. there exists some $\bar{\kappa}>0$ such that
$\log\bar{v}\leq N_{0}^{\bar{\kappa}}$ for any large $N_{0}$), where $N_{0}$
is introduced in (\ref{eq: Nv-assumption}).
\end{description}

\textbf{(i)}-\textbf{(ii)} of \textbf{CR2} suppose the uniform boundedness of
the village specific factors $d_{v}$ and covariates $W_{vh}$. While these may
be easily relaxed at the cost of some additional conditions and notational
complexity, we maintain them for simplicity. Condition \textbf{(iii)} on the
growth rate of $\bar{v}$ is required for uniform convergence of $\hat{Q}_{v}$
to $Q_{v}$, which is a mild condition. It anyway has to be satisfied under the
rate condition (\ref{eq: rate}) for the consistency in Proposition
\ref{prop: consistency-many}.\medskip

Given these conditions, we use the following three lemmas based on which the
consistency result, i.e. Proposition \ref{prop: consistency-many}, is verified
(proofs of the lemmas are provided in the next subsection):

\begin{lemma}
\label{lem: realized-gamma}Suppose that $\left\{  e_{v}\right\}  _{v\geq1}$ is
I.I.D. with $N\left(  0,(\sigma_{e}^{\ast})^{2}\right)  $ with $\sigma
_{e}^{\ast}>0$ (as implied by \textbf{C1} and (\ref{eq: xi-CR-normal})) and
\textbf{CR2 (i) }holds. Then, for each sample point $\omega\in\bar{\Omega}$
with $\Pr\left[  \bar{\Omega}\right]  =1$, there exists some sufficiently
large $K\left(  =K\left(  \omega\right)  \in\lbrack2,\infty)\right)  $ such
that for any $\bar{v}\geq K$,
\[
\max_{1\leq v\leq\bar{v}}\left\vert \gamma_{v}^{\ast}\right\vert \leq
\sqrt{4(\sigma_{e}^{\ast})^{2}\log[\bar{v}\left(  \log\bar{v}\right)  ^{t}]},
\]
where $t>1/2$ is an arbitrary constant that is independent of $\omega$ and
$\bar{v}$.
\end{lemma}

This lemma derives the almost sure uniform bound of $\left\vert \gamma
_{\bar{v}}^{\ast}\right\vert $ as $\bar{v}\rightarrow\infty$, which can be
understood as follows: For each sample point $\omega\in\bar{\Omega}$, we have
an infinite sequence of realized real numbers $\{e_{1},e_{2},\dots
\}=\{e_{1}\left(  \omega\right)  ,e_{2}\left(  \omega\right)  ,\dots\}$ and
thus its corresponding sequence of $\{\gamma_{1}^{\ast},\gamma_{2}^{\ast
},\dots\}$. For this realized sequence, we can always find some number
$K\left(  \omega\right)  $ such that if $\bar{v}\geq K\left(  \omega\right)
$, then none of $\left\{  \left\vert \gamma_{1}^{\ast}\right\vert ,\left\vert
\gamma_{2}^{\ast}\right\vert ,\dots,\left\vert \gamma_{\bar{v}}^{\ast
}\right\vert \right\}  $ is larger than $\sqrt{4(\sigma_{e}^{\ast})^{2}%
\log[\bar{v}\left(  \log\bar{v}\right)  ^{t}]}$. This upper bound depends on
the standard error $\sigma_{e}^{\ast}$ of $e_{v}$: we have a smaller maximum
of $\left\vert \gamma_{v}^{\ast}\right\vert $ when the variation $\sigma
_{e}^{\ast}$ is smaller. While the support of $\gamma_{v}^{\ast}$ is
unbounded, which consists of a normally distributed error $e_{v}$, Lemma
\ref{lem: realized-gamma} allows us to restrict its parameter space as in
(\ref{eq: gamma-space}). The following two lemmas effectively use this
restricted-parameter-space result for $\gamma_{v}^{\ast}$:

\begin{lemma}
[Identification of $(\mathbf{c}^{\ast},\gamma_{1}^{\ast},\dots,\gamma_{\bar
{v}}^{\ast})$]\label{lem: identification-lower}Suppose that $\left\{
e_{v}\right\}  _{v\geq1}$ is I.I.D. with $N\left(  0,(\sigma_{e}^{\ast}%
)^{2}\right)  $ and \textbf{CR1 (i)-(ii)} and \textbf{CR2} hold. Let
$Q_{v}\left(  \boldsymbol{c},\gamma_{v}\right)  $ be the function defined in
(\ref{def: Qvhat-Qv}). Then, for each $\epsilon_{1}>0$, there exists some
constant $C_{Q}>0$ (independent of $\bar{v}$) such that as $\bar{v}%
\rightarrow\infty$,
\begin{equation}
\inf_{1\leq v\leq\bar{v}}\left[  Q_{v}\left(  \boldsymbol{c}^{\ast},\gamma
_{v}^{\ast}\right)  -\sup_{(\boldsymbol{c},\gamma_{v})\in \Upsilon
_{\boldsymbol{c}}\times \Upsilon\left(  \bar{v}\right)  ;\text{ }\left\Vert
(\boldsymbol{c},\gamma_{v})-(\boldsymbol{c}^{\ast},\gamma_{v}^{\ast
})\right\Vert \geq\epsilon_{1}}Q_{v}\left(  \boldsymbol{c},\gamma_{v}\right)
\right]  \geq C_{Q}[\bar{v}(\log\bar{v})]^{-2(\sigma_{e}^{\ast})^{2}}
\label{eq: identification-diff}%
\end{equation}
holds with probability $1$.
\end{lemma}

\begin{lemma}
[Uniform convergence]\label{lem: UC-double}Suppose that $\left\{
e_{v}\right\}  _{v\geq1}$ is I.I.D. with $N\left(  0,(\sigma_{e}^{\ast}%
)^{2}\right)  $ and \textbf{C2, CR1}, and \textbf{CR2} hold. Let $\hat{Q}_{v}$
and $Q_{v}$ be functions defined in (\ref{def: Qvhat-Qv}). Then, as
$N_{0}\rightarrow\infty$,%
\[
\max_{1\leq v\leq\bar{v}}\sup_{\left(  \boldsymbol{c},\gamma_{v}\right)
\in \Upsilon_{1}\times \Upsilon\left(  \bar{v}\right)  }\left\vert \hat{Q}%
_{v}\left(  \boldsymbol{c},\gamma_{v}\right)  -Q_{v}\left(  \boldsymbol{c}%
,\gamma_{v}\right)  \right\vert =O_{p}(\sqrt{\left.  \left(  \log\bar
{v}\right)  ^{3}\left(  \log N_{0}\right)  \right/  N_{0}}).
\]

\end{lemma}

Lemma \ref{lem: identification-lower} shows the identification of $\left(
\boldsymbol{c}^{\ast},\gamma_{v}^{\ast}\right)  $ through the objective
function $Q_{v}$. For each $v$, we have $Q_{v}\left(  \boldsymbol{c}^{\ast
},\gamma_{v}^{\ast}\right)  -Q_{v}\left(  \boldsymbol{c},\gamma_{v}\right)
>0$ for any $\left(  \boldsymbol{c},\gamma_{v}\right)  \neq\left(
\boldsymbol{c}^{\ast},\gamma_{v}^{\ast}\right)  $, which is a standard result
easily shown under the rank condition on $\left(  W_{vh},1\right)  $ and the
probit specification; however, if the domain of $\gamma_{v}^{\ast}$ were not
restricted, the uniform lower bound $Q_{v}\left(  \boldsymbol{c}^{\ast}%
,\gamma_{v}^{\ast}\right)  -Q_{v}\left(  \boldsymbol{c},\gamma_{v}\right)  $
over $v$ would be zero, which would not allow us to show the consistency result.

Lemma \ref{lem: UC-double} derives the uniform convergence of $\hat{Q}%
_{v}\left(  \boldsymbol{c},\gamma_{v}\right)  $. This is based on the
Bernstein exponential inequality for I.I.D. sequences. The restriction of the
parameter space $\Upsilon\left(  \bar{v}\right)  $ as in
(\ref{eq: gamma-space}), which is compact for each $\bar{v}$, can also
facilitate the verification of the uniform convergence.\medskip

Given these results, we are now ready to prove consistency:

\begin{proof}
[Proof of Proposition \ref{prop: consistency-many}]Note that $\hat{Q}\left(
\boldsymbol{c},\gamma_{1},\dots.,\gamma_{\bar{v}}\right)  $ can be written as
in (\ref{eq: def-Qhat}) and this is a weighted average of $\hat{Q}_{v}\left(
\boldsymbol{c},\gamma_{v}\right)  $ with uniformly non-degenerate weights,
$\frac{\bar{v}N_{v}}{N}\in\left[  \frac{\underline{r}}{\bar{r}},\frac{\bar{r}%
}{\underline{r}}\right]  $ for any $v$. Let $\epsilon_{1}>0$ be an arbitrary
constant and $\epsilon_{2}=\epsilon_{2}\left(  \bar{v}\right)  :=\frac{C_{Q}%
}{2}[\bar{v}(\log\bar{v})]^{-2(\sigma_{e}^{\ast})^{2}}>0$, where $C_{Q}>0$ is
a constant given in Lemma \ref{lem: identification-lower}. We look at%
\begin{align*}
&  \sup_{\boldsymbol{c}\in \Upsilon_{\boldsymbol{c}}\text{; }\forall v\text{,
}\gamma_{v}\in \Upsilon\left(  \bar{v}\right)  \text{ s.t. }\exists v\text{,
}\left\vert \left(  \boldsymbol{c},\gamma_{v}\right)  -\left(  \boldsymbol{c}%
^{\ast},\gamma_{v}^{\ast}\right)  \right\vert >\epsilon_{1}}\hat{Q}\left(
\boldsymbol{c},\gamma_{1},\dots.,\gamma_{\bar{v}}\right) \\
&  \leq\sup_{\boldsymbol{c}\in \Upsilon_{\boldsymbol{c}}\text{; }\forall
v\text{, }\gamma_{v}\in \Upsilon\left(  \bar{v}\right)  \text{ s.t. }\exists
v\text{, }\left\vert \gamma_{v}-\gamma_{v}^{\ast}\right\vert >\epsilon_{1}%
}\frac{1}{\bar{v}}\sum\nolimits_{v=1}^{\bar{v}}\dfrac{\bar{v}N_{v}}{N}%
Q_{v}\left(  \boldsymbol{c},\gamma_{v}\right)  +O_{p}(\sqrt{\left.  \left(
\log\bar{v}\right)  ^{3}\left(  \log N_{0}\right)  \right/  N_{0}})\\
&  <\frac{1}{\bar{v}}\sum\nolimits_{v=1}^{\bar{v}}\dfrac{\bar{v}N_{v}}{N}%
Q_{v}\left(  \boldsymbol{c}^{\ast},\gamma_{v}^{\ast}\right)  -C_{Q}[\bar
{v}(\log\bar{v})]^{-2(\sigma_{e}^{\ast})^{2}}+O_{p}(\sqrt{\left.  \left(
\log\bar{v}\right)  ^{3}\left(  \log N_{0}\right)  \right/  N_{0}})\\
&  \leq\frac{1}{\bar{v}}\sum\nolimits_{v=1}^{\bar{v}}\dfrac{\bar{v}N_{v}}%
{N}\hat{Q}_{v}\left(  \boldsymbol{c}^{\ast},\gamma_{v}^{\ast}\right)
-C_{Q}[\bar{v}(\log\bar{v})]^{-2(\sigma_{e}^{\ast})^{2}}+O_{p}(\sqrt{\left.
\left(  \log\bar{v}\right)  ^{3}\left(  \log N_{0}\right)  \right/  N_{0}})\\
&  \leq\frac{1}{\bar{v}}\sum\nolimits_{v=1}^{\bar{v}}\dfrac{\bar{v}N_{v}}%
{N}\hat{Q}_{v}\left(  \boldsymbol{c}^{\ast},\gamma_{v}^{\ast}\right)
-\epsilon_{2}\leq\frac{1}{\bar{v}}\sum\nolimits_{v=1}^{\bar{v}}\dfrac{\bar
{v}N_{v}}{N}\hat{Q}_{v}\left(  \boldsymbol{\hat{c}},\hat{\gamma}_{v}\right)
-\epsilon_{2},
\end{align*}
where the first and third inequalities follow from the uniform convergence of
$\hat{Q}_{v}$ to $Q_{v}$ (derived in Lemma \ref{lem: UC-double}); the second
follows from Lemma \ref{lem: identification-lower}, the fourth follows from
the definition of $\epsilon_{2}$. and the rate condition in (\ref{eq: rate})
(i.e. $\sqrt{\left.  \left(  \log\bar{v}\right)  ^{3}\left(  \log
N_{0}\right)  \right/  N_{0}}$ is much smaller than $C_{Q}[\bar{v}(\log\bar
{v})]^{-2(\sigma_{e}^{\ast})^{2}}$); and the last is due to the definition of
$\left(  \boldsymbol{\hat{c}},\hat{\gamma}_{1},\dots,\hat{\gamma}_{\bar{v}%
}\right)  $, which is the maximizer of $\hat{Q}$. Now, we have verified that
for\ any $\left(  \boldsymbol{c},\gamma_{1},\dots,\gamma_{v}\right)
\in \Upsilon_{\boldsymbol{c}}\times \Upsilon_{\gamma}^{\bar{v}}\times\dots
\times \Upsilon_{\gamma}^{\bar{v}}$ with $\left\vert \left(  \boldsymbol{c}%
,\gamma_{v}\right)  -\left(  \boldsymbol{c}^{\ast},\gamma_{v}^{\ast}\right)
\right\vert >\epsilon_{1}$ for some $v$, it holds that $\hat{Q}\left(
\boldsymbol{c},\gamma_{1},\dots.,\gamma_{\bar{v}}\right)  <\frac{1}{\bar{v}%
}\sum\nolimits_{v=1}^{\bar{v}}\dfrac{\bar{v}N_{v}}{N}\hat{Q}_{v}\left(
\boldsymbol{\hat{c}},\hat{\gamma}_{v}\right)  =\hat{Q}\left(  \boldsymbol{\hat
{c}},\hat{\gamma}_{1},\dots.,\hat{\gamma}_{\bar{v}}\right)  $ with probability
approaching $1$. This implies that $\hat{\gamma}_{v}$ has to satisfy
$\left\vert \left(  \boldsymbol{\hat{c}},\hat{\gamma}_{v}\right)  -\left(
\boldsymbol{c}^{\ast},\gamma_{v}^{\ast}\right)  \right\vert \leq\epsilon_{1}$
for any $v$ with probability approaching $1$, leading to the desired result of
the uniform convergence of $\left(  \boldsymbol{\hat{c}},\hat{\gamma}%
_{v}\right)  $, completing the proof .
\end{proof}

\subsection{Proofs of Auxiliary Lemmas for Consistency}

\begin{proof}
[Proof of Lemma \ref{lem: realized-gamma}.]We first derive a tail probability
bound of $e_{v}$: for any $y>0$,%
\[
\Pr\left[  \left\vert e_{v}\right\vert >y\right]  =\Pr\left[  \left\vert
e_{v}\right\vert /\sigma_{e}^{\ast}>y/\sigma_{e}^{\ast}\right]  \leq2\frac
{1}{\sqrt{2\pi}\left(  y/\sigma_{e}^{\ast}\right)  }\exp\left\{
-\frac{\left(  y/\sigma_{e}^{\ast}\right)  ^{2}}{2}\right\}  ,
\]
where the inequality follows from a tail bound of the standard normal CDF:
$1-\Phi\left(  x\right)  \leq\frac{1}{x}\phi\left(  x\right)  $ for any $x>0$
(p. 112 of Karatzas and Shreve, 1991). This implies that%
\[
\Pr\left[  \max\limits_{1\leq v\leq\bar{v}}\left\vert e_{v}\right\vert
>y\right]  \leq\sum\nolimits_{v=1}^{\bar{v}}\Pr\left[  \left\vert
e_{v}\right\vert >y\right]  \leq\bar{v}\frac{\sigma_{e}^{\ast}}{\sqrt{\pi/2}%
y}\exp\left\{  -\frac{\left(  y/\sigma_{e}^{\ast}\right)  ^{2}}{2}\right\}  .
\]
Letting $y=\sqrt{4(\sigma_{e}^{\ast})^{2}\log[\bar{v}\left(  \log\bar
{v}\right)  ^{t_{0}}]}$ for an arbitrary $t_{0}>1/2$ , we have for $\bar
{v}\geq2$,%
\[
\Pr\left[  \max_{1\leq v\leq\bar{v}}\left\vert e_{v}\right\vert >\sqrt
{4(\sigma_{e}^{\ast})^{2}\log[\bar{v}\left(  \log\bar{v}\right)  ^{t_{0}}%
]}\right]  \leq\frac{1}{\sqrt{\pi/2}\sqrt{4\log[\bar{v}\left(  \log\bar
{v}\right)  ^{t_{0}}]}}\frac{1}{\bar{v}\left(  \log\bar{v}\right)  ^{t_{0}}}%
\]
and thus%
\[
\sum\nolimits_{\bar{v}=2}^{\infty}\Pr\left[  \max_{1\leq v\leq\bar{v}%
}\left\vert e_{v}\right\vert >\sqrt{4(\sigma_{e}^{\ast})^{2}\log[\bar
{v}\left(  \log\bar{v}\right)  ^{t_{0}}]}\right]  <\infty,
\]
which holds since $\sum_{m=2}^{\infty}\frac{1}{m\left(  \log m\right)  ^{p}%
}<\infty$ for any $p>1$. By the Borel-Cantelli lemma, this implies that the
event "$\max\limits_{1\leq v\leq\bar{v}}\left\vert e_{v}\right\vert \geq
\sqrt{4(\sigma_{e}^{\ast})^{2}\log[\bar{v}\left(  \log\bar{v}\right)  ^{t_{0}%
}]}$" may happen at most finitely many times. Thus, for each sample point
$\omega\in\bar{\Omega}$ with $\Pr\left[  \bar{\Omega}\right]  =1$, there
exists some (sufficiently large) $K\left(  =K\left(  \omega\right)
\geq2\right)  $ such that for any $\bar{v}\geq K$,
\begin{equation}
\max_{1\leq v\leq\bar{v}}\left\vert e_{v}\right\vert \leq\sqrt{4(\sigma
_{e}^{\ast})^{2}\log[\bar{v}\left(  \log\bar{v}\right)  ^{t_{0}}]}.
\label{eq: e-max}%
\end{equation}
Define $C_{\gamma}^{0}:=\left\vert c_{0}^{\ast}\right\vert +\left\vert
\alpha^{\ast}\right\vert +C_{d}\left\Vert \boldsymbol{\delta}^{\ast
}\right\Vert $ ($C_{d}$ is the upper bound of $d_{v}$ introduced in
\textbf{CR2 (i)}). Then, we have $\left\vert \gamma_{v}^{\ast}\right\vert
=\left\vert c_{0}^{\ast}+\alpha^{\ast}\bar{\pi}_{v}+d_{v}\boldsymbol{\delta
}^{\ast\prime}+e_{v}\right\vert \leq C_{\gamma}^{0}+\left\vert e_{v}%
\right\vert $. For another arbitrary constant $t>t_{0}\left(  >1/2\right)  $,%
\[
\max_{1\leq v\leq\bar{v}}\left\vert \gamma_{v}^{\ast}\right\vert \leq
C_{\gamma}^{0}+\sqrt{4(\sigma_{e}^{\ast})^{2}\log[\bar{v}\left(  \log\bar
{v}\right)  ^{t_{0}}]}\leq\sqrt{4(\sigma_{e}^{\ast})^{2}\log[\bar{v}\left(
\log\bar{v}\right)  ^{t}]},
\]
for sufficiently large $\bar{v}$ (if necessary, we may re-define $K$
introduced for (\ref{eq: e-max}) so that (\ref{eq: e-max}) and this
inequalities hold simultaneously). Now, the proof of Lemma
\ref{lem: realized-gamma} is complete.
\end{proof}

\begin{proof}
[Proof of Lemma \ref{lem: identification-lower}]For notational simplicity, let
$\vartheta_{v}:=(\boldsymbol{c},\gamma_{v})$ and $Q_{v}\left(  \boldsymbol{c}%
,\gamma_{v}\right)  =Q_{v}\left(  \vartheta_{v}\right)  $. We also write
$\vartheta_{v}^{\ast}:=(\boldsymbol{c}^{\ast},\gamma_{v}^{\ast})$ and define
$\rho_{vh}\left(  \vartheta_{v}^{\ast},\vartheta_{v}\right)  :=W_{vh}%
(\boldsymbol{c}^{\ast})^{\prime}-W_{vh}\boldsymbol{c}^{\prime}+\gamma
_{v}^{\ast}-\gamma_{v}$. Noting the \textquotedblleft single
index\textquotedblright\ structure of the model, we can write $\mathcal{L}%
_{vh}\left(  \boldsymbol{c},\gamma_{v}\right)  =\mathcal{L}_{vh}\left(
\boldsymbol{c}^{\ast},-\rho_{vh}\left(  \boldsymbol{c}^{\ast},\boldsymbol{c}%
\right)  +\gamma_{v}^{\ast}\right)  $. Then, using the Taylor expansion,%
\begin{equation}
\mathcal{L}_{vh}\left(  \boldsymbol{c}^{\ast},\gamma_{v}^{\ast}\right)
-\mathcal{L}_{vh}\left(  \boldsymbol{c},\gamma_{v}^{\ast}\right)  =\rho
_{vh}\left(  \vartheta_{v}^{\ast},\vartheta_{v}\right)  \int_{0}^{1}%
\partial_{\gamma}\mathcal{L}_{vh}\left(  \boldsymbol{c}^{\ast},-\left(
1-\lambda\right)  \rho_{vh}\left(  \vartheta_{v}^{\ast},\vartheta_{v}\right)
+\gamma_{v}^{\ast}\right)  d\lambda.\nonumber
\end{equation}
For computing the partial derivative $\partial_{\gamma}\mathcal{L}_{vh}$, we
define a function $\kappa_{0}$ as%
\[
\kappa_{0}\left(  x\right)  :=\frac{\phi\left(  x\right)  }{\Phi\left(
x\right)  \left[  1-\Phi\left(  x\right)  \right]  }.
\]
Then,%
\begin{align}
&  Q_{v}\left(  \vartheta_{v}^{\ast}\right)  -Q_{v}\left(  \vartheta
_{v}\right)  =E_{\boldsymbol{\omega}_{v}}\left[  \rho_{vh}\left(
\boldsymbol{c}^{\ast},\boldsymbol{c}\right)  \int_{0}^{1}\partial_{\gamma
}\mathcal{L}_{vh}\left(  \boldsymbol{c}^{\ast},-\left(  1-\lambda\right)
\rho_{vh}\left(  \vartheta_{v}^{\ast},\vartheta_{v}\right)  +\gamma_{v}^{\ast
}\right)  d\lambda\right] \nonumber\\
&  =E_{\boldsymbol{\omega}_{v}}%
%TCIMACRO{\TeXButton{\Bigl[}{\Bigl[}}%
%BeginExpansion
\Bigl[%
%EndExpansion
\left\vert \rho_{vh}\left(  \vartheta_{v}^{\ast},\vartheta_{v}\right)
\right\vert \int_{0}^{1}\kappa_{0}\left(  W_{vh}(\boldsymbol{c}^{\ast
})^{\prime}-\left(  1-\lambda\right)  \rho_{vh}\left(  \vartheta_{v}^{\ast
},\vartheta_{v}\right)  +\gamma_{v}^{\ast}\right) \nonumber\\
&  \times\left\vert \Phi(W_{vh}(\boldsymbol{c}^{\ast})^{\prime}+\gamma
_{v}^{\ast})-\Phi(W_{vh}(\boldsymbol{c}^{\ast})^{\prime}-\left(
1-\lambda\right)  \rho_{vh}\left(  \vartheta_{v}^{\ast},\vartheta_{v}\right)
+\gamma_{v}^{\ast})\right\vert d\lambda%
%TCIMACRO{\TeXButton{\Bigr]}{\Bigr]}}%
%BeginExpansion
\Bigr]%
%EndExpansion
, \label{eq: U-c-difference}%
\end{align}
where the absolute value signs can be given to $\rho_{vh}\left(  \vartheta
_{v}^{\ast},\vartheta_{v}\right)  $ and the last component of the integrand
since $\kappa_{0}\left(  x\right)  >0$ for any $x$ and $\Phi\left(
\cdot\right)  $ is monotone (i.e., if $\rho_{vh}\left(  \vartheta_{v}^{\ast
},\vartheta_{v}\right)  >0$,%
\[
\Phi(W_{vh}(\boldsymbol{c}^{\ast})^{\prime}+\gamma_{v}^{\ast})-\Phi
(W_{vh}(\boldsymbol{c}^{\ast})^{\prime}-\left(  1-\lambda\right)  \rho
_{vh}\left(  \vartheta_{v}^{\ast},\vartheta_{v}\right)  +\gamma_{v}^{\ast
})\geq0\text{ \ for any }\lambda\in\left[  0,1\right]  \text{;}%
\]
and if $\rho_{vh}\left(  \vartheta_{v}^{\ast},\vartheta_{v}\right)  <0$, the
inequality is reversed).

For deriving a lower bound of (\ref{eq: U-c-difference}), we use the following
inequalities:%
\begin{gather}
\kappa_{0}\left(  x\right)  =\frac{\phi\left(  x\right)  }{\Phi\left(
x\right)  \left[  1-\Phi\left(  x\right)  \right]  }\geq\phi\left(  1\right)
\text{ \ for any }x\in\mathbb{R}\text{;}\label{eq: lower-bound-ratio}\\
\left\vert \Phi\left(  y\right)  -\Phi\left(  x\right)  \right\vert
\geq\left\vert y-x\right\vert \phi\left(  \max\left\{  \left\vert y\right\vert
,\left\vert x\right\vert \right\}  \right)  \text{ \ for any }y,x\in
\mathbb{R}\text{,} \label{eq: lower-bound-diff}%
\end{gather}
where $\phi\left(  1\right)  =(1/\sqrt{2\pi})\exp\left\{  -1/2\right\}  $,
whose proofs are provided below. Using the uniform lower bound
(\ref{eq: lower-bound-ratio}),
\begin{align}
\text{the RHS of (\ref{eq: U-c-difference})}  &  \geq E_{\boldsymbol{\omega
}_{v}}%
%TCIMACRO{\TeXButton{\Bigl[}{\Bigl[}}%
%BeginExpansion
\Bigl[%
%EndExpansion
\left\vert \rho_{vh}\left(  \vartheta_{v}^{\ast},\vartheta_{v}\right)
\right\vert \phi\left(  1\right) \nonumber\\
&  \times\int_{0}^{1/2}\left\vert \Phi\left(  W_{vh}(\boldsymbol{c}^{\ast
})^{\prime}+\gamma_{v}^{\ast}\right)  -\Phi(W_{vh}(\boldsymbol{c}^{\ast
})^{\prime}-\left(  1-\lambda\right)  \rho_{vh}\left(  \vartheta_{v}^{\ast
},\vartheta_{v}\right)  +\gamma_{v}^{\ast})\right\vert d\lambda%
%TCIMACRO{\TeXButton{\Bigr]}{\Bigr]}}%
%BeginExpansion
\Bigr]%
%EndExpansion
\nonumber\\
&  \geq E_{\boldsymbol{\omega}_{v}}%
%TCIMACRO{\TeXButton{\Bigl[}{\Bigl[}}%
%BeginExpansion
\Bigl[%
%EndExpansion
\left\vert \rho_{vh}\left(  \vartheta_{v}^{\ast},\vartheta_{v}\right)
\right\vert \phi\left(  1\right) \nonumber\\
&  \times\left\vert \Phi\left(  W_{vh}(\boldsymbol{c}^{\ast})^{\prime}%
+\gamma_{v}^{\ast}\right)  -\Phi(W_{vh}(\boldsymbol{c}^{\ast})^{\prime}%
-\tfrac{1}{2}\rho_{vh}\left(  \vartheta_{v}^{\ast},\vartheta_{v}\right)
+\gamma_{v}^{\ast})\right\vert
%TCIMACRO{\TeXButton{\Bigr]}{\Bigr]}}%
%BeginExpansion
\Bigr]%
%EndExpansion
\nonumber\\
&  \geq E_{\boldsymbol{\omega}_{v}}%
%TCIMACRO{\TeXButton{\Bigl[}{\Bigl[}}%
%BeginExpansion
\Bigl[%
%EndExpansion
\left\vert \rho_{vh}\left(  \vartheta_{v}^{\ast},\vartheta_{v}\right)
\right\vert \phi\left(  1\right) \nonumber\\
&  \times\left\vert \tfrac{1}{2}\rho_{vh}\left(  \vartheta_{v}^{\ast
},\vartheta_{v}\right)  \right\vert \phi\left(  \max\left\{  \left\vert
W_{vh}(\boldsymbol{c}^{\ast})^{\prime}+\gamma_{v}^{\ast}\right\vert
,\left\vert W_{vh}(\boldsymbol{c}^{\ast})^{\prime}-\tfrac{1}{2}\rho
_{vh}\left(  \vartheta_{v}^{\ast},\vartheta_{v}\right)  +\gamma_{v}^{\ast
}\right\vert \right\}  \right)
%TCIMACRO{\TeXButton{\Bigr]}{\Bigr]}}%
%BeginExpansion
\Bigr]%
%EndExpansion
,\nonumber\\
&  \geq\frac{\phi\left(  1\right)  }{2}E_{\boldsymbol{\omega}_{v}}\left[
\left\vert \rho_{vh}\left(  \vartheta_{v}^{\ast},\vartheta_{v}\right)
\right\vert ^{2}\right]  \phi\left(  C_{W}||\boldsymbol{c}^{\ast}%
||+\frac{\epsilon_{1}}{2}\left(  C_{W}+1\right)  +\left\vert \gamma_{v}^{\ast
}\right\vert \right)  , \label{eq: U-c-difference2}%
\end{align}
where the second inequality holds since $\Phi$ is monotone and the difference
between $W_{vh}(\boldsymbol{c}^{\ast})^{\prime}+\gamma_{v}^{\ast}$ and
$W_{vh}(\boldsymbol{c}^{\ast})^{\prime}-\left(  1-\lambda\right)  \rho
_{vh}\left(  \vartheta_{v}^{\ast},\vartheta_{v}\right)  +\gamma_{v}^{\ast}$ is
minimized at $\lambda=1/2$ over $\left[  0,1/2\right]  $; the third equality
uses $\left\vert \Phi\left(  y\right)  -\Phi\left(  x\right)  \right\vert
\geq\left\vert y-x\right\vert \phi\left(  \max\left\{  \left\vert y\right\vert
,\left\vert x\right\vert \right\}  \right)  $ in (\ref{eq: lower-bound-ratio})
with $y=W_{vh}(\boldsymbol{c}^{\ast})^{\prime}+\gamma_{v}^{\ast}$ and
$x=W_{vh}\left(  \boldsymbol{c}^{\ast}\right)  ^{\prime}-\tfrac{1}{2}\rho
_{vh}\left(  \vartheta_{v}^{\ast},\vartheta_{v}\right)  +\gamma_{v}^{\ast}$
(note that $\phi\left(  s\right)  $ is decreasing for $s\geq0$); and the last
equality holds since
\[
\max\left\{  \left\vert y\right\vert ,\left\vert x\right\vert \right\}  \leq
C_{W}||\boldsymbol{c}^{\ast}||+\frac{\epsilon_{1}}{2}\left(  C_{W}+1\right)
+\left\vert \gamma_{v}^{\ast}\right\vert .
\]
Recalling properties of quadratic forms (and noting that all vectors are
defined as row vectors), we have%
\begin{align*}
E_{\boldsymbol{\omega}_{v}}\left[  \left\vert \rho_{vh}\left(  \vartheta
_{v}^{\ast},\vartheta_{v}\right)  \right\vert ^{2}\right]   &  =\left(
\vartheta_{v}^{\ast},\vartheta_{v}\right)  E_{\boldsymbol{\omega}_{v}}\left[
(W_{vh},1)^{\prime}(W_{vh},1)\right]  \left(  \vartheta_{v}^{\ast}%
,\vartheta_{v}\right)  ^{\prime}\\
&  \geq\left(  \inf\nolimits_{v\geq1}\lambda_{v}^{\min}\right)  \left\Vert
\vartheta_{v}^{\ast}-\vartheta_{v}\right\Vert ^{2}\geq\left(  \inf
\nolimits_{v\geq1}\lambda_{v}^{\min}\right)  \epsilon_{1}^{2},
\end{align*}
where $\lambda_{v}^{\min}$ is the minimum eigenvalue of the symmetric matrix
$E_{\boldsymbol{\omega}_{v}}\left[  W_{vh}^{\prime}W_{vh}\right]  $. From
these, we can obtain
\[
Q_{v}\left(  \vartheta_{v}^{\ast}\right)  -Q_{v}\left(  \vartheta_{v}\right)
\geq\frac{\phi\left(  1\right)  }{2}\left(  \inf\limits_{v\geq1}\lambda
_{v}^{\min}\right)  \epsilon_{1}^{2}\phi\left(  C_{W}||\boldsymbol{c}^{\ast
}||+\frac{\epsilon_{1}}{2}\left(  C_{W}+1\right)  +\left\vert \gamma_{v}%
^{\ast}\right\vert \right)  .
\]
Noting that this lower bound is independent of $\vartheta_{v}=\left(
\boldsymbol{c},\gamma_{v}\right)  $ and using the almost sure bound of
$\left\vert \gamma_{v}^{\ast}\right\vert $ (Lemma \ref{lem: realized-gamma}
with $t\in\left(  1/2,1\right)  $), we can obtain%
\begin{align*}
&  \inf_{1\leq v\leq\bar{v}}\left[  Q_{v}\left(  \vartheta_{v}^{\ast}\right)
-\sup_{\vartheta_{v}\in \Upsilon_{\boldsymbol{c}}\times \Upsilon\left(  \bar
{v}\right)  ;\text{ }\left\Vert \vartheta_{v}-\vartheta_{v}^{\ast}\right\Vert
\geq\epsilon_{1}}Q_{v}\left(  \vartheta_{v}\right)  \right] \\
&  \geq\frac{\phi\left(  1\right)  }{2}\left(  \inf\limits_{v\geq1}\lambda
_{v}^{\min}\right)  \epsilon_{1}^{2}\inf_{\left\vert \gamma_{v}^{\ast
}\right\vert \leq\sqrt{4(\sigma_{e}^{\ast})^{2}\log[\bar{v}\left(  \log\bar
{v}\right)  ^{t}]}}\phi\left(  C_{W}||\boldsymbol{c}^{\ast}||+\frac
{\epsilon_{1}}{2}\left(  C_{W}+1\right)  +\left\vert \gamma_{v}^{\ast
}\right\vert \right)  ,
\end{align*}
where the inequality holds almost surely for (sufficiently) large $\bar{v}$.
We can also derive lower bound of the last component:%
\begin{align*}
&  \inf_{\left\vert \gamma_{v}^{\ast}\right\vert \leq\sqrt{4(\sigma_{e}^{\ast
})^{2}\log[\bar{v}\left(  \log\bar{v}\right)  ^{t}]}}\phi\left(
C_{W}||\boldsymbol{c}^{\ast}||+\frac{\epsilon_{1}}{2}\left(  C_{W}+1\right)
+\left\vert \gamma_{v}^{\ast}\right\vert \right) \\
&  =\frac{1}{\sqrt{2\pi}}\exp\left\{  -\frac{1}{2}\left\vert C_{W}%
||\boldsymbol{c}^{\ast}||+\frac{\epsilon_{1}}{2}\left(  C_{W}+1\right)
+\sqrt{4(\sigma_{e}^{\ast})^{2}\log[\bar{v}\left(  \log\bar{v}\right)  ^{t}%
]}\right\vert ^{2}\right\} \\
&  \geq\frac{1}{\sqrt{2\pi}}\exp\left\{  -\frac{1}{2}\left\vert \sqrt
{4(\sigma_{e}^{\ast})^{2}\log[\bar{v}\left(  \log\bar{v}\right)  ]}\right\vert
^{2}\right\}  =\frac{1}{\sqrt{2\pi}}[\bar{v}(\log\bar{v})]^{-2(\sigma
_{e}^{\ast})^{2}},
\end{align*}
where the last equality holds for sufficiently large $\bar{v}$ (since
$\log[\bar{v}(\log\bar{v})]$ is much larger than $\log[\bar{v}(\log\bar
{v})^{t}]$). From these, we can obtain the desired result:%
\[
\text{the LHS of (\ref{eq: identification-diff})}\geq C_{Q}[\bar{v}(\log
\bar{v})]^{-2(\sigma_{e}^{\ast})^{2}}%
\]
with $C_{Q}=C_{Q}(\epsilon_{1})=\frac{\phi\left(  1\right)  }{2\sqrt{2\pi}%
}\left(  \inf\nolimits_{v\geq1}\lambda_{v}^{\min}\right)  \epsilon_{1}^{2}$.

It remains to verify two inequalities (\ref{eq: lower-bound-ratio}) and
(\ref{eq: lower-bound-diff}).

\noindent\textbf{Proof of inequality (\ref{eq: lower-bound-ratio}):} Using a
bound of the standard normal CDF (p. 112 of Karatzas and Shreve, 1991):
$\int_{x}^{\infty}\phi\left(  u\right)  du\leq\frac{1}{x}\phi\left(  x\right)
$ for any $x>0$, we can derive
\[
1-\Phi\left(  x\right)  \leq\frac{1}{x}\phi\left(  x\right)  \text{ for
}x>0\text{; \ }\Phi\left(  x\right)  \leq\frac{1}{\left\vert x\right\vert
}\phi\left(  \left\vert x\right\vert \right)  \text{ for }x<0
\]
through the symmetry of normal distributions. To avoid explosive behavior of
these upper bounds when $x$ is close to $0$, we use slightly modified bounds
as follows:%
\begin{align}
1-\Phi\left(  x\right)   &  \leq1\left\{  0\leq x\leq1\right\}  +\frac{1}%
{x}\phi\left(  x\right)  1\left\{  1<x\right\}  \text{ \ for }%
x>0;\label{eq: upper-positive}\\
\Phi\left(  x\right)   &  \leq1\left\{  -1\leq x\leq0\right\}  +\frac
{1}{\left\vert x\right\vert }\phi\left(  x\right)  1\left\{  x<-1\right\}
\text{ \ for }x<0, \label{eq: upper-negative}%
\end{align}
which also covers the case with $x=0$. Since $\Phi\left(  x\right)  \in\left(
0,1\right)  $, we have
\[
\frac{\phi\left(  x\right)  }{\Phi\left(  x\right)  \left[  1-\Phi\left(
x\right)  \right]  }\geq\left\{
\begin{array}
[c]{cc}%
\phi\left(  x\right)  & \text{if }\left\vert x\right\vert \leq1,\\
\left\vert x\right\vert ^{2} & \text{if }x>1,
\end{array}
\right.
\]
which, together with (\ref{eq: upper-positive}) and (\ref{eq: upper-negative}%
), leads to%
\begin{align*}
\frac{\phi\left(  x\right)  }{\Phi\left(  x\right)  \left[  1-\Phi\left(
x\right)  \right]  }  &  \geq1\left\{  x\leq1\right\}  \phi\left(  1\right)
+1\left\{  \left\vert x\right\vert >1\right\}  x^{2}\\
&  \geq\phi\left(  1\right)  \text{ \ for any }x\in\mathbb{R}\text{,}%
\end{align*}
where the first inequality holds since $\phi\left(  x\right)  \geq\phi\left(
1\right)  =\frac{1}{\sqrt{2\pi}}\exp\{-\frac{1}{2}\}$ for any $\left\vert
x\right\vert \leq1$ and the second holds since $1>\phi\left(  1\right)
$.\medskip

\noindent\textbf{Proof of inequality (\ref{eq: lower-bound-diff}):} By the
Taylor expansion, we have%
\begin{equation}
\left\vert \Phi\left(  y\right)  -\Phi\left(  x\right)  \right\vert
=\left\vert y-x\right\vert \int_{0}^{1}\phi\left(  x+\lambda\left(
y-x\right)  \right)  d\lambda. \label{eq: lower-bound-Taylor}%
\end{equation}
Consider an interval defined as%
\[
\{x+\lambda\left(  y-x\right)  \text{ }|\text{ }\lambda\in\left[  0,1\right]
\}.
\]
This interval is a connected subset on $\mathbb{R}$; thus it may contain zero,
be in the negative region ($x,y<0$), or be in the positive region ($x,y>0$).
In each of these cases, by the shape of $\phi\left(  z\right)  $ (having a
unique peak at $z=0$, being symmetric, decreasing in the region $z>0$, and
increasing $z<0$), we have for any $\lambda\in\left[  0,1\right]  $,%
\[
\phi\left(  x+\lambda\left(  y-x\right)  \right)  \geq\phi\left(  \max\left\{
\left\vert y\right\vert ,\left\vert x\right\vert \right\}  \right)  \text{
\ for any }\lambda\in\left[  0,1\right]  \text{,}%
\]
which, together with (\ref{eq: lower-bound-Taylor}), implies the desired
result (\ref{eq: lower-bound-diff}). The proof of Lemma
\ref{lem: identification-lower} is now complete.
\end{proof}

\begin{proof}
[Proof of Lemma \ref{lem: UC-double}]Let $\bar{\Upsilon}\left(  \bar
{v}\right)  :=\Upsilon_{1}\times \Upsilon\left(  \bar{v}\right)  $. This set
$\bar{\Upsilon}\left(  \bar{v}\right)  $ is compact for each $\bar{v}$ and
thus it can be divided into $M\left(  \epsilon_{2}\right)  $ subsets,
$\bar{\Upsilon}^{1}\left(  \bar{v}\right)  ,\bar{\Upsilon}^{2}\left(  \bar
{v}\right)  ,\dots,\Upsilon^{M\left(  \epsilon_{2}\right)  }\left(  \bar
{v}\right)  $, such that $\left\Vert \left(  \boldsymbol{c},\gamma_{v}\right)
-\left(  \boldsymbol{\tilde{c}},\tilde{\gamma}_{v}\right)  \right\Vert
<\epsilon_{2}$ whenever $\left(  \boldsymbol{c},\gamma_{v}\right)  $ and
$\left(  \boldsymbol{\tilde{c}},\tilde{\gamma}_{v}\right)  $ are in the same
subset. Let $(\boldsymbol{c}^{\left(  j\right)  },\gamma_{v}^{\left(
j\right)  })$ be some point in $\bar{\Upsilon}^{j}\left(  \bar{v}\right)  $
for each $j\in\{1,2,\dots,M(\epsilon_{2})\}$. Since $\Upsilon_{1}$ is a
compact subset of $\mathbb{R}^{d_{W}}$ and $\Upsilon\left(  \bar{v}\right)  $
is a compact interval on $\mathbb{R}$ that may grow with the rate of
$\sqrt{\log[\bar{v}\left(  \log\bar{v}\right)  ^{t}]}$, the number of subsets
that cover $\bar{\Upsilon}\left(  \bar{v}\right)  $ can be bounded as follows:%
\[
M\left(  \epsilon_{2}\right)  \leq O(\epsilon_{2}^{-d_{W}})\times
O(\epsilon_{2}^{-1}\sqrt{\log[\bar{v}\left(  \log\bar{v}\right)  ^{t}%
]})=O(\epsilon_{2}^{-\left(  d_{W}+1\right)  }\sqrt{\log[\bar{v}\left(
\log\bar{v}\right)  ^{t}]}).
\]
Then,
\begin{align}
&  \max_{1\leq v\leq\bar{v}}\sup_{\left(  \boldsymbol{c},\gamma_{v}\right)
\in\bar{\Upsilon}\left(  \bar{v}\right)  }\left\vert \hat{Q}_{v}\left(
\boldsymbol{c},\gamma_{v}\right)  -Q_{v}\left(  \boldsymbol{c},\gamma
_{v}\right)  \right\vert \nonumber\\
&  \leq\max_{1\leq v\leq\bar{v}}\max_{j\in\{1,2,\dots,M\left(  \epsilon
_{2}\right)  \}}\sup_{\left(  \boldsymbol{c},\gamma_{v}\right)  \in
\bar{\Upsilon}^{j}\left(  \bar{v}\right)  }\left\vert \hat{Q}_{v}\left(
\boldsymbol{c},\gamma_{v}\right)  -Q_{v}\left(  \boldsymbol{c},\gamma
_{v}\right)  \right\vert \nonumber\\
&  \leq\max_{1\leq v\leq\bar{v}}\max_{j\in\{1,2,\dots,M\left(  \epsilon
_{2}\right)  \}}\sup_{\left(  \boldsymbol{c},\gamma_{v}\right)  \in
\bar{\Upsilon}^{j}\left(  \bar{v}\right)  }\left\vert \frac{1}{N_{v}}%
\sum\nolimits_{v=1}^{N_{v}}\left[  \mathcal{L}_{vh}\left(  \boldsymbol{c}%
,\gamma_{v}\right)  -\mathcal{L}_{vh}(\boldsymbol{c}^{\left(  j\right)
},\gamma_{v}^{\left(  j\right)  })\right]  \right\vert \nonumber\\
&  +\max_{1\leq v\leq\bar{v}}\max_{j\in\{1,2,\dots,M\left(  \epsilon
_{2}\right)  \}}\sup_{\left(  \boldsymbol{c},\gamma_{v}\right)  \in
\bar{\Upsilon}^{j}\left(  \bar{v}\right)  }\left\vert \frac{1}{N_{v}}%
\sum\nolimits_{v=1}^{N_{v}}E_{\boldsymbol{\omega}_{v}}\left[  \mathcal{L}%
_{vh}(\boldsymbol{c},\gamma_{v})-\mathcal{L}_{vh}(\boldsymbol{c}^{\left(
j\right)  },\gamma_{v}^{\left(  j\right)  })\right]  \right\vert \nonumber\\
&  +\max_{1\leq v\leq\bar{v}}\max_{j\in\{1,2,\dots,M\left(  \epsilon
_{2}\right)  \}}\left\vert \hat{Q}_{v}(\boldsymbol{c}^{\left(  j\right)
},\gamma_{v}^{\left(  j\right)  })-Q_{v}(\boldsymbol{c}^{\left(  j\right)
},\gamma_{v}^{\left(  j\right)  })\right\vert . \label{eq: UC-3terms}%
\end{align}
Each of the first two terms on the majorant side is bounded by%
\begin{equation}
\left\Vert \left(  \boldsymbol{c},\gamma_{v}\right)  -\left(
\boldsymbol{\tilde{c}},\tilde{\gamma}_{v}\right)  \right\Vert \times
C_{2}\sqrt{\log[\bar{v}\left(  \log\bar{v}\right)  ^{t}]}\leq\epsilon
_{2}\times C_{2}\sqrt{\log[\bar{v}\left(  \log\bar{v}\right)  ^{t}]},
\label{eq: first-two}%
\end{equation}
which follows from
\begin{equation}
\left\vert \mathcal{L}_{vh}\left(  \boldsymbol{c},\gamma_{v}\right)
-\mathcal{L}_{vh}\left(  \boldsymbol{\tilde{c}},\tilde{\gamma}_{v}\right)
\right\vert \leq\left\Vert \left(  \boldsymbol{c},\gamma_{v}\right)  -\left(
\boldsymbol{\tilde{c}},\tilde{\gamma}_{v}\right)  \right\Vert \times
C_{2}\sqrt{\log[\bar{v}\left(  \log\bar{v}\right)  ^{t}]}, \label{eq: Lip}%
\end{equation}
where $C_{2}>0$ is some constant that is independent of $\left(  v,h\right)
$, $\left(  \boldsymbol{c},\gamma_{v}\right)  $, and $\left(
\boldsymbol{\tilde{c}},\tilde{\gamma}_{v}\right)  $; the proof of
(\ref{eq: Lip}) is provided below. Here, by setting
\begin{equation}
\epsilon_{2}=\sqrt{\log[\bar{v}\left(  \log\bar{v}\right)  ^{t}](\log
N_{0})/N_{0}}, \label{eq: epsilon2-def}%
\end{equation}
we have the first two terms on the RHS\ of (\ref{eq: UC-3terms}) be
$O(\sqrt{\left(  \log[\bar{v}\left(  \log\bar{v}\right)  ^{t}]\right)
^{2}\left(  \log N_{0}\right)  /N_{0}})$.

We next derive the probability bound of the third term on the RHS\ of
(\ref{eq: UC-3terms}). To this end, we use the following bounds:%
\begin{equation}
|\mathcal{L}_{vh}\left(  \boldsymbol{c},\gamma_{v}\right)  |\leq
C_{\mathcal{L}}\log[\bar{v}\left(  \log\bar{v}\right)  ^{t}],
\label{eq: U-bound}%
\end{equation}
uniformly over $\left(  v,h\right)  $ and $\left(  \boldsymbol{c},\gamma
_{v}\right)  $, whose proof is provided below, as well as
\[
\mathrm{Var}[%
%TCIMACRO{\tsum \nolimits_{v=1}^{N_{v}}}%
%BeginExpansion
{\textstyle\sum\nolimits_{v=1}^{N_{v}}}
%EndExpansion
\mathcal{L}_{vh}\left(  \boldsymbol{c},\gamma_{v}\right)  ]\leq N_{v}%
\left\vert C_{\mathcal{L}}\log[\bar{v}\left(  \log\bar{v}\right)
^{t}]\right\vert ^{2},
\]
which follows from $\mathrm{Var}\left[  \mathcal{L}_{vh}\left(  \boldsymbol{c}%
,\gamma_{v}\right)  \right]  \leq E[|\mathcal{L}_{vh}\left(  \boldsymbol{c}%
,\gamma_{v}\right)  |^{2}]\leq\left\vert C_{\mathcal{L}}\log[\bar{v}\left(
\log\bar{v}\right)  ^{t}]\right\vert ^{2}$ almost surely as $\bar
{v}\rightarrow\infty$. Thus, by Bernstein's inequality for independent
variables (p. 102, van der Vaart and Wellner, 1996) with these two bounds for
$|\mathcal{L}_{vh}\left(  \boldsymbol{c},\gamma_{v}\right)  |$ and
$\mathrm{Var}[%
%TCIMACRO{\tsum \nolimits_{v=1}^{N_{v}}}%
%BeginExpansion
{\textstyle\sum\nolimits_{v=1}^{N_{v}}}
%EndExpansion
\mathcal{L}_{vh}\left(  \boldsymbol{c},\gamma_{v}\right)  ]$, we have for
$s>0$,%
\begin{align}
&  P_{\boldsymbol{\omega}_{v}}\left[  \max_{1\leq v\leq\bar{v}}\max
_{j\in\{1,2,\dots,M\left(  \epsilon_{2}\right)  \}}\left\vert \hat{Q}%
_{v}(\boldsymbol{c}^{\left(  j\right)  },\gamma_{v}^{\left(  j\right)
})-Q_{v}(\boldsymbol{c}^{\left(  j\right)  },\gamma_{v}^{\left(  j\right)
})\right\vert \geq s\right] \nonumber\\
&  \leq\bar{v}\max_{1\leq v\leq\bar{v}}M\left(  \epsilon_{2}\right)
P_{\boldsymbol{\omega}_{v}}\left[  \sum\nolimits_{v=1}^{N_{v}}\left\{
\mathcal{L}_{vh}(\boldsymbol{c}^{\left(  j\right)  },\gamma_{v}^{\left(
j\right)  })-E[\mathcal{L}_{vh}(\boldsymbol{c}^{\left(  j\right)  },\gamma
_{v}^{\left(  j\right)  })]\right\}  \geq N_{v}s\right] \nonumber\\
&  \leq\bar{v}\max_{1\leq v\leq\bar{v}}M\left(  \epsilon_{2}\right)
2\exp\left\{  -\frac{1}{2}\frac{\left(  N_{v}s\right)  ^{2}}{N_{v}\left\vert
C_{\mathcal{L}}\log[\bar{v}\left(  \log\bar{v}\right)  ^{t}]\right\vert
^{2}+\frac{1}{3}C_{\mathcal{L}}\log[\bar{v}\left(  \log\bar{v}\right)
^{t}]\left(  N_{v}s\right)  }\right\} \nonumber\\
&  =\bar{v}\max_{1\leq v\leq\bar{v}}M\left(  \epsilon_{2}\right)
2\exp\left\{  -\frac{1}{2}\frac{\left(  N_{v}/N_{0}\right)  s_{0}^{2}\left(
\log N_{0}\right)  }{\left\vert C_{\mathcal{L}}\right\vert ^{2}+\frac
{C_{\mathcal{L}}}{3}s_{0}\sqrt{\left(  \log N_{0}\right)  /N_{0}}}\right\}
\nonumber\\
&  \leq\bar{v}\times O\left(  \epsilon_{2}^{-\left(  d_{W}+1\right)  }%
\sqrt{\log[\bar{v}\left(  \log\bar{v}\right)  ^{t}]}\right)  \times
N_{0}^{-\underline{r}s_{0}^{2}\left/  2(\left\vert C_{\mathcal{L}}\right\vert
^{2}+s_{0})\right.  }, \label{eq: Exp-Pr}%
\end{align}
where the equality holds with%
\begin{equation}
s=s_{0}\sqrt{\left\vert \log[\bar{v}\left(  \log\bar{v}\right)  ^{t}%
]\right\vert ^{2}\left(  \log N_{0}\right)  /N_{0}}\text{ \ (}s_{0}>0\text{ is
a constant),} \label{eq: s0-def}%
\end{equation}
and the last inequality holds since $\left(  N_{v}/N_{0}\right)
\geq\underline{r}$ (for any $v$) and $\frac{C_{U}}{3}\sqrt{\left(  \log
N_{0}\right)  /N_{0}}\leq1$ for sufficiently large $N_{0}$. Since we have
defined $\epsilon_{2}=\sqrt{\log[\bar{v}\left(  \log\bar{v}\right)  ^{t}](\log
N_{0})/N_{0}}$ in (\ref{eq: epsilon2-def}) and $\log\bar{v}$ is at most of
polynomial order of $N_{0}$, for some sufficiently large $s_{0}>0$, the
majorant side of (\ref{eq: Exp-Pr}) tends to zero as $N_{0}\rightarrow\infty$.
That is, given (\ref{eq: s0-def}) and $N_{v}=r_{v}N_{0}$ with $\underline{r}%
\leq r_{v}\leq\bar{r}$, we have%
\[
\max_{1\leq v\leq\bar{v}}\max_{j\in\{1,2,\dots,M\left(  \epsilon_{2}\right)
\}}\left\vert \hat{Q}_{v}(\boldsymbol{c}^{\left(  j\right)  },\gamma
_{v}^{\left(  j\right)  })-Q_{v}(\boldsymbol{c}^{\left(  j\right)  }%
,\gamma_{v}^{\left(  j\right)  })\right\vert =O_{p}(\sqrt{\left(  \log[\bar
{v}\left(  \log\bar{v}\right)  ^{t}]\right)  ^{2}\left(  \log N_{0}\right)
/N_{0}}).
\]
Noting that $\left(  \log[\bar{v}\left(  \log\bar{v}\right)  ^{t}]\right)
^{2}<$ $\left(  \log\bar{v}\right)  ^{3}$, this bound, together with
(\ref{eq: UC-3terms}) and (\ref{eq: first-two}), implies the conclusion of
this Lemma \ref{lem: UC-double}.

To complete the proof of the lemma, it remains to show inequalities
(\ref{eq: Lip}) and (\ref{eq: U-bound}).

\noindent\textbf{Proof of (\ref{eq: Lip}): }We look at%
\begin{align*}
&  \left\vert \mathcal{L}_{vh}\left(  \boldsymbol{c},\gamma_{v}\right)
-\mathcal{L}_{vh}\left(  \boldsymbol{\tilde{c}},\tilde{\gamma}_{v}\right)
\right\vert \\
&  =A_{vh}\left[  \log F_{\varepsilon}\left(  W_{vh}\boldsymbol{c}^{\prime
}+\gamma_{v}\right)  -\log F_{\varepsilon}\left(  W_{vh}\boldsymbol{\tilde{c}%
}^{\prime}+\tilde{\gamma}_{v}\right)  \right] \\
&  +\left(  1-A_{vh}\right)  \left[  \log\left(  1-F_{\varepsilon}\left(
W_{vh}\boldsymbol{c}^{\prime}+\gamma_{v}\right)  \right)  -\log\left(
1-F_{\varepsilon}\left(  W_{vh}\boldsymbol{\tilde{c}}^{\prime}+\tilde{\gamma
}_{v}\right)  \right)  \right] \\
&  \leq\left\Vert \left(  \boldsymbol{c},\gamma_{v}\right)  -\left(
\boldsymbol{\tilde{c}},\tilde{\gamma}_{v}\right)  \right\Vert \left\{
\sup_{\left(  \boldsymbol{c},\gamma_{v}\right)  \in\bar{\Upsilon}\left(
\bar{v}\right)  }\frac{f_{\varepsilon}\left(  W_{vh}\boldsymbol{c}^{\prime
}+\gamma_{v}\right)  }{F_{\varepsilon}\left(  W_{vh}\boldsymbol{c}^{\prime
}+\gamma_{v}\right)  }+\sup_{\left(  \boldsymbol{c},\gamma_{v}\right)  \in
\bar{\Upsilon}\left(  \bar{v}\right)  }\frac{f_{\varepsilon}\left(
W_{vh}\boldsymbol{c}^{\prime}+\gamma_{v}\right)  }{1-F_{\varepsilon}\left(
W_{vh}\boldsymbol{c}^{\prime}+\gamma_{v}\right)  }\right\}  .
\end{align*}
Since the parameter space $\Upsilon_{1}$ in which $\boldsymbol{c}$ lies is
compact, the support of $W_{vh}$ is bounded, and $\left\vert \gamma
_{v}\right\vert \leq\sqrt{4(\sigma_{e}^{\ast})^{2}\log[\bar{v}\left(  \log
\bar{v}\right)  ^{t}]}$, we can bound possible minimum and maximum values of
$W_{vh}\boldsymbol{c}^{\prime}+\gamma_{v}$ uniformly over $\left(  v,h\right)
$. That is, by letting
\begin{equation}
I_{\bar{v}}:=C^{0}+C^{1}\sqrt{\log[\bar{v}\left(  \log\bar{v}\right)  ^{t}]}
\label{eq: regressor-bound}%
\end{equation}
with sufficiently large constants, $C^{0},C^{1}>0$, we may suppose that
$-I_{\bar{v}}\leq W_{vh}\boldsymbol{c}^{\prime}+\gamma_{v}\leq I_{\bar{v}}$.
By the inequalities in (\ref{eq: NCDF-bound}), we can find
\begin{align*}
\sup_{\left(  \boldsymbol{c},\gamma_{v}\right)  \in\bar{\Upsilon}\left(
\bar{v}\right)  }\frac{f_{\varepsilon}\left(  W_{vh}\boldsymbol{c}^{\prime
}+\gamma_{v}\right)  }{F_{\varepsilon}\left(  W_{vh}\boldsymbol{c}^{\prime
}+\gamma_{v}\right)  }  &  \leq\sup_{-I_{\bar{v}}\leq x\leq I_{\bar{v}}}%
\frac{f_{\varepsilon}\left(  x\right)  }{F_{\varepsilon}\left(  x\right)  }\\
&  \leq\sup_{-I_{\bar{v}}\leq x<-1}\frac{\phi\left(  x\right)  }%
{\frac{\left\vert x\right\vert }{1+x^{2}}\phi\left(  x\right)  }+\sup
_{x\geq-1}\frac{f_{\varepsilon}\left(  x\right)  }{F_{\varepsilon}\left(
x\right)  }\leq\sup_{-I_{\bar{v}}\leq x<-1}\left(  1+\left\vert x\right\vert
\right)  +\frac{\phi\left(  0\right)  }{\Phi\left(  -1\right)  }%
\end{align*}
and
\begin{align*}
\sup_{\left(  \boldsymbol{c},\gamma_{v}\right)  \in \Upsilon_{\boldsymbol{c}%
}\times \Upsilon_{\gamma}^{\bar{v}}}\frac{f_{\varepsilon}\left(  W_{vh}%
\boldsymbol{c}^{\prime}+\gamma_{v}\right)  }{1-F_{\varepsilon}\left(
W_{vh}\boldsymbol{c}^{\prime}+\gamma_{v}\right)  }  &  \leq\sup_{-I_{\bar{v}%
}\leq x\leq I_{\bar{v}}}\frac{f_{\varepsilon}\left(  x\right)  }%
{1-F_{\varepsilon}\left(  x\right)  }\\
&  \leq\sup_{x\leq1}\frac{f_{\varepsilon}\left(  x\right)  }{1-F_{\varepsilon
}\left(  x\right)  }+\sup_{1<x\leq I_{\bar{v}}}\frac{\phi\left(  x\right)
}{\frac{x}{1+x^{2}}\phi\left(  x\right)  }\\
&  \leq\frac{\phi\left(  0\right)  }{1-\Phi\left(  1\right)  }+\sup_{1<x\leq
I_{\bar{v}}}\left(  1+x\right)  .
\end{align*}
Therefore, given these bounds and the definition of $I_{\bar{v}}$, we can
write
\[
\left\vert \mathcal{L}_{vh}\left(  \boldsymbol{c},\gamma_{v}\right)
-\mathcal{L}_{vh}\left(  \boldsymbol{\tilde{c}},\tilde{\gamma}_{v}\right)
\right\vert \leq\left\Vert \left(  \boldsymbol{c},\gamma_{v}\right)  -\left(
\boldsymbol{\tilde{c}},\tilde{\gamma}_{v}\right)  \right\Vert \times
C_{2}\sqrt{\log[\bar{v}\left(  \log\bar{v}\right)  ^{t}]},
\]
for any large $\bar{v}$, where $C_{2}>0$ is some constant that is independent
of $\left(  v,h\right)  $, $\bar{v}$, $\left(  \boldsymbol{c},\gamma
_{v}\right)  $, and $\left(  \boldsymbol{\tilde{c}},\tilde{\gamma}_{v}\right)
$.\medskip

\noindent\textbf{Proof of (\ref{eq: U-bound}).} Note that $F_{\varepsilon
}\left(  W_{vh}\boldsymbol{c}^{\prime}+\gamma_{v}\right)  \in\left(
0,1\right)  $ and thus $\mathcal{L}_{vh}\left(  \boldsymbol{c},\gamma
_{v}\right)  \leq0$. To find a lower bound of $\mathcal{L}_{vh}\left(
\boldsymbol{c},\gamma_{v}\right)  $, we use the following results:
\begin{equation}
\frac{x}{1+x^{2}}\phi\left(  x\right)  \leq1-F_{\varepsilon}\left(  x\right)
\text{ for }x\geq0\text{; and }\frac{\left\vert x\right\vert }{1+x^{2}}%
\phi\left(  x\right)  \leq F_{\varepsilon}\left(  x\right)  \text{ \ for }x<0,
\label{eq: NCDF-bound}%
\end{equation}
which follows p. 112 of Karatzas and Shreve (1991) and the symmetry of the
standard normal distribution (noting that $1-F_{\varepsilon}\left(  x\right)
=\int_{x}^{\infty}\phi\left(  u\right)  du$), and separately look at
$\log(1-F_{\varepsilon}\left(  W_{vh}\boldsymbol{c}^{\prime}+\gamma
_{v}\right)  )$ and $\log F_{\varepsilon}\left(  W_{vh}\boldsymbol{c}^{\prime
}+\gamma_{v}\right)  $. This inequality (\ref{eq: NCDF-bound}) and the bound
of $W_{vh}\boldsymbol{c}^{\prime}+\gamma_{v}$, $I_{\bar{v}}$ (defined
in(\ref{eq: regressor-bound})), we obtain%
\begin{align*}
\log(1-F_{\varepsilon}\left(  W_{vh}\boldsymbol{c}^{\prime}+\gamma_{v}\right)
)  &  \geq\log(1-F_{\varepsilon}\left(  I_{\bar{v}}\right)  )\geq\log\left(
\frac{I_{\bar{v}}}{1+\left(  I_{\bar{v}}^{\ast}\right)  ^{2}}\phi\left(
I_{\bar{v}}\right)  \right) \\
&  =\log I^{\ast}-\log(1+\left\vert I_{\bar{v}}\right\vert ^{2})+\log
\phi(I_{\bar{v}})
\end{align*}
and analogously,
\begin{align*}
\log F_{\varepsilon}\left(  W_{vh}\boldsymbol{c}^{\prime}+\gamma_{v}\right)
&  \geq\log F_{\varepsilon}\left(  -I_{\bar{v}}\right) \\
&  \geq\log I_{\bar{v}}-\log(1+\left\vert I_{\bar{v}}\right\vert ^{2}%
)+\log\phi(-I_{\bar{v}}).
\end{align*}
Therefore, since $\phi\left(  I_{\bar{v}}\right)  =\phi(-I_{\bar{v}})$ (by the
symmetry of $\phi$), we have%
\[
0\geq\mathcal{L}_{vh}\left(  \boldsymbol{c},\gamma_{v}\right)  \geq\log
I_{\bar{v}}-\log(1+\left\vert I_{\bar{v}}\right\vert ^{2})+\log\phi(I_{\bar
{v}}).
\]
By definition of $I_{\bar{v}}$, $\log\phi\left(  I_{\bar{v}}\right)
=-\log\sqrt{2\pi}-\frac{\left\vert C^{0}+C^{1}\sqrt{\log[\bar{v}\left(
\log\bar{v}\right)  ^{t}]}\right\vert ^{2}}{2}$ and thus we can find some
constant $C_{\mathcal{L}}\in\left(  0,\infty\right)  $ that is independent of
$\left(  v,h\right)  $ and $\bar{v}$ satisfying%
\[
|\mathcal{L}_{vh}\left(  \boldsymbol{c},\gamma_{v}\right)  |\leq
C_{\mathcal{L}}\log[\bar{v}\left(  \log\bar{v}\right)  ^{t}],
\]
which is the desired result. The proof of Lemma \ref{lem: UC-double} is complete.
\end{proof}

\subsection{Simulation Exercise\label{sec: simulation}}

To see how our two-step probit performs in finite samples, we set up the
following simulation exercise. We generate data according to the model
specified in (\ref{eq: factor}), (\ref{eq: xi-CR}) (\ref{eq: xi-CR-normal})
and (\ref{eq: normal-epsilon}) above. We vary $\bar{v}$ while holding
$N_{v}=250$ in each case, to resemble our application. We choose $\dim\left(
W_{vh}\right)  =2$, and pick for $v=1,...,\bar{v}$,%
\[
d_{v}=\left[  v\times U_{1}+0.1\times U_{2}\text{, \ }v\times0.02\times
N\left(  0,1\right)  \right]  ,
\]
where $U_{1}$ and $U_{2}$ are independent uniform [0,1], I.I.D. across
villages. Then generate $P_{vh}\simeq2\times U[0,1]+v\times0.1\times U\left[
0,1\right]  $, $e_{v}\simeq N\left(  0,0.2^{2}\right)  $, $\delta=\left(
1,1\right)  $ and $\tau_{vh},\varepsilon_{vh}\simeq N\left(  0,1\right)  $,
I.I.D. across $v$ and $h$. The multiplication by $v$ in the generation of
$d_{v}$, $P_{vh}$ lead to variation in the distribution of observables across
villages, which produces variation in $\pi_{v}$ necessary to point-identify
$\alpha$. Finally, as in (\ref{eq: factor}), (\ref{eq: xi-CR}) above, we generate%

\[
W_{vh}=d_{v}+\tau_{vh}\text{ \ and \ }\xi_{v}=d_{v}\boldsymbol{\delta}%
^{\prime}+e_{v}=\bar{W}_{v}\boldsymbol{\delta}^{\prime}+e_{v}-\bar{\tau}%
_{v}\boldsymbol{\delta}^{\prime}\text{,}
\]
and the $\pi_{v}$'s by solving the fixed point problem%
\[
\arg\min_{\pi_{v}}\left\{  \pi_{v}-\frac{1}{N_{v}}\sum_{h=}^{N_{v}}1\left(
W_{vh}^{\prime}\left(
\begin{array}
[c]{c}%
1\\
2
\end{array}
\right)  -P_{vh}+2\pi_{v}+\xi_{v}+\varepsilon_{vh}>0\right)  \right\}  ^{2}
\]

Using the $\pi_{v}$s obtained from the previous step. we generate the outcome
as%
\[
A_{vh}=1\left\{  1.75+0.5\times W_{vh}^{1}+W_{vh}^{2}-2\times P_{vh}%
+2\times\pi_{v}+\xi_{v}+\varepsilon_{vh}>0\right\}  \text{,}%
\]
and compute $\sigma_{e}$, $\alpha$ and the price coefficient by the double
probit exercise described in equations (\ref{without}) and (\ref{with}%
)/(\ref{with-approx}) above. For each choice of $\bar{v}$, we repeat the
process over 100 replications. The results of the simulation exercise are
reported in Table 6, where the true values of the parameters used to generate
the data are displayed at the top.%
%TCIMACRO{\FRAME{dtbpF}{4.1187in}{3.7343in}{0pt}{}{}{Figure}%
%maintain-aspect-ratio TRUE;  display "USEDEF";  valid_file "T";
%width 4.1187in;  height 3.7343in;  depth 0pt;  original-width 6.9237in;
%original-height 6.2708in;  cropleft "0";  croptop "1";  cropright "1";
%cropbottom "0";  tempfilename 'RM2JQ500.wmf';tempfile-properties "XPR";}}}%
%BeginExpansion
\begin{center}
\includegraphics[
natheight=6.270800in,
natwidth=6.923700in,
height=3.7343in,
width=4.1187in
]%
{RM2JQ500.wmf}%
\end{center}
%EndExpansion

Based on the root mean square error and quantile values of the estimates, it
is clear that for $\bar{v}=10$, the estimates are more precise than when we
have $\bar{v}=5,25$. As $\bar{v}$ rises from $5$ to $10$, the standard
deviation of the estimated $\alpha$ and the root-mean square error decrease
due to larger variation in $\pi_{v}$ which helps pin down $\alpha$. On the
other hand, the deterioration from $\bar{v}=10$ to $\bar{v}=25$ results from
the need to estimate many more village fixed effects $\gamma_{v}$ in
(\ref{without}). This represents the fundamental trade-off discussed in the
asymptotic results in the paper; a larger $\bar{v}$ helps average out the
$e_{v}$s, but also increases the number of nuisance parameters $\gamma_{v}$ to
be estimated in the first stage probit. Given that $\bar{v}=11$ in our
application, the good performance under $\bar{v}=10$ in our simulation
exercise is reassuring.

\end{document}